%% file: main.tex
\documentclass[11pt]{article}

\usepackage{amsthm}
\usepackage{graphicx} 
\usepackage{array} 

\usepackage{amsmath, amssymb, amsfonts, verbatim}
\usepackage{hyphenat,epsfig,multirow}

\usepackage[usenames,dvipsnames]{xcolor}
\usepackage[ruled]{algorithm2e}

\usepackage{tcolorbox}
\tcbuselibrary{skins,breakable}
\tcbset{enhanced jigsaw}

\usepackage[compact]{titlesec}

\definecolor{DarkRed}{rgb}{0.5,0.1,0.1}
\definecolor{DarkBlue}{rgb}{0.1,0.1,0.5}

\usepackage{nameref}
\definecolor{ForestGreen}{rgb}{0.1333,0.5451,0.1333}
\definecolor{Red}{rgb}{0.9,0,0}
\usepackage[linktocpage=true,
	pagebackref=true,colorlinks,
	linkcolor=DarkRed,citecolor=ForestGreen,
	bookmarks,bookmarksopen,bookmarksnumbered]
	{hyperref}

\usepackage{tikz}
\usetikzlibrary{arrows}
\usetikzlibrary{arrows.meta}
\usetikzlibrary{shapes}
\usetikzlibrary{backgrounds}
\usetikzlibrary{positioning}
\usetikzlibrary{decorations.markings}
\usetikzlibrary{patterns}
\usetikzlibrary{calc}
\usetikzlibrary{fit}
\usetikzlibrary{snakes}

\usepackage{bm}
\usepackage{url}
\usepackage{xspace}
\usepackage[mathscr]{euscript}

\usepackage{mdframed}

\usepackage[noend]{algpseudocode}
\makeatletter
\def\BState{\State\hskip-\ALG@thistlm}
\makeatother

\usepackage{cite}
\usepackage{enumitem}

\usepackage[margin=1in]{geometry}

\newtheorem{theorem}{Theorem}
\newtheorem{lemma}{Lemma}[section]
\newtheorem{proposition}[lemma]{Proposition}
\newtheorem{corollary}[lemma]{Corollary}
\newtheorem{claim}[lemma]{Claim}

\newtheorem{definition}{Definition}[section]

\newtheorem*{claim*}{Claim}
\newtheorem*{proposition*}{Proposition}
\newtheorem*{lemma*}{Lemma}
\newtheorem*{problem*}{Problem}

\theoremstyle{definition}
\newtheorem{remark}[lemma]{Remark}

\newtheorem{assumption}{Assumption}

\newtheorem{mdresult}{Result}
\newenvironment{result}{\begin{mdframed}[backgroundcolor=lightgray!40,topline=false,rightline=false,leftline=false,bottomline=false,innertopmargin=2pt]\begin{mdresult}}{\end{mdresult}\end{mdframed}}

\usepackage{tikz}
\usetikzlibrary{arrows}
\usetikzlibrary{arrows.meta}
\usetikzlibrary{shapes}
\usetikzlibrary{backgrounds}
\usetikzlibrary{positioning}
\usetikzlibrary{decorations.markings}
\usetikzlibrary{patterns}
\usetikzlibrary{calc}
\usetikzlibrary{fit}
\usetikzlibrary{snakes}
\usetikzlibrary{shadows.blur}
\usetikzlibrary{petri,decorations.markings}

\usepackage{caption}
\usepackage{subcaption}

\allowdisplaybreaks

\renewcommand{\qed}{\nobreak \ifvmode \relax \else
      \ifdim\lastskip<1.5em \hskip-\lastskip
      \hskip1.5em plus0em minus0.5em \fi \nobreak
      \vrule height0.75em width0.5em depth0.25em\fi}

\newcommand{\Qed}[1]{\ensuremath{\qed_{\textnormal{~#1}}}}

\input{macros}

\title{Palette Sparsification Beyond $(\Delta+1)$ Vertex Coloring\footnote{An extended abstract of this paper appears in RANDOM 2020.}} 
\author{Noga Alon\footnote{Department of Mathematics, Princeton University, Princeton, New Jersey, USA and Schools of Mathematics and Computer 
Science, Tel Aviv University, Tel Aviv, Israel. Research supported
in part by NSF grant DMS-1855464 and the Simons Foundation.
} 
\and Sepehr Assadi\footnote{Department of Computer Science, Rutgers University, Piscataway, New Jersey, USA.  Part of this work was done while the author was a postdoctoral researcher at Princeton University and was supported in part by Simons Collaboration on Algorithms and Geometry.}  
}

\date{}

\begin{document}
\maketitle

\pagenumbering{roman}

\input{abstract}

\clearpage

\setcounter{tocdepth}{2} 
\tableofcontents

\clearpage

\pagenumbering{arabic}
\setcounter{page}{1}

\input{intro}

\input{prelim}

\input{palette-sparsification-global}

\input{palette-sparsification-local}


\input{sublinear}
\input{vertex-sampling}

\subsection*{Acknowledgements}

Sepehr Assadi would like to thank Suman Bera, Amit Chakrabarti, Prantar Ghosh, Guru Guruganesh, David Harris, Sanjeev Khanna, and Hsin-Hao Su for helpful conversations and Mohsen Ghaffari for communicating the $(\deg+1)$ coloring problem  
and an illuminating discussion that led us to the proof of the palette sparsification theorem for this problem in this paper. We are also thankful to the anonymous reviewers of RANDOM 2020 for helpful suggestions on the presentation of the paper, 
and to Stijn Cambie and Ross Kang for helpful comments.

\bibliographystyle{abbrv}
\bibliography{general}

\clearpage
\appendix

\input{proof-triangle}

\input{background-ACK19}

\input{omitted-proofs}

\end{document}

%% file: macros.tex
\newcommand{\Ot}{\ensuremath{\widetilde{O}}}
\newcommand{\eps}{\ensuremath{\varepsilon}}
\newcommand{\Paren}[1]{\Big(#1\Big)}

\newcommand{\bracket}[1]{\left[#1\right]}
\newcommand{\paren}[1]{\ensuremath{\left(#1\right)}\xspace}
\newcommand{\card}[1]{\left\vert{#1}\right\vert}

\newcommand{\IN}{\ensuremath{\mathbb{N}}}

\newcommand{\ceil}[1]{{\left\lceil{#1}\right\rceil}}

\newcommand{\expect}[1]{\Exp\bracket{#1}}

\newcommand{\set}[1]{\ensuremath{\left\{ #1 \right\}}}
\newcommand{\poly}{\mbox{\rm poly}}
\newcommand{\polylog}{\mbox{\rm  polylog}}

\DeclareMathOperator*{\Exp}{\ensuremath{{\mathbb{E}}}}
\DeclareMathOperator*{\Prob}{\ensuremath{\mathbb{P}}}
\renewcommand{\Pr}{\Prob}
\newcommand{\EX}{\Exp}

\newcommand{\FG}{\ensuremath{\mathcal{G}}}

\newcommand{\event}{\ensuremath{{\mathcal{E}}}}

\newenvironment{tbox}{\begin{tcolorbox}[
		enlarge top by=5pt,
		enlarge bottom by=5pt,
		 breakable,
		 boxsep=0pt,
                  left=4pt,
                  right=4pt,
                  top=10pt,
                  arc=0pt,
                  boxrule=1pt,toprule=1pt,
                  colback=white
                  ]
	}
{\end{tcolorbox}}

\newcommand{\col}{\ensuremath{\mathcal{C}}}
\newcommand{\drem}[1]{\ensuremath{d^{\textnormal{\textsf{rem}}}(#1)}}

\newcommand{\pactive}{\ensuremath{p_{\tt active}}}
\newcommand{\istar}{\ensuremath{i^{\star}}}

%% file: abstract.tex
\begin{abstract}

	A recent \emph{palette sparsification theorem} of Assadi, Chen, 
	and Khanna [SODA'19] states that in every $n$-vertex graph $G$ with maximum degree $\Delta$, sampling $O(\log{n})$ colors per each vertex independently from $\Delta+1$ colors 
	almost certainly allows for proper coloring of $G$
	from the sampled colors. Besides being a combinatorial statement of its own independent interest, this theorem was shown to have various applications to design of algorithms for $(\Delta+1)$ coloring in different 
	models of computation on massive graphs such as
	 streaming or sublinear-time algorithms. 
	
	\smallskip
	
	In this paper, we focus on palette sparsification beyond 
$(\Delta+1)$ coloring, in both regimes when the number of available colors is much larger than $(\Delta+1)$, and when it is much smaller. In particular, 
	\begin{itemize}
		\item We prove that for $(1+\eps) \Delta$ coloring, sampling only $O_{\eps}(\sqrt{\log{n}})$ colors per vertex is sufficient and necessary to obtain a proper coloring from the sampled colors -- this shows a  separation between
		 $(1+\eps) \Delta$ and $(\Delta+1)$ coloring in the context of palette sparsification.

		\item A natural family of graphs with chromatic number much smaller than $(\Delta+1)$ are triangle-free graphs which are $O(\frac{\Delta}{\ln{\Delta}})$ colorable. 
		We prove a palette sparsification theorem tailored to these graphs: 
		Sampling $O(\Delta^{\gamma} + \sqrt{\log{n}})$ colors per vertex is sufficient and necessary to obtain a proper $O_{\gamma}(\frac{\Delta}{\ln{\Delta}})$ coloring of triangle-free graphs. 	
		
		\item We also consider the ``local version'' of graph coloring where every vertex $v$ can only be colored from a list of colors with size proportional to the degree $\deg(v)$ of $v$. We show
		that sampling $O_{\eps}(\log{n})$ colors per vertex is 
		sufficient for proper coloring of any graph with high probability 
		whenever each vertex is sampling from 
		a list of $(1+\eps) \cdot \deg(v)$ arbitrary colors, 
		or even only $\deg(v)+1$ colors when the lists are the sets $\set{1,\ldots,\deg(v)+1}$. 
	\end{itemize} 
	
	\smallskip

	Similar to previous work, our new palette sparsification results naturally lead to a host of new and/or improved algorithms for vertex coloring in different models including streaming and sublinear-time algorithms.

\end{abstract}

%% file: intro.tex

\section{Introduction}\label{sec:intro}

Given a graph $G(V,E)$, let $n := \card{V}$ be the number of vertices and $\Delta$ denote the maximum degree. 
A proper $c$-coloring of $G$ is an assignment of colors to vertices from the palette of colors $\set{1,\ldots,c}$ such that adjacent vertices receive distinct colors.  
The minimum number of colors needed for proper coloring of $G$ is referred to as the chromatic number of $G$ and is denoted by $\chi(G)$. 
An interesting variant of graph coloring is \emph{list-coloring} whereby every vertex $v$ is given a set $S(v)$ of available colors and the goal is to find a proper coloring of $G$ such that the color of every  $v$ belongs to $S(v)$. 
When this is possible, we say that $G$ is list-colorable from the lists $S$. 

It is well-known that $\chi(G) \leq \Delta+1$ for every graph $G$; the algorithmic problem of finding such a coloring---the $(\Delta+1)$ coloring problem---can also be solved via a text-book greedy algorithm. 
Very recently, Assadi, Chen, and Khanna~\cite{AssadiCK19} proved the following \textbf{palette sparsification theorem} for the $(\Delta+1)$ coloring problem: Suppose for every vertex $v$ of a graph $G$, 
we \emph{independently} sample $O(\log{n})$ colors $L(v)$ uniformly at random from the palette $\set{1,\ldots,\Delta+1}$; then $G$ is almost-certainly list-colorable from the sampled lists $L$ (see Appendix~\ref{sec:background-ACK19} for 
a formal statement).

The palette sparsification theorem of~\cite{AssadiCK19}, besides being a purely graph-theoretic result of its own independent interest, also had several interesting algorithmic implications 
for the $(\Delta+1)$ coloring problem owing to its ``sparsification'' nature: it is easy to see that by sampling only $O(\log{n})$ colors per vertex, the total number of edges that can ever become monochromatic while coloring $G$ from lists $L$
is with high probability only $O(n \cdot \log^2{n})$; at the same time we can safely ignore all other edges of $G$. This theorem thus reduces 
the $(\Delta+1)$ coloring problem, in a \emph{non-adaptive} way, to a list-coloring problem on a graph
with (potentially) much smaller number of edges. 

The aforementioned aspect of this palette sparsification is particularly appealing for the design of \emph{sublinear algorithms}---these 
are algorithms which require computational resources that are substantially smaller than the size of their input. Indeed, one of the interesting applications of this theorem, proven (among other things) in~\cite{AssadiCK19},
is a randomized algorithm for the $(\Delta+1)$ coloring problem that runs in $\Ot(n\sqrt{n})$\footnote{Here and throughout the paper, we use the notation $\Ot(f) := O(f \cdot \polylog(f))$ to suppress log-factors.} time; for sufficiently dense graphs, this is faster than even 
reading the entire input once! 

Palette sparsification in~\cite{AssadiCK19} was tailored specifically to the $(\Delta+1)$ coloring problem. 
Motivated by the ubiquity of graph coloring problems on one hand, and the wide range of applications of this palette sparsification result on the other hand, the following question is natural: 
\begin{quote}
\emph{What other graph coloring problems admit (similar) palette sparsification theorems?} 
\end{quote}
This is precisely the question we study in this work from both upper and lower bound fronts. 

\subsection{Our Contributions}\label{sec:results}

We consider palette sparsification beyond $(\Delta+1)$ coloring: when the number of available colors is much larger than $\Delta+1$, when it is much smaller, and when the number of available colors for vertices depend on 
``local'' parameters of the graph. We elaborate on each part below.  

\paragraph{$\bm{(1+\eps)\Delta}$ Coloring.} The palette sparsification theorem of~\cite{AssadiCK19} is shown to be tight in the sense that on some 
graphs, sampling $o(\log{n})$ colors 
per vertex from $\set{1,\ldots,\Delta+1}$, results in the sampled list-coloring instance to have no proper coloring with high probability. 
We prove that in contrast to this, if one allows for a larger number of available colors, then indeed we can obtain a palette sparsification with asymptotically smaller sampled lists.  
\begin{result}[Informal -- Formalized in Theorem~\ref{thm:ps-od-coloring}]\label{res:od}
	For any graph $G(V,E)$, sampling $O_{\eps}(\sqrt{\log{n}})$ colors per vertex from a set of size $(1+\eps)\Delta$ colors with high probability allows for a proper list-coloring of $G$ from the sampled lists. 
\end{result}
Result~\ref{res:od}, combined with the lower bound of~\cite{AssadiCK19}, provides a separation between $(\Delta+1)$ coloring and $(1+\eps)\Delta$ coloring in the context of palette sparsification. 
We also prove that the bound of $\Theta(\sqrt{\log{n}})$ sampled colors is (asymptotically) optimal in Result~\ref{res:od}. 

To prove Result~\ref{res:od}, we unveil a new connection between palette sparsification theorems and some of the classical list-coloring problems studied in the literature. In particular, 
several works in the past (see, e.g.~\cite{Reed99,Haxell01,ReedS02} and \cite[Proposition~5.5.3]{ProbBook}) have studied the following question: Suppose in a list-coloring instance on a graph $G$, we define the $c$-degree of a vertex-color pair $(v,c)$ as the number of neighbors of $v$ that also contain $c$ in their list; what conditions
on maximum $c$-degrees and minimum list sizes imply that $G$ is list-colorable from such lists? 

Palette sparsification theorems turned out to be closely related to these questions as the sampled lists in these results can be viewed through the lens of these list-coloring results. In particular, Reed and Sudakov~\cite{ReedS02} proved
that in the above question if the size of each list is larger than
the maximum $c$-degree by a $(1+o(1))$ factor, then $G$ is always
list-colorable. The question here is then whether or not
the lists sampled in Result~\ref{res:od} satisfy this condition
with high probability. The answer turns out to be \emph{no} as
sampling only $O(\sqrt{\log{n}})$ colors does not provide
the proper concentration needed for this guarantee. Despite this,
we show that one can still use~\cite{ReedS02} to prove
Result~\ref{res:od}  with a more delicate argument 
by applying~\cite{ReedS02} to carefully 
chosen subsets of the sampled lists. 

\paragraph{$\bm{O(\frac{\Delta}{\ln{\Delta}})}$ Coloring of Triangle-Free Graphs.} Even though $\chi(G)$ in general can be $\Delta+1$, many natural families of graphs have chromatic number (much) smaller than $\Delta+1$. 
One key example is the set of triangle-free graphs which are $O(\frac{\Delta}{\ln{\Delta}})$ colorable by a celebrated result of Johansson~\cite{Johansson96a} (this result was recently simplified and 
improved to $(1+o(1)) \cdot \frac{\Delta}{\ln{\Delta}}$ by Molloy~\cite{Molloy19}; see also~\cite{PettieS15,Bernshteyn19}). We prove a palette sparsification theorem tailored to these graphs. 

\begin{result}[Informal -- Formalized in Theorem~\ref{thm:ps-triangle-free}]\label{res:triangle-free}
	For any triangle-free graph $G(V,E)$, sampling $O(\Delta^{\gamma}+\sqrt{\log{n}})$ colors per vertex from a set of size $O_{\gamma}(\frac{\Delta}{\ln{\Delta}})$ colors with high probability allows for a proper list-coloring of $G$ from the sampled lists. 
\end{result}
 	Unlike Result~\ref{res:od} of our paper and the theorem of~\cite{AssadiCK19}, in this result we also have a dependence of $\Delta^{\gamma}$ on the number of sampled colors 
	(where the exponent depends on the number of available colors). We prove that this dependence is also necessary in this result (Proposition~\ref{prop:tf-optimal}). 

	The proof of Result~\ref{res:triangle-free} is also based on the aforementioned connection to list-coloring problems based on $c$-degrees. However, unlike the case for Result~\ref{res:od}, here we are not aware 
	of any such list-coloring result that allows us to infer Result~\ref{res:triangle-free}. As such, a  key part of the proof of Result~\ref{res:triangle-free} is exactly to establish such a result. Our proof for the corresponding 
	list-coloring problem is by the probabilistic method and in particular a version of the so-called ``R\"{o}dl Nibble'' or the ``semi-random method''; see, e.g.~\cite{Rodl85,ColoringBook}. 
	Similar to previous work on coloring triangle-free graphs, the main challenge here is to establish the desired concentration bounds. We do this following the approach of Pettie and Su~\cite{PettieS15} in their  
	distributed algorithm for coloring triangle-free graphs. 

We shall note that our proofs of Results~\ref{res:od} and~\ref{res:triangle-free} are almost entirely disjoint from the techniques in~\cite{AssadiCK19} and instead build on classical
work on list-coloring problems in the graph theory literature. 

\paragraph{Coloring with Local Lists Size.} Finally, we consider a coloring problem with ``local'' list sizes where the number of available colors for vertices 
depends on a local parameter, namely their degree as
opposed to a global parameter such as maximum degree. 

\begin{result}[Informal -- Formalized in Theorem~\ref{thm:ps-deg+1-coloring}]\label{res:deg+1}
	For any graph $G(V,E)$, sampling $O_{\eps}(\log{n})$ colors for each vertex $v$ with degree $\deg(v)$ from a set $S(v)$ of $(1+\eps) \cdot \deg(v)$ arbitrary colors or only $\deg(v)+1$ colors when the lists are the sets $\set{1,\ldots,\deg(v)+1}$, 
	allows for a proper coloring of $G$ from the sampled colors. 
\end{result}
	
	Coloring problems with local lists size have been studied
before in both the graph theory literature, e.g. in~\cite{DaviesVKP18,BonamyKNP2018} for coloring triangle-free graphs (and as pointed out by~\cite{DaviesVKP18}, the general idea
	goes all the way back to the notion of degree-choosability in one of the original list-coloring papers~\cite{ErdosRT79}), and theoretical computer science, e.g. in~\cite{ChangLP18}. 
	
	To be more precise, the first part of Result~\ref{res:deg+1} refers to the standard $(1+\eps)\deg$ \emph{list}-coloring problem and the second part corresponds to the so-called $(\deg+1)$ coloring problem
	introduced first (to our knowledge) in the recent work of Chang, Li, and Pettie~\cite{ChangLP18} (see also~\cite{AmirKKNP16} for an application of this problem). We remark that the $(\deg+1)$ coloring 
	problem is a generalization of the $(\Delta+1)$ coloring problem and hence our Result~\ref{res:deg+1} generalizes that of~\cite{AssadiCK19} (although technically we build on many of the ideas and tools
	developed in~\cite{AssadiCK19} for $\Delta+1$ coloring).   
	
	Our proof of Result~\ref{res:deg+1} takes a different route than  Results~\ref{res:od} and~\ref{res:triangle-free} that were based on list-coloring and instead we follow the approach of~\cite{AssadiCK19}  for the $(\Delta+1)$ coloring problem 
	(outlined in Appendix~\ref{sec:background-ACK19}). A fundamental challenge here is that the graph decomposition for partitioning vertices into sparse and dense parts that played
	a key role in~\cite{AssadiCK19} is no longer applicable to the $(\deg+1)$ coloring problem. We address this by ``relaxing'' the requirements of the decomposition 
and develop a new one that despite being somewhat ``weaker'' than 
	the ones for $(\Delta+1)$ coloring in~\cite{HarrisSS16,ChangLP18,AssadiCK19} (themselves  based on~\cite{Reed98}), takes into account the disparity between degrees of vertices in the $(\deg+1)$ coloring problem. 
	Similar to~\cite{AssadiCK19}, we then handle ``sparse''\footnote{Technically speaking, this decomposition allows for vertices that are neither sparse nor dense according to standard definitions and are key to extending the decomposition from
	$(\Delta+1)$ coloring to $(\deg+1)$ coloring.} and dense vertices of this decomposition separately but unlike~\cite{AssadiCK19}, here the main part of the argument is to handle these ``sparse'' vertices and the result
	for the dense part follows more or less directly from~\cite{AssadiCK19}. 
	
	We conclude this section by noting that our proof for $(1+\eps)\deg$-list coloring problem also immediately gives a palette sparsification result for obtaining a $(1+\eps)\kappa$-list coloring where $\kappa$ is the degeneracy of the graph
	(see Remark~\ref{rem:degeneracy-coloring}). This problem was studied very recently in the context of sublinear or ``space conscious'' algorithms by Bera, Chakrabarti, and Ghosh~\cite{BeraCG19} who also 
	proved, among many other interesting results, a lower bound
	that $(\kappa+1)$ coloring cannot be achieved via palette sparsification (see~\cite[Section 5.3]{BeraCG19} -- our result thus complements their lower bound. 
\subsection{Implication to Sublinear Algorithms for Graph Coloring} 

As stated earlier, one motivation in studying palette sparsification is in its application to design of sublinear algorithms. 
As was shown in~\cite{AssadiCK19}, these theorems imply sublinear algorithms in various models in ``almost'' a black-box way (see Section~\ref{sec:sublinear} for details). For concreteness, in this paper, we stick to their application
to the two canonical examples of streaming and sublinear-time algorithms. We only note in passing that exactly as in~\cite{AssadiCK19}, our results also imply new algorithms in models 
such as massively parallel computation (MPC) or distributed/linear sketching; see also~\cite{ChangFGUZ18,BeraCG19} for more recent results on graph coloring problems in these and related models. 

Our results in this part appear in Section~\ref{sec:sublinear}. Table~\ref{tab:sublinear} presents a summary of our sublinear algorithms and the directly related previous work (although our Result~\ref{res:od} implies
a separation between $(\Delta+1)$ and $(1+\eps)\Delta$ coloring for palette sparsification, the resulting sublinear algorithms from Result~\ref{res:od} are subsumed by the previous work in~\cite{BeraCG19} and hence are omitted from Table~\ref{tab:sublinear}). 

\paragraph{Sublinear Algorithms from Graph Partitioning.} Motivated by our results on sublinear algorithms for triangle-free graphs, we also consider sublinear algorithms for coloring other ``locally sparse'' graphs such as $K_r$-free graphs, locally $r$-colorable 
graphs, and graphs with sparse neighborhood. We give several results for these problems through a general algorithm based on the graph partitioning technique (see, e.g.~\cite{ChangFGUZ18, Parter18,ParterS18,BeraCG19}). Our results in this 
part are presented in Section~\ref{sec:vertex-sampling}. 

\input{tab-results}


%% file: tab-results.tex
 \def\arraystretch{2}

\newsavebox{\tabone}

\sbox{\tabone}{
             {\small
             
        \centering
      
        \begin{tabular}{|c||c|c|c|c|}
             \hline
     
        \textbf{Problem} & \textbf{Graph Family} &  \textbf{Streaming} & \textbf{Sublinear-Time} & \textbf{Source} \\
             \hline 
             \hline
             
	     {$(\Delta + 1)$ Coloring}  & 	General 	&  $O(n\log^2{n})$ space  & $\Ot(n^{3/2}) $ time & \cite{AssadiCK19} \\	
   	     {$(1+\eps)\kappa $ Coloring} & $\kappa$-Degenerate & $O(n\log{n})$ space  & $\Ot(n^{3/2}) $ time & \cite{BeraCG19} \\	
	     {$O_{\gamma}(\frac{\Delta}{\ln{\Delta}})$ Coloring}  & Triangle-Free & $O(n \cdot \Delta^{\gamma})$ space  & $O(n^{3/2+\gamma}) $ time & \underline{our work} \\
	     {$(1+\eps)\deg$ List-Coloring}  & General & $O(n \cdot \log^2{n})$ space  & $\Ot(n^{3/2}) $ time & \underline{our work} \\
	    {$(\deg+1)$ Coloring}  & General & $O(n \cdot \log^2{n})$ space  & $\Ot(n^{3/2}) $ time & \underline{our work} \\
	   \hline
        \end{tabular}
      }
  }
  
 \begin{table}[t!]
\begin{tikzpicture}
   \node[fill=white](boz){};
  \node[drop shadow={black, shadow xshift=5pt,shadow yshift=-5pt, opacity=0.5}, fill=white, inner xsep=-7pt, inner ysep=0pt](table)[right=5pt of boz]{\usebox{\tabone}};
\end{tikzpicture}
\vspace{0.25cm}
          \caption{A sample of our sublinear algorithms as corollaries of Results~\ref{res:od},~\ref{res:triangle-free}, and~\ref{res:deg+1}, and the previous work in~\cite{AssadiCK19} and~\cite{BeraCG19} (for brevity, we assume
          $\eps,\gamma$ are constants). All streaming algorithms here are \emph{single-pass} and all sublinear-time algorithms except for $(1+\eps)\kappa$ coloring are \emph{non-adaptive}. 
        \label{tab:sublinear}}

    \end{table}

%% file: prelim.tex

\newcommand{\bd}{\ensuremath{\bar{d}}}
\newcommand{\bC}{\ensuremath{\overline{C}}}
\newcommand{\hC}{\ensuremath{\widehat{C}}}

\newcommand{\Vr}{\ensuremath{V^{\textnormal{\textsf{rem}}}}}
\newcommand{\Gr}{\ensuremath{G^{\textnormal{\textsf{rem}}}}}

\renewcommand{\drem}[1]{\ensuremath{d^{\textnormal{\textsf{rem}}}(#1)}}

\newcommand{\GC}{\ensuremath{\textnormal{\textsf{GreedyColor}}}\xspace}

\renewcommand{\col}{\ensuremath{\mathcal{C}}}

\newcommand{\Vdense}{\ensuremath{V^{\textnormal{\textsf{dense}}}}\xspace}

\newcommand{\Vsparse}{\ensuremath{V^{\textnormal{\textsf{sparse}}}}\xspace}

\newcommand{\barS}{\overline{S}}

\newcommand{\mumin}{\mu_{\textnormal{\textsf{min}}}}
\newcommand{\mumax}{\mu_{\textnormal{\textsf{max}}}}

\section{Preliminaries}\label{sec:prelim}

\paragraph{Notation.} For any integer $t \geq 1$, we define $[t] := \set{1,\ldots,t}$. For a graph $G(V,E)$, we use $V(G) := V$ and $E(G) := E$ to denote the vertex-set and edge-set respectively. 
For a vertex $v \in V$, $N_G(v)$ denotes the neighborhood of $v$ in $G$ and $\deg_G(v) := \card{N_G(v)}$ denotes the degree of $v$ (when clear from the context, we may drop the subscript $G$).  For a vertex-set $U \subseteq V$, $G[U]$ denotes the induced 
subgraph of $G$ on $U$. 

When there are lists of colors $S(v)$ given to vertices $v$, we use the term \textbf{$\bm{c}$-degree} of $v$ to mean the number of neighbors $u$ of $v$ of with color $c$ in their list $S(u)$
and denote this by $\deg_S(v,c)$.   

Throughout, we use the term ``with high probability''  (w.h.p.) for an event to mean that the probability of this event happening is at least $1-1/n^{c}$ where $c$ is a sufficiently large constant. 

\subsection{Probabilistic Tools}\label{sec:prob}

We use the following standard probabilistic tools. 

\begin{proposition}[Lov\'asz Local Lemma -- symmetric form; 
cf.~\cite{ProbBook}]\label{prop:lll}
	 Let $\event_1,\ldots,\event_n$ be $n$ events such that 
each event $\event_i$ is mutually independent of all other events 
besides at most $d$, and $\Pr\paren{\event_i} \leq p$ for all $i \in [n]$. If $e \cdot p \cdot (d+1) \leq 1$ (where $e=2.71...$), then $\Pr\paren{\wedge_{i=1}^{n} \overline{\event_i}} > 0$. 
\end{proposition}

\begin{proposition}[Chernoff-Hoeffding bound; cf.~\cite{ProbBook,MitzenmacherU17}]\label{prop:chernoff}
	Let $X_1,\ldots,X_n$ be $n$ independent random variables where each $X_i \in [0,b]$. Define $X:= \sum_{i=1}^{n} X_i$. Then, for any $t > 0$, 
	\begin{align*}
		\Pr\Paren{\card{X - \expect{X}} > t} \leq 2 \cdot \exp\paren{-\frac{2t^2}{n \cdot b^2}}.
	\end{align*}
	Moreover, for any $\delta \in (0,1)$, and $\mumin \leq \expect{X} \leq \mumax$: 
	\begin{align*}
		\Pr\paren{X > (1+\delta) \cdot \mumax} \leq \exp\paren{-\frac{\delta^2 \cdot\mumax}{3b}}, \qquad		\Pr\paren{X < (1-\delta) \cdot\mumin} \leq \exp\paren{-\frac{\delta^2 \cdot\mumin}{2b}}.
	\end{align*} 
\end{proposition}

A function $f(x_1,\ldots,x_n)$ is called \emph{$c$-Lipschitz} iff changing any single $x_i$ can affect the value of $f$ by at most $c$. Additionally, $f$ is called \emph{$r$-certifiable} iff whenever $f(x_1,\ldots,x_n) \geq s$, there exists
at most $r \cdot s$ variables $x_{i_1},\ldots,x_{i_{r \cdot s}}$ so that knowing the values of these variables certifies $f \geq s$.

\begin{proposition}[Talagrand's inequality; cf.~\cite{ColoringBook}]\label{prop:talagrand}
	Let $X_1,\ldots,X_n$ be $n$ independent random variables and $f(X_1,\ldots,X_n)$ be a $c$-Lipschitz function; then for any $t \geq 1$, 
	\begin{align*}
		\Pr\paren{\card{f - \expect{f}} > t } \leq 2 \exp\paren{-\frac{t^2}{2c^2 \cdot n}}.
	\end{align*} 
	 Moreover, if $f$ is additionally $r$-certifiable, then for any $b \geq 1$,
	\begin{align*}
		\Pr\paren{\card{f - \expect{f}} > b + 30c \sqrt{r \cdot \expect{f}}} \leq 4 \exp\paren{-\frac{b^2}{8c^2 r \expect{f}}}. 
	\end{align*} 
\end{proposition}

\subsection{List-Coloring with Constraints on Color-Degrees} We use the following result of Reed and Sudakov~\cite{ReedS02} on list-coloring of graphs with constraints on $c$-degrees of vertices. 

\begin{proposition}[\!\!\cite{ReedS02}]\label{prop:lc-eps}
   For every  $\eps > 0$ there exists a $d_0 := d_0(\eps)$ such that for all $d \geq d_0$ the following is true. 
   Suppose $G(V,E)$ is a graph with lists $S(v)$ for every $v \in V$ such that:
    \begin{enumerate}[label=(\roman*)]
    	\item for every vertex $v$, $\card{S(v)} \geq (1+\eps) \cdot d$, and
	\item for every vertex $v$ and color $c \in S(v)$, $\deg_S(v,c) \leq d$ (recall that $\deg_S(v,c)$ denotes the $c$-degree of $v$ which is the number of neighbors $u$ of $v$ with color $c \in S(u)$). 
    \end{enumerate} 
    Then, there exists a proper coloring of $G$ from these lists. 
\end{proposition}

A weaker version of this result obtained by replacing $(1+\eps)$ above with some absolute constant appeared earlier in~\cite{Reed99} (see also~\cite[Proposition~5.5.3]{ProbBook} and~\cite{Haxell01}). 
For some of our proofs, we only require this weaker version whose easy proof is provided below for completeness.  

\begin{proposition}[cf.~\cite{Reed99}]\label{prop:lc-const}
Suppose $G(V,E)$ is a graph with lists $S(v)$ for every $v \in V$ such that $\card{S(v)} \geq \ceil{2ed}$ (where $e=2.71...$) and for every color $c \in S(v)$, $c$-degree of $v$ is at most $d$. Then, there exists a proper
coloring of $G$ from these lists. 
\end{proposition}
\begin{proof}
	Pick a color for each vertex $v$ independently and uniformly 
at random from $S(v)$. For an edge $e = (u,v) \in E$ and each color
$c$ that appears in $S(u) \cap S(v)$, 
define an event $\event_{e,c}$ as the event that both endpoints $u$ and $v$ of $e$ have chosen $c$ as their color. 
	Clearly, $\Pr\paren{\event_{e,c}} \leq 1/(2ed)^2$.
On the other hand, each $\event_{e,c}$ is mutually independent of
all other events
$\event_{e',c'}$ besides those where 
$e$ and $e'$ share a vertex and $c'$ is contained in both end-points 
of $e'$. The total number of such events 
	is at most $2d (2ed)-1$.  The proof now follows from Lov\'asz Local Lemma (Proposition~\ref{prop:lll}) as there is an assignment 
of colors to vertices in which none of the events $\event_{e,c}$ happens.
\end{proof}

%% file: palette-sparsification-global.tex

\section{Two New Palette Sparsification Theorems}\label{sec:ps-global}

We present our new palette sparsification theorems in Result~\ref{res:od} and Result~\ref{res:triangle-free} in this section. 


\input{od-coloring}

\input{triangle-free}

%% file: od-coloring.tex

\newcommand{\bad}{\ensuremath{\textnormal{\textsf{bad}}}}
\newcommand{\hL}{\ensuremath{\widehat{L}}}

\subsection{Palette Sparsification for $(1+\eps)\Delta$ Coloring}\label{sec:od-coloring}

We start with our improved palette sparsification theorem for $(1+\eps)\Delta$ coloring.

\begin{theorem}\label{thm:ps-od-coloring}
    For every $\eps \in (0,1/2)$, there exists an integer $n_0(\eps) \geq 1$ such that the following is true. 
    Let $G(V,E)$ be any graph with $n \geq n_0(\eps)$ vertices and maximum degree $\Delta$, and define $C := C(\eps) = {(1+\eps)\cdot \Delta}$. Suppose for every vertex $v \in V$,
    we independently sample a set $L(v)$ of colors of size $O\paren{\sqrt{\log{n}}/\eps^{1.5}}$ uniformly at random  from colors $\set{1,\ldots,C}$. Then, 
    with high probability, there exists a proper coloring of $G$ from lists $L(v)$ for every $v \in V$. 
\end{theorem}

We shall note that in contrast to Theorem~\ref{thm:ps-od-coloring}, it was shown in~\cite{AssadiCK19} that for the more stringent problem of $(\Delta+1)$ coloring, sampling $\Omega(\log{n})$ colors per vertex is necessary. 
As such, Theorem~\ref{thm:ps-od-coloring} presents a separation between these two problems in the context of palette sparsification. 

\subsubsection*{Proof of Theorem~\ref{thm:ps-od-coloring}}
	The proof of this theorem is by showing that the lists sampled for vertices can be adjusted so that they satisfy the requirement of Proposition~\ref{prop:lc-eps}; we then apply this proposition to 
	obtain a list-coloring of $G$ from the sampled lists. Let $\ell := \paren{20\sqrt{\log{n}}/\eps^{1.5}}$ denote the number of sampled colors per vertex. 
	
	Recall that $\deg_L(v,c)$ denotes the $c$-degree of vertex $v$ with respect to lists $L$. For every $c \in L(v)$,
	\begin{align}
		\expect{\deg_L(v,c)} := \sum_{u \in N(v)} \Pr\paren{\textnormal{$u$ samples $c$ in $L(u)$}} \leq \Delta \cdot \frac{\ell}{C} = \frac{\ell}{1+\eps}. \label{eq:od1}
	\end{align}
	
	Now \emph{if} $\deg_L(v,c)$ was concentrated enough so that $\max_{v,c} \deg_L(v,c) = (1-\Theta(\eps)) \cdot \ell$, we would have been done already: by Proposition~\ref{prop:lc-eps}, 
	there is always a proper coloring of $G$ from such lists (take the parameter $d$ to be $\max_{v,c}\deg_L(v,c)$ and so size of each list is $(1+\Theta(\eps))d$). 
	Unfortunately however, it is easy to see that as $\ell = \Theta(\sqrt{\log{n}})$ in general no such  concentration is guaranteed. 
	
	We fix the issue above by showing existence of a \emph{subset} $\hL(v)$ of each list $L(v)$ such that these new lists can indeed be used in Proposition~\ref{prop:lc-eps}. The argument is intuitively as follows: the probability that 
	$\deg_L(v,c)$ deviates significantly from its expectation 
is $2^{-\Theta(\ell)} = 2^{-\Theta(\sqrt{\log{n}})}$ by a simple Chernoff bound. Moreover, the probability that $\Omega(\sqrt{\log{n}})$ colors in $L(v)$ all deviate from their expectation 
	can be bounded by $ \paren{2^{-\Theta(\sqrt{\log{n}})}}^{\Omega(\sqrt{\log{n}})}$ (ignoring dependency issues for the moment). This probability is now $n^{-\Theta(1)}$, enough for us to take a union bound over all vertices. As such, 
	by removing some fraction of the 
colors from the list of each vertex, we can  satisfy the $c$-degree requirements for applying Proposition~\ref{prop:lc-eps} and conclude the proof. We now formalize this.

	We say that a color $c \in L(v)$ is \emph{bad} for $v$ iff $\deg_L(v,c) > (1+\eps/2) \cdot \frac{\ell}{1+\eps}$.  As the choice of color $c$ for each vertex $u \in N(v)$ is independent, 
	by  Eq~\eqref{eq:od1} and Chernoff bound (Proposition~\ref{prop:chernoff}), 
	\begin{align}
		\Pr\Paren{\deg_L(v,c) > (1+\eps/2) \cdot \frac{\ell}{1+\eps}} \leq \exp\paren{-\frac{\eps^2}{12} \cdot \frac{\ell}{1+\eps}}.  \label{eq:od2}
	\end{align}

Define $\bad(v)$ as the number of colors $c$ in $L(v)$ that are bad for vertex $v$. We note that by the sampling process in Theorem~\ref{thm:ps-od-coloring}, conditioning on some colors being bad for $v$ can only reduce the chance of the remaining colors being bad for $v$. As such, by Eq~\eqref{eq:od2},  
	\begin{align*}
		\Pr\Paren{\bad(v) \geq \eps/4 \cdot \ell} &\leq {{\ell}\choose{\eps/4 \cdot \ell}} \cdot \exp\paren{-\frac{\eps^2}{12} \cdot \frac{\ell}{1+\eps}}^{\eps \cdot \ell/4} &\leq 2^{\ell} \cdot \exp\paren{-\frac{\eps^3}{72} \cdot \ell^2} \leq \exp\paren{-5\log{n}} \tag{by the choice of 
		$\ell = 20\sqrt{\log{n}}/\eps^{1.5}$ and as $\eps < 1/2$ is sufficiently smaller than $n$}.
	\end{align*}
	
	By a union bound over all $n$ vertices, with high probability, for every vertex $v$, $\bad(v) \leq \eps\cdot\ell/4$. We let $\hL(v)$ to be a subset of $L(v)$ obtained by removing 
	all bad colors from $L(v)$. For any $c \in \hL(v)$: 
	\begin{align*}
	\deg_{\hL}(v,c) \leq \deg_{L}(v,c) \leq (1+\eps/2) \cdot \frac{\ell}{1+\eps} \leq (1-\eps/3) \cdot \ell. \tag{for $\eps < 1/2$}
	\end{align*}
	On the other hand, as $\bad(v) \leq \eps\cdot\ell/4$, we have $\card{\hL(v)} \geq (1-\eps/4) \cdot \ell$. As such, by Proposition~\ref{prop:lc-eps} (as $\eps$ is a constant with respect to $\ell$), 
	we can list-color $G$ from lists $\hL$ and consequently also $L$, finalizing the proof. \Qed{Theorem~\ref{thm:ps-od-coloring}}

\subsubsection*{Asymptotic Optimality of the Bounds in Theorem~\ref{thm:ps-od-coloring}}\label{sec:od-lower}

We give a simple proof of the (asymptotic) optimality of $O(\sqrt{\log{n}})$ sampled colors in Theorem~\ref{thm:ps-od-coloring}. That is, if we instead sample slightly smaller number of colors per each vertex, then
there are graphs where, w.h.p., 
the resulting list-coloring instance has no proper coloring. For concreteness, we focus on $2\Delta$ coloring; it will be evident how to extend this to other choices of $O(\Delta)$ coloring. 

\begin{proposition}\label{prop:od-optimal}
  There exists an $n$-vertex graph $G$ with maximum degree $\Delta = 0.5\sqrt{\log{n}}$ such that if for each vertex $v \in V$,
    we independently pick a set $L(v)$ of colors with size $\ell=0.5\sqrt{\log{n}}$ uniformly at random from $2\Delta$ colors, then, 
    with probability $1-o(1)$, there exists no proper coloring of $G$ such that for all vertices $v \in V$ color of $v$ is chosen from $L(v)$. 
\end{proposition}

\begin{proof}
	Consider a graph $G$ which is a collection of $(\ell+1)$-cliques $C_1,\ldots,C_k$ for $k = n/(\ell+1)$. As such, maximum degree
	of this graph is $\Delta = \ell$. For a clique $C_i$, let $L(C_i) := \cup_{v \in C_i} L(v)$ denote the set of sampled
	colors for vertices in $C_i$. As we are sampling the colors from a set of size $2\Delta=2\ell$ colors, and by the independence across vertices in their choice of colors, we have, 
	\begin{align*}
		\Pr\paren{L(C_i) = \set{1,\ldots,\ell}} = {{2\ell}\choose{\ell}}^{-(\ell+1)} \geq \paren{2^{2\ell-2}}^{-(\ell+1)} \geq 2^{-2\ell^2} = n^{-1/2}, 
	\end{align*} 
	by the choice of $\ell = 0.5\sqrt{\log{n}}$. Using the fact that $n/(\ell+1) = \omega(n^{1/2})$ and that the event above is independent across the cliques, with probability $1-o(1)$, there exists a 
	clique $C_i$ in which $L(C_i) = \set{1,\ldots,\ell}$. This clique clearly cannot be colored using the colors $L(v)$ for $v \in C_i$.
\end{proof}

%% file: triangle-free.tex

\newcommand{\FGf}{\ensuremath{\FG^{-{K_3}}_{n,p}}}

\newcommand{\LL}{\ensuremath{\mathcal{L}}}

\newcommand{\passign}{\ensuremath{p_{\textnormal{\textsf{assign}}}}}
\newcommand{\Astar}{\ensuremath{A^{\star}}}
\newcommand{\astar}{\ensuremath{a^{\star}}}
\newcommand{\pkeep}{\ensuremath{{\textnormal{\textsf{keep}}}}}
\newcommand{\pcolor}{\ensuremath{{\textnormal{\textsf{color}}}}}

\newcommand{\wasteful}{\ensuremath{\textnormal{\textsf{WastefulColoring}}}\xspace}

\newcommand{\bp}{\ensuremath{\widetilde{b}}}
\newcommand{\sbp}{\ensuremath{\widetilde{sb}}}
\newcommand{\II}{\ensuremath{\mathbb{I}}}
\newcommand{\Bp}{\ensuremath{\widetilde{B}}}

\subsection{Palette Sparsification for Triangle-Free Graphs}\label{sec:triangle-free}

We now prove a palette sparsification theorem for triangle-free graphs. 

\begin{theorem}\label{thm:ps-triangle-free}
	   Let $G(V,E)$ be any $n$-vertex triangle-free graph with maximum degree $\Delta$. Let $\gamma \in (0,1)$ be a parameter and define $C := C(\gamma) = \Paren{\frac{9\Delta}{\gamma \cdot \ln{\Delta}}}$. 
	   Suppose for every vertex $v \in V$,
    we independently sample a set $L(v)$ of size $b \cdot (\Delta^{\gamma}
+ \sqrt{\log{n}})$ uniformly at random  from colors
$\set{1,\ldots,C}$ for an appropriate absolute positive constant $b$. Then, 
    with high probability there exists a proper coloring of $G$ from lists $L(v)$ for every vertex $v \in V$. 
\end{theorem}

It is known that there are triangle-free graphs with chromatic number $\Omega(\frac{\Delta}{\ln{\Delta}})$~\cite{Bollobas78} (In fact this bound holds even for graphs with arbitrarily large girth not only girth $>3$). 
Theorem~\ref{thm:ps-triangle-free} then shows that one can match the chromatic number of these graphs asymptotically
by sampling only a small number of colors per vertex (almost as small as $O(\Delta^{o(1)} + \sqrt{\log{n}})$ in the limit). 

\subsubsection*{Proof of Theorem~\ref{thm:ps-triangle-free}}

As we already saw in the proof of Theorem~\ref{thm:ps-od-coloring}, looking at the sampled lists $L(v)$ of vertices as a list-coloring problem with constraints on $c$-degrees can be quite helpful in proving the corresponding palette sparsification result. 
We take the same approach in proving Theorem~\ref{thm:ps-triangle-free} as well. However, unlike for $(1+\eps)\Delta$ coloring, to the best of our knowledge, no such list-coloring results (with  constraints on $c$-degrees 
\emph{instead} of maximum degree) are known
for coloring triangle-free graphs. Our main task here is then exactly to prove such a result formalized as follows. 

\begin{proposition}\label{prop:lc-triangle}
   There exists an absolute constant $d_0$ such that for all $d \geq d_0$ the following holds. 
   Suppose $G(V,E)$ is a triangle-free graph with lists $S(v)$ for every $v \in V$ such that:
   \begin{enumerate}[label=(\roman*)]
    	\item for every vertex $v$, $\card{S(v)} \geq 8 \cdot \frac{d}{\ln{d}}$, and
	\item for every vertex $v$ and color $c \in S(v)$, $\deg_S(v,c) \leq d$. 
    \end{enumerate} 
    Then, there exists a proper coloring of $G$ from these lists. 
\end{proposition}

A word of interpretation is in order. It is known that any triangle-free graph $G$ with maximum degree $\Delta$ is $O(\frac{\Delta}{\ln{\Delta}})$ (list-)colorable~\cite{Johansson96a,Molloy19}. 
However, in Proposition~\ref{prop:lc-triangle}, the maximum degree
of a vertex can be as large as $\Theta(d^2/\ln{d})$ even after
omitting all edges between adjacent vertices with disjoint lists,
while the size of each list is only $O(d/\ln{d})$. (In fact this is precisely the setting of parameters we will be interested in while 
proving Theorem~\ref{thm:ps-triangle-free}). Proposition~\ref{prop:lc-triangle} shows that even in this case, as long as the $c$-degrees are bounded by $d$, we can list-color the graph 
with $O(d/\ln{d})$ colors  (similar to Proposition~\ref{prop:lc-eps} for $(1+\eps)\Delta$ coloring)\footnote{It is worth mentioning that 
transforming results on maximum degree to ones on
maximum $c$-degree in general is a non-trivial task and not even always true: it was shown in~\cite{BohmanH02} that there are graphs and lists so that $c$-degree of every vertex is $d$ and still the graph is not $d+1$ list-colorable (even though every graph 
is $(\Delta+1)$ list-colorable).}.  

We give the proof of Theorem~\ref{thm:ps-triangle-free} assuming Proposition~\ref{prop:lc-triangle} here. The proof of Proposition~\ref{prop:lc-triangle} itself  is  technical and detailed and thus even though  interesting on its own, 
we opted to postpone it to Appendix~\ref{app:lc-triangle} to preserve the flow of the paper. 

\begin{proof}[Proof of Theorem~\ref{thm:ps-triangle-free}]
	We prove this theorem with the weaker bound of $O(\Delta^{\gamma} + {\log{n}})$ (as opposed to $O(\Delta^{\gamma} + \sqrt{\log{n}})$) for the number of sampled colors. The extension to the improved bound with $O(\sqrt{\log{n}})$ dependence
	is exactly as in the proof of Theorem~\ref{thm:ps-od-coloring} and is thus omitted. 
	
	Let $\ell := \paren{\Delta^{\gamma}+3000\ln{n}}$ and suppose each vertex samples $\ell$ colors from $\set{1,\ldots,C}$ for $C:= C(\gamma) = \paren{\frac{9\Delta}{\gamma \cdot \ln{\Delta}}}$. 
	Let $p := \ell/C$ which is equal to the probability that any vertex $v$ samples a particular color in $L(v)$. We have, 
	\begin{align*}
		\expect{\deg_L(v,c)} = \sum_{u \in N(v)} \Pr\paren{\text{$u$ samples $c$ in $L(u)$}} \leq p \cdot \Delta. 
	\end{align*}
	Note that as $p \cdot \Delta \geq p \cdot C = \ell \geq 3000\ln{n}$, a simple application of Chernoff bound (Proposition~\ref{prop:chernoff}) plus union bound ensures that, for every vertex $v$ and color $c$, $\deg_L(v,c) \leq (1.1) \cdot p \Delta$ with high probability. In the following, we condition
	on this event. 
	
	Let $d:= (1.1) \cdot p\Delta$. By the above conditioning, $c$-degree of every vertex $v \in V$ is at most $d$. In order to apply Proposition~\ref{prop:lc-triangle} to graph $G$ with lists $L$, we only 
	need to prove that $\ell \geq \frac{8d}{\ln{d}}$. We prove that in fact $\ell \cdot \ln{\ell} \geq 8d$ which implies the desired bound as $\ell = p \cdot C \leq p \cdot \Delta \leq d$. We have, 
	\begin{align*}
		\ell \cdot \ln{\ell} &\geq \paren{p \cdot C} \cdot \ln{\paren{\Delta^{\gamma}}} \tag{as $\Delta^{\gamma} < \ell = p \cdot C$ and by the choice of $C$} 
		= p \cdot \paren{\frac{9\Delta}{\gamma \cdot \ln{\Delta}}} \cdot \gamma \cdot \ln{\Delta} = 9 \cdot p \Delta > 8d.
	\end{align*}
	The proof now follows from applying Proposition~\ref{prop:lc-triangle} to lists $L$. \Qed{Theorem~\ref{thm:ps-triangle-free}}
	
\end{proof}

\subsubsection*{Asymptotic Optimality of the Bounds in Theorem~\ref{thm:ps-triangle-free}}

We now prove the optimality of Theorem~\ref{thm:ps-triangle-free} up to constant factors. 

\begin{proposition}\label{prop:tf-optimal}
    There exists a distribution on $n$-vertex graphs with maximum 
degree $\Delta = \Theta(n^{1/3})$ such that for 
every $\gamma < 1/16$ and $C:= C(\gamma) = \frac{\Delta}{16\gamma \cdot \ln{\Delta}}$ the following is true. Suppose we sample a graph $G(V,E)$ 
from this distribution and then for each vertex $v \in V$,
    we independently pick a set $L(v)$ of colors with size $\Delta^{\gamma}$ uniformly at random from colors $\set{1,\ldots,C}$; then, 
    with high probability there exists no proper coloring of $G$ where for all $v \in V$ color of $v$ is chosen from $L(v)$. 
\end{proposition}

Let $\FG_{n,p}$ denote the Erd\H{o}s-R\'enyi 
distribution of random graphs on $n$ vertices in which each edge is chosen independently with probability $p$. 
Define the following distribution $\FGf$ on triangle-free graphs: Sample a graph $G$ from $\FG_{n,p}$, then remove \emph{every} edge that was part of a triangle originally. 
Clearly, the graphs output by $\FGf$ are triangle-free. Throughout this section, we take $p = \Theta(n^{-2/3})$ (the exact choice of the leading constant will be determined later). 

We prove Proposition~\ref{prop:tf-optimal} by considering the distribution $\FGf$. However, we first present some basic properties 
of distribution $\FG_{n,p}$ needed for our purpose. The proofs are simple exercises in random graph theory and are provided in Appendix~\ref{app:random-graph-theory} for completeness. 
In the following, let $t(G)$ denote the number of triangles in $G$ and $\alpha(G)$ denote the maximum independent set size, and 
recall that $\Delta(G)$ denotes the maximum degree of $G$. 

\begin{lemma}\label{lem:random-triangle}
	For $G \sim \FG_{n,p}$, $\expect{t(G)} \leq (np)^3$, and $t(G) \leq (1+o(1))\expect{t(G)}$ w.h.p.
\end{lemma}

\begin{lemma}\label{lem:g5-random}
	For $G \sim \FG_{n,p}$, $\expect{\alpha(G)} \leq  
\frac{3 \cdot \ln{(np)}}{p}$, and $\alpha(G)  \leq  \frac{3 \cdot
\ln{(np)}}{p}$ w.h.p.
\end{lemma}

\begin{lemma}\label{lem:g5-max-degree}
	For $G \sim \FG_{n,p}$, $\Delta(G) \leq 2np$ w.h.p. 
\end{lemma}

We are now ready to prove Proposition~\ref{prop:tf-optimal}. 

\newcommand{\hDelta}{\ensuremath{\widetilde{\Delta}}}
\newcommand{\eventsize}{\ensuremath{\event_{\textnormal{\textsf{size}}}}}

\begin{proof}[Proof of Proposition~\ref{prop:tf-optimal}]
	Let $p := \frac{1}{3} \cdot (n)^{-2/3}$ for this proof and consider the distribution $\FGf$. Moreover, let $\LL$ denote the distribution of lists of colors sampled for vertices.  
	By Lemma~\ref{lem:g5-max-degree}, the maximum degree of $G \sim \FG_{n,p}$ and consequently $G \sim \FGf$ is at most $\hDelta := 2np$ with high probability. 
	Throughout the following argument, we condition on this event. This can only change the probability calculations by a negligible factor (that we ignore for the simplicity of exposition).
	This way, the number of colors sampled in $\LL$ can be assumed to be at most $C := \frac{\hDelta}{16\gamma\cdot \ln{\hDelta}}$. We further use $q := {\hDelta^{\gamma}}/{C}$ 
	to denote the probability that a color $c$ is sampled in list $L(v)$ of a vertex $v$.

	For a graph $G(V,E) \sim \FGf$ and lists $L \sim \LL$, let $V_1,\ldots,V_C$ be a collection of subsets of $V$ (not necessarily disjoint)
	where for every $c \in [C]$, $V_c$ denotes the vertices $v$ that sampled the color $c$ in their list $L(v)$. As each color is sampled with probability $q$ by a vertex, 
	and the choices are independent across vertices, a simple application of Chernoff bound ensures that with high probability, 
$\card{V_c} \leq 2q \cdot n$ for all $c$. 
	We also condition on this event in the following (and similarly as before ignore the negligible contribution of this conditioning to the probability calculations below).

	Let $\delta$ denote the probability of ``error'' i.e.,  the event that the sampled colors do not lead to a proper coloring of 
	the graph. An averaging argument implies that there exists a fixed set of lists $L \sim \LL$ such that for $G$ sampled from $\FGf$, 
	the error probability of $L$ on $G$ is at most $\delta$.
Fix such a choice of $L$ in the following. We will show that
$\delta =1-o(1)$.  
	
	Recall that $G \sim \FGf$ is chosen independent of the
lists $L$ (by definition of palette sparsification). 
For any graph $G$, define:
	\begin{itemize}
	\item  $\mu_L(G) := \max_{(U_1,\ldots,U_C)} \sum_{c=1}^C \card{U_c}$  where all $U_c$'s are \emph{disjoint}, each $U_c \subseteq V_c$, and $G[U_c]$ is an independent set. 
	\end{itemize}
	As we have fixed the choice of the lists $L$, the function $\mu_L(\cdot)$ is fixed at this point and its value only depends on $G$. 
	A necessary condition for $G$ to be colorable from the lists $L$ is that $\mu_L(G) = n$. This is because $(i)$ any proper coloring of $G$ from lists $L$ necessarily induces an independent set inside each $V_c$; $(ii)$  
	these independent sets are disjoint and hence we can take them as a feasible solution $(U_1,\ldots,U_C)$ to $\mu_L(G)$; $(iii)$ these independent sets cover all vertices of $G$. Our task is 
	now to bound the probability that $\mu(G) = n$ to lower bound $\delta$. 
	
	Firstly, we can switch from the distribution $\FGf$ to $\FG_{n,p}$ using the following equation (recall that $t(G)$ denotes the number of triangles): 
	\begin{align}
		\EX_{G \sim \FGf}\bracket{\mu_L(G)} \leq \EX_{H \sim \FG_{n,p}}\bracket{\mu_L(H) + 3 \cdot t(H)} \label{eq:tri-G-H}.  
	\end{align}
	This is because any graph $G \sim \FGf$ is obtained by removing edges of every triangle in a graph $H \sim \FG_{n,p}$ and removing these edges can only increase the total size of a collection
	 of \emph{disjoint} independent sets (namely, the value of
$\mu_L$) by the number of vertices in the triangles (in fact, by at
most two vertices from each triangle). We can upper bound the second-term in Eq~\eqref{eq:tri-G-H} using Lemma~\ref{lem:random-triangle}. 
	 We now bound the first term. In the following, let $n_c := \card{V_c}$ for $c \in [C]$.  We have, 
	 \begin{align*}
	 \EX_{H \sim \FG_{n,p}}\bracket{\mu_L(H)} &\leq  \EX_{H \sim \FG_{n,p}}\bracket{\sum_{c=1}^{C} \alpha(H[V_c])}, \tag{by removing the disjointness condition between sets $U_c$'s we can only increase value of $\mu_L(H)$} \\
	 &= \sum_{c=1}^{C} \EX_{H_c \sim \FG_{n_c,p}}\bracket{\alpha(H_c)} \tag{by linearity of expectation and as for every $c \in [C]$,  $H[V_c]$ is  sampled from $\FG_{n_c,p}$} \\
	 &\leq \sum_{c=1}^{C} \frac{3 \cdot \ln{(n_c p)}}{p} \tag{by Lemma~\ref{lem:g5-random}} \\
	 &\leq C \cdot \frac{3 \cdot \ln{(2qn \cdot p)}}{p} \tag{as we conditioned on $n_c \leq 2q \cdot n$} \\
	 &= \frac{\hDelta}{16\gamma\cdot \ln{\hDelta}} \cdot \frac{3 \cdot \ln{(q \cdot \hDelta)}}{(\hDelta/2n)} \tag{by definitions of $C$ and $\hDelta$} \\
	 &= \frac{6n}{16} \cdot \frac{\ln{(q \cdot \hDelta)}}{\ln{(\hDelta^\gamma)}} \tag{by a simple re-arranging of terms} \\
	 &< \frac{6n}{8}. \tag{as $\ln{\paren{q \cdot \hDelta}} = \ln{\paren{\hDelta^{\gamma} \cdot 16\gamma \cdot \ln{\hDelta}}} < 2\ln{\paren{\hDelta^\gamma}}$} 
	 \end{align*}
	 Plugging this in Eq~\eqref{eq:tri-G-H} together with Lemma~\ref{lem:random-triangle} to bound the second term, implies that: 
	 \begin{align*}
	 	\EX_{G \sim \FGf}\bracket{\mu_L(G)} \leq \frac{6n}{8} + 3 \cdot (\frac{n^{1/3}}{3})^{3} < \frac{7n}{8}. 
	 \end{align*}
	 Finally, by the assertions of Lemma
\ref{lem:random-triangle} and Lemma \ref{lem:g5-random}, 
$\mu_L(G) < n$ w.h.p. This implies that $\delta=1-o(1)$ as needed. 
\Qed{Proposition~\ref{prop:tf-optimal}}
	  
\end{proof}

%% file: palette-sparsification-local.tex

\section{A Local Version of Palette Sparsification}\label{sec:ps-local}

We now give a ``local version'' (see, e.g.~\cite{DaviesVKP18,BonamyKNP2018}) of the palette sparsification theorem in which the initial number of available colors for vertices depends on the
local parameters of the vertices, namely, their degree, as opposed to a global parameter such as maximum degree. 


\begin{theorem}\label{thm:ps-deg+1-coloring}
    Let $G(V,E)$ be any $n$-vertex graph and assume each vertex $v \in V$ is given a list $S(v)$ of colors. Suppose for every vertex $v \in V$,
    we independently sample a set $L(v)$ of colors of size $\ell$ uniformly at random  from colors in $S(v)$:
    \begin{enumerate}[label=(\roman*)]
    	\item\label{part:deg+1-p1} if $S(v)$ is \emph{any arbitrary set} of $(1+\eps)\cdot\deg(v)$ colors and $\ell = \Theta(\eps^{-1}\cdot\log{n})$ for  $\eps > 0$, 
	\item\label{part:deg+1-p2} or if $S(v) = \set{1,\ldots,\deg(v)+1}$ and $\ell = \Theta(\log{n})$, 
    \end{enumerate} 
    then, with high probability, there exists a proper coloring of $G$ from lists $L(v)$ for  $v \in V$. 
\end{theorem}


The main part of the proof of Theorem~\ref{thm:ps-deg+1-coloring} is Part~\ref{part:deg+1-p2} as the proof of the first part follows almost directly from this proof. 
However, we start with a standalone proof of Part~\ref{part:deg+1-p1} as a warm-up and then present the proof of Part~\ref{part:deg+1-p2}, which involves the bulk of our effort in this section.

\input{deg+1-coloring}

%% file: deg+1-coloring.tex

\newcommand{\Nunfri}{\ensuremath{\overline{{F}}}}
\newcommand{\Nunbal}{\ensuremath{\overline{{B}}}}
\newcommand{\Nunbalp}{\ensuremath{\overline{{B}}}_{+}}
\newcommand{\Nunbalm}{\ensuremath{\overline{{B}}}_{-}}

\newcommand{\Nrem}{\ensuremath{{{R}}}}

\newcommand{\Vunbal}{\ensuremath{V^{\textnormal{\textsf{uneven}}}}}

\newcommand{\FirstColor}{\ensuremath{\textnormal{\textsf{FirstStepColoring}}}\xspace}
\newcommand{\SecondColor}{\ensuremath{\textnormal{\textsf{SecondStepColoring}}}\xspace}

\newcommand{\abort}{\textnormal{\textbf{abort}}\xspace}

\newcommand{\bard}{\ensuremath{\overline{d}}}

\newcommand{\Se}{\ensuremath{S_{\textnormal{\textsf{ext}}}}}
\newcommand{\se}{\ensuremath{s_{\textnormal{\textsf{ext}}}}}

\newcommand{\adeg}{\ensuremath{\deg_{\textnormal{\textsf{active}}}}}
\newcommand{\eventact}{\ensuremath{\event_{\textnormal{\textsf{active}}}}}

\newcommand{\ngood}{\ensuremath{n_{\textnormal{\textsf{good}}}}}
\newcommand{\Uactive}{\ensuremath{U^{\textnormal{\textsf{active}}}}}

\newcommand{\RR}{\ensuremath{\mathcal{R}}}

\newcommand{\Enon}{\ensuremath{\textnormal{NonEdge}}}
\newcommand{\Egood}{\ensuremath{{E}_{\textnormal{\textsf{good}}}}}

\newcommand{\Ndeg}{\ensuremath{N_{\textnormal{\textsf{deg}}}}}
\newcommand{\Ndegp}{\ensuremath{N^+_{\textnormal{\textsf{deg}}}}}
\newcommand{\Enono}{\ensuremath{{E}^+_{\textnormal{\textsf{non}}}}}

\subsection{Warm Up: Palette Sparsification for $(1+\eps)\deg$ List-Coloring}\label{sec:dod-coloring}

\begin{proof}[Proof of Theorem~\ref{thm:ps-deg+1-coloring} -- Part~\ref{part:deg+1-p1}]
Fix any $\eps > 0$ (not necessarily a constant) and suppose we sample $\ell := \frac{10}{\eps} \cdot \ln{n}$ colors $L(v)$ from $S(v)$ for every vertex $v \in V$. 
Consider the following process: 
\begin{tbox}
\begin{enumerate}
	\item Iterate over vertices $v$ in an \emph{arbitrary} order and for each vertex $v$, let $N^{<}(v)$ denote the neighbors of $v$ that appear before $v$ in this ordering. 
	\item\label{line:abort} For each vertex $v$, if there exists a color $c(v)$ in $L(v)$ that is not used to color any vertex $u \in N^{<}(v)$, color $v$ with $c(v)$. Otherwise  \abort. 
\end{enumerate}
\end{tbox}
\noindent
We argue that this procedure will terminate with high probability without having to \abort. This ensures that $G$ is colorable from sampled lists $L$, thus proving Part~\ref{part:deg+1-p1} of Theorem~\ref{thm:ps-deg+1-coloring}. 
We have, 
\begin{align*}
	\Pr\paren{\text{\abort}} &\leq \sum_{v} \Pr\paren{\text{$L(v)$ is a subset of colors chosen for $N^{<}(v)$}} \tag{by union bound} \\ 
	&\leq \sum_{v} \paren{\frac{\card{N^{<}(v)}}{\card{S(v)}}}^{\ell} \leq n \cdot \paren{\frac{\deg(v)}{(1+\eps)\cdot\deg(v)}}^{\ell} \leq n \cdot (1-\eps/2)^{\ell} \tag{by the sampling without replacement
	 procedure of Theorem~\ref{thm:ps-deg+1-coloring} }\\ 
	&\leq n \cdot \exp\paren{-\frac{\eps}{2} \cdot \frac{10}{\eps} \cdot \ln{n}} = n^{-4}. \tag{by the choice of $\ell$}
\end{align*}
This concludes the proof of Part~\ref{part:deg+1-p1} of Theorem~\ref{thm:ps-deg+1-coloring}. \Qed{Theorem~\ref{thm:ps-deg+1-coloring}}

\end{proof}

We conclude this part by noting that our proof above can be also tailored to obtain a palette sparsification theorem for coloring a graph with ``about $\kappa$'' colors where $\kappa$ is the \emph{degeneracy} of the graph 
(see~\cite{BeraCG19} for a recent application of such a result to algorithms in ``space-conscious'' models). 

\begin{remark}[\textbf{Palette sparsification for coloring via degeneracy}]\label{rem:degeneracy-coloring}
	For the above proof, we  considered an \emph{arbitrary} ordering of vertices and upper bounded $\card{N^<(v)}$ by $\card{N(v)} = \deg(v)$ which sufficed for our purpose. However, if we instead worked
	with the \emph{degeneracy ordering} of vertices\footnote{A degeneracy ordering of $G$ is obtained by repeatedly picking the vertex of minimum remaining degree, removing it and updating the degree of remaining vertices, and moving on to the next vertex.}, we could have upper bounded $\card{N^<(v)}$ by $\kappa(v) \leq \kappa$ where $\kappa$ is the degeneracy of the graph and $\kappa(v) \leq \deg(v)$ is the degree of $v$ in the degeneracy ordering. 
	This immediately allows us to extend the previous argument
	to the case where size of each $S(v)$ is only $(1+\eps)\kappa(v)$. This shows that {palette sparsification works for coloring with ``about $\kappa$'' colors} (and $\kappa(v)$ colors for a local version). 
\end{remark}

Remark~\ref{rem:degeneracy-coloring} is closely related to a very recent work of Bera, Chakrabarti, and Ghosh~\cite{BeraCG19} that obtained similar-in-spirit results for graph coloring using about $\kappa$ colors based on graph partitioning (see Section~\ref{sec:vertex-sampling}). 
Our Remark~\ref{rem:degeneracy-coloring} thus gives an alternative way of obtaining (some of the) sublinear algorithms for $\kappa+o(\kappa)$ coloring studied in~\cite{BeraCG19} such as streaming and sublinear-time algorithms. 
As such results (in more details) have already been obtained in~\cite{BeraCG19} and this is not the contribution of 
our work, we omit the details and only note that in our approach, unlike~\cite{BeraCG19}, an additional care is also needed to keep the running time of algorithms small. 

Finally, we note that~\cite{BeraCG19} shows that obtaining a $(1+\eps)\kappa$ coloring via palette sparsification requires sampling $\Omega(\log{n}/\poly(\eps))$ colors per vertex (when $\eps = o(1/\log{n})$); our upper bound
matches this bound to within $\poly(1/\eps)$ terms. 

\subsection{Palette Sparsification for $(\deg+1)$ Coloring}\label{sec:d1-coloring}

We now prove the second and the main part of Theorem~\ref{thm:ps-deg+1-coloring}. 
We follow the approach of~\cite{AssadiCK19} for $(\Delta+1)$ coloring problem (outlined in Appendix~\ref{sec:background-ACK19}) to prove this result. The key difference here is that the graph decomposition for partitioning the graph into \emph{sparse} and \emph{dense} parts that played a key role in~\cite{AssadiCK19} is no longer applicable to the $(\deg+1)$ coloring problem. 

In the following, we first give a new graph decomposition tailored to $(\deg+1)$ coloring problem and states its main properties as well as its differences 
with similar decompositions for $(\Delta+1)$ coloring in~\cite{HarrisSS16,ChangLP18,AssadiCK19} (themselves  based on~\cite{Reed98}). 
The next step is then to show that this decomposition, even though ``weaker'' than the one for $(\Delta+1)$ coloring, still 
has enough structure to carry out the proof for $(\deg+1)$ coloring along the lines of the one for $(\Delta+1)$ coloring in~\cite{AssadiCK19} with the main difference 
being on how we handle the ``sparse'' vertices in our new decomposition. 

\subsubsection{A Graph Decomposition for $(\deg+1)$ Coloring}

Let $\eps \in (0,1)$ be a parameter. We define the following structures for any graph $G(V,E)$. 

\begin{definition}\label{def:almost-clique}
	We say that an induced subgraph $K$ of $G$ is an \textbf{$\bm{\eps}$-almost-clique} iff:
	\begin{enumerate}[label=(\roman*)]
		\item\label{ac1} For every $v \in K$, $\deg_G(v) \geq (1-8\eps) \cdot \Delta(K)$ where we define $\Delta(K) := \max_{v \in K} \deg_G(v)$;
		\item\label{ac2} $(1-\eps) \cdot \Delta(K) \leq \card{V(K)} \leq (1+8\eps) \cdot \Delta(K)$;
		\item\label{ac3} Any vertex $v \in K$ has at most $8\eps \cdot \Delta(K)$ \emph{non-neighbors} (in $G$) \emph{inside} $K$; 
		\item\label{ac4} Any vertex $v \in K$ has at most $9\eps \cdot \Delta(K)$ \emph{neighbors} (in $G$) \emph{outside} $K$.
	\end{enumerate}
\end{definition}

Definition~\ref{def:almost-clique} can be seen as a natural analogue of $(\Delta,\eps)$-almost-cliques
defined in~\cite{AssadiCK19} (see Appendix~\ref{sec:background-ACK19}). The main difference is that instead of having dependence on the global parameter $\Delta$ in a $(\Delta,\eps)$-almost-clique of~\cite{AssadiCK19}, our $\eps$-almost-cliques 
only depend on $\Delta(K)$ which is a $(1+\Theta(\eps))$-approximation of the degree of every vertex in $K$ (and thus can be much smaller than $\Delta$). 

\begin{definition}\label{def:eps-sparse}
	We say a vertex $v \in G$ is \textbf{$\bm{\eps}$-sparse} iff there are at least $\eps^2 \cdot {{\deg(v)}\choose{2}}$ \emph{non-edges} in the neighborhood of $v$. 
\end{definition}

Again, Definition~\ref{def:eps-sparse} is a natural analogue of sparse vertices in~\cite{AssadiCK19,HarrisSS16,ChangLP18} by replacing the dependence on $\Delta$ with $\deg(v)$ instead. 

\begin{definition}\label{def:eps-unbalanced}
	We say a vertex $v \in G$ is \textbf{$\bm{\eps}$-uneven} iff for \emph{at least} $\eps \cdot \deg(v)$ neighbors $u$ of $v$, we have $\deg(v) < (1-\eps)\cdot\deg(u)$. 
\end{definition}

Roughly speaking, a vertex $v$ is considered uneven if it has a ``sufficiently large'' number of neighbors with ``sufficiently larger'' degree than $v$. Definition~\ref{def:eps-unbalanced} is tailored specifically to $(\deg+1)$ coloring problem and does not 
have an analogue in~\cite{AssadiCK19,HarrisSS16,ChangLP18} for $(\Delta+1)$ coloring. We prove the following decomposition result using the definitions above. 

\begin{lemma}[Graph Decomposition for $(\deg+1)$ Coloring]\label{lem:decomposition}
	For any sufficiently small $\eps > 0$, any graph $G(V,E)$ can be partitioned into vertices $V := \Vunbal \sqcup \Vsparse \sqcup K_1\sqcup \ldots \sqcup K_k$ such that: 
	\begin{enumerate}[label=(\roman*)]
		\item For every $i \in [k]$, the induced subgraph $G[K_i]$ is an $\eps$-almost-clique; 
		\item Every vertex in $\Vsparse$ is $(\eps/2)$-sparse;
		\item Every vertex in $\Vunbal$ is $(\eps/4)$-uneven. 
	\end{enumerate}
\end{lemma}
\noindent
The key difference of Lemma~\ref{lem:decomposition} with prior decompositions for $(\Delta+1)$ coloring in~\cite{Reed98,AssadiCK19,HarrisSS16,ChangLP18} is the introduction of $\Vunbal$ that
captures vertices with ``sufficiently large'' higher degree neighbors. Allowing for such vertices is (seemingly) crucial for this type of decomposition that depends on the local degrees of vertices as opposed to maximum degree%
 \footnote{For instance, consider a vertex of degree $d$ that is incident to $d$ vertices of a $2d$-clique. Such a vertex is neither sparse (its neighborhood is a clique), nor belongs to an almost-clique for small $\eps < 1$.}. 

\smallskip
Before we move on, a word of caution is in order. By definition, any $\eps$-almost-clique is also an $\eps'$-almost clique for $\eps' \geq \eps$.  On the other hand, 
the exact opposite relation holds for $\eps$-sparse and $\eps$-uneven vertices: any $\eps$-sparse vertex is also $\eps''$-sparse for $\eps'' \leq \eps$ (similarly for uneven vertices). 
As such, one cannot simply ``rescale'' the value of $\eps$ in above definitions and lemma directly (although there are enough slacks in our arguments to allow for proper changes when needed).

\subsubsection*{Proof of Lemma~\ref{lem:decomposition}}
We prove this lemma through a series of simple claims along the lines of the HSS decomposition~\cite{HarrisSS16} and its extension in~\cite{AssadiCK19}. The general approach is similar to~\cite{HarrisSS16,AssadiCK19} but there are some key differences in several places as well.

We start with some necessary definitions. For any sufficiently small  $\theta \in (0,1)$ ($\theta < 1/20$ suffices for our purpose), we define the following: 
	\begin{itemize}[leftmargin=15pt]
		\item An edge $(u,v)$ is \textbf{$\bm{\theta}$-balanced} iff $\min\set{\deg(u),\deg(v)} \geq (1-\theta) \cdot \max\set{\deg(u),\deg(v)}$. 
		\item An edge $(u,v)$ is \textbf{$\bm{\theta}$-friend} iff it is $\theta$-balanced and $\card{N(u) \cap N(v)} \geq (1-\theta) \cdot \min\set{\deg(u),\deg(v)}$. 
		\item A vertex $v$ is \textbf{$\bm{\theta}$-dense} iff it is incident on at least $(1-\theta) \cdot \deg(v)$ many $\theta$-friend edges. 
	\end{itemize}
	Let $F_{\theta} \subseteq E$ denote the set of $\theta$-friend edges and $D_{\theta} \subseteq V$ denote the set of $\theta$-dense vertices. 
	Consider the (not necessarily induced) subgraph $H_{\theta}$ of $G$ defined as $H_{\theta} := (D_{\theta},F_{\theta})$, i.e., the subgraph on $\theta$-dense vertices
	and consisting of only the $\theta$-friend edges (here we slightly abused the notation as endpoints of some edges in $F_\theta$ may not belong to $D_\theta$ in which case we ignore them
	in $H_\theta$ as well). 
	
\paragraph{Handling Vertices in $\bm{D_{\theta}}$.} 
	We use connected components of $H_\theta$ to identify the almost-cliques 
	in the decomposition (where we take $\theta = \Theta(\eps)$). To do so, we need a series of simple claims. In the following, we use $C$ to denote an arbitrary connected component of $H_{\theta}$. 

	\begin{claim}\label{clm:decomp-shared-neighbors}
		For any $u,v \in C \subseteq D_\theta$, $\card{N_{}(u) \cap N_{}(v)} \geq (1-5\theta) \cdot \min\set{\deg(u),\deg(v)}$. 
	\end{claim}
	\begin{proof}
		Consider a path $u = w_0, w_1,\ldots,w_t = v$ between $u$ and $v$ in $H_\theta$ ($u$ and $v$ belong to the same connected component). 
		We prove inductively that for every $i \in [t]$ (the case $i=t$ proves the claim): 
		\begin{align*}
		\card{N_{}(u) \cap N_{{}}(w_i)} &\geq (1-5\theta) \cdot \min\set{\deg(u),\deg(w_i)}, \text{and} \\
		\min\set{\deg(u),\deg(w_i)} &\geq (1-2\theta) \cdot \max\set{\deg(u),\deg(w_i)}.
		\end{align*}

		The induction step for $i=1$ is true because $(u,w_1)$ is a $\theta$-friend edge. 
		Now suppose this is true up until some $i$ and consider $i+1$. Since $(w_i,w_{i+1})$ is a $\theta$-friend edge, we have:  
		\begin{align}
			\card{N_{}(w_i) \cap N_{{}}(w_{i+1})} &\geq (1-\theta) \cdot \min\set{\deg(w_i),\deg(w_{i+1})}, \text{and} \notag \\
		\min\set{\deg(w_i),\deg(w_i)} &\geq (1-\theta) \cdot \max\set{\deg(w_i),\deg(w_{i+1})}. \label{eq:decomp1} 
		\end{align}
		On the other hand, the induction hypothesis implies that:
		\begin{align}
		\card{N_{}(u) \cap N_{{}}(w_i)} &\geq (1-5\theta) \cdot \min\set{\deg(u),\deg(w_i)}, \text{and} \notag \\
		\min\set{\deg(u),\deg(w_i)} &\geq (1-2\theta) \cdot \max\set{\deg(u),\deg(w_i)}. \label{eq:decomp2}
		\end{align}
		We use this to show that there exists a vertex $z$ (not necessarily in $C$ or even $D_\theta$) 
		such that both $(u,z)$ and $(z,w)$ are $\theta$-friend edges. As $u$ is $\theta$-dense and by Eq~\eqref{eq:decomp2}, we have that $u$ has a $\theta$-friend edge to at least 
		$(1-8\theta) \cdot \deg(w_i)$ neighbors of $w_i$. Similarly, as $w_{i+1}$ is $\theta$-dense and by Eq~\eqref{eq:decomp1}, 
		we have that $w_{i+1}$ has a $\theta$-friend edge to at least $(1-3\theta) \deg(w_i)$ neighbors of $w_i$. For $\theta < 1/11$, this implies
		that there exists some neighbor $z$ of $w_i$ where both $u$ and $w_{i+1}$ have a $\theta$-friend edge to. 
		
		Since $(u,z)$ and $(z,w_{i+1})$ are $\theta$-friend edges and thus $\theta$-balanced as well, we obtain the second part of the induction hypothesis for $i+1$. For the first part, again by
		using the fact that $(u,z)$ and $(z,w_{i+1})$ are $\theta$-friend edges, we have that : 
		\begin{align*}
			\card{N(u) \cap N(z)} &\geq (1-\theta) \cdot \min\set{\deg(u),\deg(z)}, \text{and} \\
			\card{N(z) \cap N(w_{i+1})} &\geq (1-\theta) \cdot \min\set{\deg(z),\deg(w_{i+1})}.
		\end{align*}
		implying that $\card{N(u) \cap N(w_{i+1})} \geq (1-5\theta) \cdot \min\set{\deg(u),\deg(z)}$ (using the bound on degrees of $u$ and $w_{i+1}$). This concludes the proof of the induction hypothesis 
		and the claim. \Qed{Claim~\ref{clm:decomp-shared-neighbors}}
		
	\end{proof}
	The following claim is an immediate corollary of Claim~\ref{clm:decomp-shared-neighbors} (and was directly proved there). 
	\begin{claim}\label{clm:decomp-degree}
		For any $u,v \in C \subseteq D_\theta$, $\min\set{\deg(u),\deg(v)} \geq (1-2\theta) \max\set{\deg(u),\deg(v)}.$ 
	\end{claim}
	
	We further bound the number of $\theta$-dense neighbors of any vertex $v \in C$ that are outside $C$. 
	\begin{claim}\label{clm:decomp-out-neighbor}
		For any $v \in C$, $\card{N(v) \cap D_\theta \setminus C} \leq 2\theta \cdot \deg(v)$. 
	\end{claim}
	\begin{proof}
		As $v$ is a $\theta$-dense vertex, it has at least $(1-\theta) \cdot \deg(v)$ edges that are $\theta$-friend edges. If the end point of any such edge belongs to $D_\theta$, then that vertex
		clearly belongs to $C$ as well. As such, at most $\theta \cdot \deg(v)$ neighbors of $v$ that are in $D_\theta$ maybe outside of $C$, proving the claim. \Qed{Claim~\ref{clm:decomp-non-neighbor}}
		 
	\end{proof}
	
	The next step is to bound the number of non-neighbors of any vertex $v \in C$ inside $C$. Following~\cite{HarrisSS16}, we do this via a double-counting argument. 
	However, we shall note that the parameter  we use for double-counting is crucially different than the one in~\cite[Lemma~3.9]{HarrisSS16}. 
	\begin{claim}\label{clm:decomp-non-neighbor}
		For any $v \in C$, $\card{C \setminus N(v)} \leq 2\theta \cdot \deg(v)$. 
	\end{claim}
	\begin{proof}
		Let $\bard(v) := \card{C \setminus N(v)}$ denote the number of non-neighbors of $v$ in $C$. 
		Let $T$ denote the number of triples $(v,w,u)$ where $(v,w)$ and $(w,u)$ are both $\theta$-friend edges of $G$ while $u \in C \setminus N(v)$. We have, 
		\begin{align*}
			T &= \sum_{u \in C \setminus N(v)} \card{\set{w:  (v,w),(w,u) \in F_{\theta}}} \tag{by definition} \\
			&\geq \sum_{u \in C \setminus N(v)} (1-5\theta) \cdot \min\set{\deg(u),\deg(v)}-2\theta\cdot\max\set{\deg(u),\deg(v)} \tag{by Claim~\ref{clm:decomp-shared-neighbors} and since both $u$ and $v$ are $\theta$-dense} \\
			&\geq \bard(v) \cdot (1-9\theta) \cdot \deg(v) \tag{by definition of $\bard(v)$ and Claim~\ref{clm:decomp-degree} as both $u,v \in C$}; \\ \\ 
			T &= \sum_{w: (v,w) \in F_{\theta}} \card{\set{u: (w,u) \in F_\theta} \cap (C \setminus N(v))} \tag{by definition} \\
			&\leq \sum_{w: (v,w) \in F_{\theta}} \card{N(w) \setminus N(v)} \leq \deg(v) \cdot \theta \cdot \deg(v) \tag{as $w$ and $v$ are $\theta$-friend}. 
		\end{align*}
		Combining the bounds above implies that $\bard(v) \leq \frac{\theta}{1-9\theta}\cdot \deg(v) \leq 2\theta \cdot \deg(v)$ for $\theta < 1/18$. \Qed{Claim~\ref{clm:decomp-non-neighbor}}
		
	\end{proof}
	
	The following claim summarizes the key properties of connected components of $H_{\theta}$. 
	\begin{claim}\label{clm:decomp-cc}
		For any connected component $C$ of $H_{\theta}$, define $\Delta(C) := \max_{v \in C} \deg(v)$. Then: 
		\begin{enumerate}[label=(\roman*)]
			\item For all $v \in C$, $\deg(v) \geq (1-2\theta) \cdot \Delta(C)$; 
			\item For all $v \in C$, $\card{N(v) \cap D_\theta \setminus C} \leq 2\theta \cdot \Delta(C)$; 
			\item For all $v \in C$, $\card{C \setminus N(v)} \leq 2\theta \cdot \Delta(C)$;
			\item Size of $C$ is $\card{C} \leq (1+2\theta) \cdot \Delta(C)$. 
		\end{enumerate}
	\end{claim}
	\begin{proof}
		The first three items are restatements of Claims~\ref{clm:decomp-degree},~\ref{clm:decomp-out-neighbor},~\ref{clm:decomp-non-neighbor} and the last one is an immediate corollary of Claim~\ref{clm:decomp-non-neighbor}. 
		\Qed{Claim~\ref{clm:decomp-cc}}
		
	\end{proof}
	
	\paragraph{Handling Vertices Not in $\bm{D_{\theta}}$.} So far, we only focused on vertices of $D_{\theta}$ (through connected components of $H_{\theta}$). We now show a simple property of vertices that are not in $D_{\theta}$ that
	would immediately allows us to partition them into $\Vsparse$ and $\Vunbal$. 
	
	\begin{claim}\label{clm:decomp-not-dense}
		Any vertex $v$ \emph{not} in $D_{\theta}$ is either $(\theta/2)$-sparse or $(\theta/4)$-uneven. 
	\end{claim}
	\begin{proof}
		Because $v$ is not $\theta$-sparse, it has at least at least $\theta \cdot \deg{(v)}$ neighbors that are \emph{not} $\theta$-friend with $v$. Let $\Nunfri(v) \subseteq N(v)$ denote the set of these vertices.
		Recall that a vertex $u$ is not $\theta$-friend with $v$ iff either $(u,v)$ is not a $\theta$-balanced edge or $\card{N(u) \cap N(v)} < (1-\theta) \cdot \min\set{\deg(u),\deg(v)}$. 
		Let $\Nunbal(v)$ denote the vertices in $N_1(v)$ that were added because of the first reason and $\Nrem(v)$ denote the remaining vertices in $\Nunfri(v)$. There a couple cases to consider here. 
		
		\smallskip
		{\textbf{Case 1: $\bm{\card{\Nunbal(v)} < \card{\Nrem(v)}}$.}} Any vertex $u$ in $\Nrem(v)$ contributes at least $\theta \cdot \deg(v)$ non-edges to the neighborhood of $v$ (when $\deg(u) < \deg(v)$ it can only contribute more non-edges). 
		As such, in this case there are at least 
		\begin{align*}
			\frac{1}{2} \cdot \card{\Nrem(v)} \cdot \Paren{\theta \cdot \deg(v)} \geq \frac{1}{2} \cdot \Paren{\frac{\theta \cdot \deg(v)}{2}} \cdot \Paren{\theta \cdot \deg(v)} = (\theta/2)^2 \cdot \deg(v)^2,
		\end{align*}
		many non-edges in the neighborhood of $v$; hence $v$ is $(\theta/2)$-sparse in this case. 
		
		\smallskip
		{\textbf{Case 2: $\bm{\card{\Nunbal(v)} \geq \card{\Nrem(v)}}$.}} Let $\Nunbalp(v)$ denote $u \in \Nunbal(v)$ where $\deg(v) < (1-\theta) \cdot \deg(u)$ 
		and $\Nunbalm(v)$ denote the ones where $\deg(u) < (1-\theta) \cdot \deg(v)$ (since $(u,v)$ is not $\theta$-balanced, one of the two cases must happen for $u$). We partition this case into another two cases. 
		
		\smallskip
		{\textbf{Case 2a: $\bm{\card{\Nunbalp(v)} < \card{\Nunbalm(v)}}$.}} Any vertex $u$ in $\Nunbalm(v)$ already contributes $\theta \cdot \deg(v)$ non-edges to the neighborhood of $v$ (simply because its degree is sufficiently small). 
		Hence in this case there are at least
		\begin{align*}
			\frac{1}{2} \cdot \card{\Nunbalm(v)} \cdot \Paren{\theta \cdot \deg(v)} \geq \frac{1}{2} \cdot \Paren{\frac{\theta \cdot \deg(v)}{4}} \cdot \Paren{\theta \cdot \deg(v)} > (\theta/2)^2 \cdot {{\deg(v)}\choose{2}}
		\end{align*}
		many non-edges in the neighborhood of $v$; hence $v$ is $(\theta/2)$-sparse in this case also. 
		
		\smallskip
		{\textbf{Case 2a: $\bm{\card{\Nunbalp(v)} \geq \card{\Nunbalm(v)}}$.}} In this case, we have at least 
		$(\theta/4) \cdot \deg(v)$ neighbors $u$ of $v$ such that $\deg(v) \leq (1-\theta) \cdot \deg(u) < (1-\theta/4) \cdot \deg(u)$, hence $v$ is $(\theta/4)$-uneven in this case. This concludes the proof. 
		\Qed{Claim~\ref{clm:decomp-not-dense}}
		
	\end{proof}
	
	\paragraph{Concluding the Proof of Lemma~\ref{lem:decomposition}.} We are now ready to finalize the proof of the decomposition. The general strategy is to let the connected components of $H_\theta$ be the almost-cliques and 
	then use Claim~\ref{clm:decomp-not-dense} to partition remaining vertices in $\Vsparse$ and $\Vunbal$ accordingly. The catch at this point is that
	Claim~\ref{clm:decomp-cc} does \emph{not} allow us to lower bound size of connected components of $H_{\theta}$ nor it bounds the number 
	of neighbors of vertices in a connected component to outside vertices in $G$ (only in $H_\theta$). We handle these using a similar approach as in~\cite{AssadiCK19}. 
	
	\begin{proof}[Proof of Lemma~\ref{lem:decomposition}]
		Let $\theta = 4\eps$. Consider the graph $H_\theta(D_\theta,F_\theta)$ defined earlier and let $C_1,\ldots,C_\ell$ be its connected components. 
		Let $K_1,\ldots,K_k$ be the components among these that contain at least one $\eps$-dense vertex. Moreover, define $U$ as the set of vertices in $V \setminus K_1 \cup \ldots \cup K_k$. 
		
		None of the vertices in $U$ are $\eps$-dense, hence by Claim~\ref{clm:decomp-not-dense}, we can decompose them into $\Vsparse$ consisting of $(\eps/2)$-sparse vertices 
		and $\Vunbal$ consisting of $(\eps/4)$-uneven vertices (breaking the ties between the two sets arbitrarily). Hence,  these two sets satisfy the requirements of the lemma. 
		
		We now show that for every $i \in [k]$, $K_i$ is an $\eps$-almost-clique according to Definition~\ref{def:almost-clique}. To do so, we prove the properties of 
		Definition~\ref{def:almost-clique} for $K_i$ one by one. 
		\begin{itemize}
		\item \textbf{Property~\ref{ac1}}: For any $v \in K_i$, by Claim~\ref{clm:decomp-cc}, $\deg(v) \geq (1-2\theta) \cdot \Delta(K_i) = (1-8\eps) \cdot \Delta(K_i)$. 
		
		\item \textbf{Property~\ref{ac2}}: By Claim~\ref{clm:decomp-cc}, $\card{K_i} \leq (1+2\theta) \cdot \Delta(K_i) = (1+8\eps) \cdot \Delta(K_i)$, hence we only need to prove the lower bound. Let $v$ be any $\eps$-dense vertex in $K_i$ and $F(v)$
		be the neighbors of $v$ that are $\eps$-friend with $v$ and thus $\card{F(v)} \geq  (1-\eps)\cdot \deg{(v)}$. At the same time, any vertex $u \in F(v)$ shares at least $(1-\eps) \cdot \min\set{\deg(v),\deg(u)} \geq (1-2\eps) \cdot \deg(v)$ 
		neighbors with $v$ 
		by definition of the $(u,v)$ being $\eps$-friend. As such, $u$ has at least $(1-3\eps) \cdot \deg{(v)}$ neighbors in $F(v)$. Moreover, because any two vertices in $F(v)$ share a common neighbor over their $\eps$-friend 
		edges (namely $v$), their degrees are within a factor $(1-2\eps)$ of each other. As such, any vertex in $F(v)$ has a $(4\eps)$-friend edge to at least $(1-3\eps) \cdot \deg{(v)}$ other vertices in $S_v$ (these edges are $(4\eps)$-friend and not $(3\eps)$ to account for the fact that degrees of vertices in $F(v)$ can be larger than $\deg(v)$ by (at most) $(1-\eps)^{-1}$ factor). This in particular 
		implies that all vertices in $F(v) \cup \set{v}$ are part of the same connected component $K_i$ in $H_{\theta} = H_{4\eps}$. Hence, $\card{K_i} \geq \card{F(v)} \geq (1-\eps) \cdot \Delta(K_i)$.
		
		\item \textbf{Property~\ref{ac3}}: By Claim~\ref{clm:decomp-cc}, any vertex $v \in K_i$ has at most $2\theta \cdot \Delta(K_i) = 8\eps \cdot \Delta(K_i)$ non-neighbors in $K_i$. 
		
		\item \textbf{Property~\ref{ac4}}: By combining the lower bound in Property~\ref{ac2} with Property~\ref{ac3}, we have $v \in K_i$ can only have $9\eps \cdot \Delta(K_i)$ neighbors outside of $C$. 
		\end{itemize}
		This concludes the proof of the lemma. \Qed{Lemma~\ref{lem:decomposition}}
		
	\end{proof}
	
\subsubsection{Proof of Theorem~\ref{thm:ps-deg+1-coloring} -- Part~\ref{part:deg+1-p2}}

For the rest of the proof, fix a decomposition of the graph $G(V,E)$ with some sufficiently small \textbf{absolute constant} ${\eps > 0}$ (taking $\eps = 10^{-4}$ would certainly suffice\footnote{In the interest of simplifying the exposition of the proof, we made no attempt in optimizing the constants in this section and instead chose the most straightforward values in every step. Our results continue to hold with much smaller constants.}). 
In the following, we show that we can handle both $\Vunbal$ and $\Vsparse$ vertices first, and then color the almost-cliques using a result of~\cite{AssadiCK19} almost in a black-box way. As such, the main difference
between our work and~\cite{AssadiCK19} (beside the decomposition) is in the treatment of vertices in $\Vunbal \cup \Vsparse$. 

Before we move on, we make an assumption (without loss of generality) that is used to make sure various concentration bounds in the proof hold. 

\newcommand{\Dmin}{\ensuremath{D_{\textnormal{\textsf{min}}}}}

\begin{assumption}\label{assumption}
We may and will assume that degree of every vertex is at least $\Dmin := \alpha \cdot \eps^{10} \cdot \log{n}$ for some sufficiently large \textbf{absolute constant} $\alpha > 0$. This is without loss of generality because by sampling 
$\Theta(\log{n})$ colors, any vertex with lower degree will have $L(v)=S(v)$ and hence we can greedily color these vertices after finding a proper coloring of the rest of the graph. 
\end{assumption}

\subsubsection*{Coloring Sparse and Unbalanced Vertices}

We prove the following lemma in this part. 
\begin{lemma}\label{lem:coloring-sparse-unbalanced}
	Suppose for every vertex $v \in \Vsparse \cup \Vunbal$, we sample a set $L(v)$ of $\Theta(\eps^{-6} \cdot \log{n})$ colors 
	independently and uniformly at random from $S(v) := \set{1,\ldots,\deg(v)+1}$. Then, with high probability, the induced subgraph $G[\Vsparse \cup \Vunbal]$ can be properly colored from 
	the sampled lists. 
\end{lemma}
We construct the coloring of Lemma~\ref{lem:coloring-sparse-unbalanced} in two steps. The first step is to create ``excess'' colors on vertices (reducing the problem essentially to $(1+o(1))\deg(v)$ coloring) and the second one
is to exploit these excess colors to color the vertices using an argument similar to Part~\ref{part:deg+1-p1} of Theorem~\ref{thm:ps-deg+1-coloring}. One important bit is that
the first step of this argument should be done \emph{simultaneously} for both $\Vunbal$ and $\Vsparse$.

For the proof of Lemma~\ref{lem:coloring-sparse-unbalanced}, we need to partition vertices in $\Vsparse$ and $\Vunbal$ further in order to be able to handle the disparity in degree of vertices. As such, we define: 

\newcommand{\Vsmall}{\ensuremath{V^{\textnormal{\textsf{small}}}}}
\newcommand{\SSmall}{\ensuremath{\textnormal{Small}}}
\newcommand{\dsmall}{\ensuremath{d_{\textnormal{\textsf{small}}}}}

\newcommand{\Vlarge}{\ensuremath{V^{\textnormal{\textsf{large}}}}}
\newcommand{\LLarge}{\ensuremath{\textnormal{Large}}}
\newcommand{\dlarge}{\ensuremath{d_{\textnormal{\textsf{large}}}}}

\begin{itemize}
	\item $\psi := \eps^2/32$: a parameter used throughout the definitions in this part for ease of notation. 
	\item\Vsmall: Let $\SSmall(v) := \set{u \in N(V): \deg(u) < \dsmall(v)}$ where $\dsmall(v) := \psi \cdot \deg(v)$. \\
	We define $\Vsmall \subseteq \Vsparse \cup \Vunbal$ as all vertices $v$ with $\card{\SSmall(v)} \geq 2\dsmall(v)$. 
	\item \Vlarge: Let $\LLarge(v) := \set{u \in N(v): \deg(u) > \dlarge(v)}$ where $\dlarge(v) := 2\deg(v)$. \\
	We define $\Vlarge \subseteq \Vsparse \cup \Vunbal$ as all vertices $v$ with $\card{\LLarge(v)} \geq \psi \cdot \deg(v)$.\footnote{We remark that the change in the place where $\psi$ used in the two definitions above is intentional and not a typo.}
\end{itemize}

As stated earlier, the goal of our first step is to construct excess colors for vertices. As it will become evident shortly, vertices in $\Vsmall$ actually
do not need require having excess colors to begin with (roughly speaking, after coloring their very ``low degree'' neighbors in $\SSmall(v)$, we are anyway left with many excess colors). 
Hence, we ignore these vertices in the first step altogether and handle them directly in the second one. Another important remark about the first step is that even though its goal is to color only $\Vsparse \cup \Vunbal$ (minus $\Vsmall$), 
we assume \emph{all} vertices of the graph (including almost-cliques) participate in its coloring procedure. This is only to simplify the math and after this step we simply \emph{uncolor} all vertices that are not in $\Vsparse \cup \Vunbal$.  

\paragraph{Creating Excess Colors.} We start with the following coloring procedure as our first step: 

\begin{tbox}
	$\FirstColor$: A procedure for finding a (partial) coloring of $G[\Vsparse \cup \Vunbal]$. 
\begin{enumerate}
	\item Iterate over vertices of $V$ in an arbitrary order. 
	\item For every vertex $v$, \textbf{activiate} $v$ w.p. $\pactive := \psi/16 ~~(=\Theta(\eps^2))$. 
	\item For every activated vertex $v$, pick a color $c_1(v)$ uniformly at random from $L(v)$ and if $c(v)$ is not used to
	color any neighbor of $v$ so far, \textbf{color} $v$ with $c_1(v)$. 
\end{enumerate}
\end{tbox}

We shall note right away that distribution of $c_1(v)$ for every vertex $v$ in $\FirstColor$ is simply uniform over $S(v)$. 
For any vertex $v \in V$, let $S_1(v)$ denote the list of available colors $S(v)$ after removing the colors assigned to neighbors of $v$ in this procedure. Similarly, let $\deg_1(v)$ denote the degree of $v$ after removing 
the colored neighbors of $v$ from the graph. We show that $S_1(v)$ is ``sufficiently larger'' than $\deg_1(v)$ for all vertices in $\Vsparse \cup \Vunbal \setminus \Vsmall$. Formally, 
\begin{lemma}\label{lem:excess-colors}
	There exists an {absolute constant} $\alpha \in (0,1)$ such that with high probability, for every $v \in \Vsparse \cup \Vunbal \setminus \Vsmall$, we have 
	$
	\card{S_1(v)} \geq \deg_1(v) + \alpha \cdot \eps^6 \cdot \deg(v).
	$ 
\end{lemma}

The proof of of this lemma is given in three parts, each for coloring one of the sets $\Vunbal$, $\Vlarge$ and $\Vsparse \setminus (\Vsmall \cup \Vlarge)$ separately. The first two have an almost identical proof and are based on a novel argument -- the third part uses a different argument which on a high level is similar to the approach of~\cite{AssadiCK19} (and~\cite{ElkinPS15,HarrisSS16,ChangLP18}, all rooted in an earlier work of~\cite{MolloyR97}) for coloring sparse vertices (according to a global definition of sparse based on $\Delta$), although several new challenges has to be addressed there as well. 
\begin{lemma}\label{lem:first-uneven}
	W.h.p. for every $v \in \Vunbal$ we have $\card{S_1(v)} \geq \deg_1(v) + \alpha \cdot \eps^4 \cdot \deg(v)$. 
\end{lemma}
\begin{proof}
	Let $\theta := (\eps/4)$ and recall that all vertices in $\Vunbal$ are $\theta$-uneven by Lemma~\ref{lem:decomposition}. 
	Fix a vertex $v$ in $\Vunbal$ and let $U(v)$ be the neighbors $u$ of $v$ where $\deg(v) < (1-\theta) \cdot \deg(u)$. As $v$ is $\theta$-uneven 
	$\card{U(v)} \geq \theta \cdot \deg(v)$. For any $u \in U(v)$, let $\Se(u) = S(u) \setminus S(v)$ denote the set of colors that are available (originally) to $u$ but not to $v$. 
	For $\se(u) := \card{\Se(u)}$, we have,
	\begin{align}
	\se(u) = \deg(u) - \deg(v) \geq \deg(u) - (1-\theta) \cdot \deg(u) = \theta \cdot \deg(u). \label{eq:uneven1}
	\end{align} 
	
	We say that a vertex $u \in U(v)$ is \emph{good} iff $u$ is colored from $\Se(u)$ by $\FirstColor$. Let $\ngood(v)$ denote the number of good neighbors of $v$. 
	It is easy to see that $\card{S_1(v)} \geq \deg_1(v) + \ngood(v)$ as colors of good vertices are not removed from $S(v)$. Our goal is to lower bound $\ngood(v)$ then. 

	Define the following two events: 
	\begin{itemize}
	\item $\eventact$: For every vertex $u \in V$, the number of active neighbors of $u$, denoted by $\adeg(u)$, is between $(\pactive/2) \cdot \deg(u)$ and $(2\pactive) \cdot \deg(u)$. 
	\item  $\eventact^U(v)$: The set $\Uactive(v)$ of \emph{active} vertices in $U(v)$ has size at least $(\pactive/2) \cdot \theta \cdot \deg(v)$. 
	\end{itemize}
	By our Assumption~\ref{assumption} and a simple application of Chernoff bound, both event $\eventact$ and $\eventact^U(v)$ hold with high probability (recall the lower bound on size of $U(v)$) above. 
	Note that both these events are only a function of the probability of activating each vertex and independent of choice of lists $L$. Hence, in the following we condition on these events (and all coins tosses for activation probabilities)
	and only consider the randomness with respect to choices in $L$. 
	
	Let $u_1,\ldots,u_k$ for $k := (\pactive/2) \cdot \theta \cdot \deg(v)$ be the first $k$ vertices in $\Uactive(v)$ according to the ordering of $\FirstColor$. Let $\RR(u_i)$ denote all the random choices 
	that govern whether $u_i$ will be good or not. Note that by the time we process $u_i$ at most $\adeg(u_i)$ colors from $S(u_i)$ may have been assigned to neighbors of $u_i$. Even if all of these colors
	are adversarially chosen to be in $\Se(u_i)$, the number of colors that if chosen by $u_i$ make $u_i$ a good vertex is at least: 
	\begin{align*}
	\se(u_i) - \adeg(u_i) \geq \theta \cdot \deg(u_i) - (2\pactive) \cdot \deg(u_i) > (\theta/2) \cdot \deg(u_i). \tag{by Eq~\eqref{eq:uneven1} and event $\eventact$, respectively and since $\pactive = \Theta(\eps^2) < \theta/4$}  
	\end{align*}
	Even conditioned on everything else, this choice is only a function of $c_1(u_i)$ chosen uniformly at random from $S(u_i)$. As such,  
	\begin{align*}
		\Pr\paren{\text{$u_i$ is good} \mid \RR(u_1),\ldots,\RR(u_{i-1})} \geq \frac{(\theta/2) \cdot \deg(u_i)}{\deg(u_i)+1} \geq (\theta/3). 
	\end{align*}
	This implies that $(i)$ $\expect{\ngood(v)} \geq (\theta/3) \cdot k$ and $(ii)$ the distribution of good vertices among first $k$ vertices in $\Uactive(v)$ stochastically dominates the binomial distribution $\mathcal{B}(k,\theta/3)$.  
	 By a basic concentration of binomial distributions (say by using Chernoff bound in Proposition~\ref{prop:chernoff}): 
	 \begin{align*}
	 	\Pr\paren{\ngood(v) < (\theta/6) \cdot k} \leq \exp\paren{-\Theta(1) \cdot \theta \cdot k } = \exp\paren{-\Theta(1) \cdot \eps^4 \cdot \log{n}} \ll n^{-10}. \tag{by the choice of $\theta= \Theta(\eps)$, $\pactive = \Theta(\eps^2)$, $k$, and Assumption~\ref{assumption}}
	 \end{align*}
	As $k = \Theta(\eps^3 \cdot \deg(v))$ and $\theta=\Theta(\eps)$, we obtain that
	 w.h.p. $\ngood(v) \geq \Theta(\eps^4) \cdot \deg(v)$. \Qed{Lemma~\ref{lem:first-uneven}}
	 
\end{proof}

\begin{lemma}\label{lem:first-large}
	W.h.p. for every $v \in \Vlarge$ we have $\card{S_1(v)} \geq \deg_1(v) + \alpha \cdot \eps^4 \cdot \deg(v)$. 
\end{lemma}
\begin{proof}
	Proof of this lemma is almost identical to that of Lemma~\ref{lem:first-uneven}. The reason is that since $v$ belongs to $\Vlarge$: 
	\begin{enumerate}[label=(\roman*)]
		\item $N(v)$ contains at least $(\eps^2/8) \cdot \deg(v)$ vertices  in $\LLarge(v)$ with degree $\geq \dlarge(v) = 2\deg(v)$; 
		\item Each vertex in $u \in \LLarge(v)$ have $\se(u) \geq \deg(v)$ for $\se(u)$ defined in Lemma~\ref{lem:first-uneven} to be the number of colors in $S(u) \setminus S(v)$. 
	\end{enumerate}
	As such, we can apply the same exact argument in Lemma~\ref{lem:first-uneven} to vertices in $\LLarge(v)$ (i.e., take $U(v)$ there to be $\LLarge(v)$) and bound the number
	of resulting good vertices. The proof now follows verbatim from the proof of Lemma~\ref{lem:first-uneven} and hence is omitted. We only note that even though size of $\LLarge(v)$ is smaller 
	by a factor $\Theta(\eps)$ here than $U(v)$ in the other lemma, size of $\se(u)$ for $u \in \LLarge(v)$ is a factor $\Theta(1/\eps)$ larger than than $\se(u) \in U(v)$ in there and thus
	we obtain the same exact bound up to constant factors. \Qed{Lemma~\ref{lem:first-large}}
	
\end{proof}

\newcommand{\Enonact}{\ensuremath{\Enon^{\textnormal{\textsf{active}}}}}
\newcommand{\eventnon}{\ensuremath{\event^{\textnormal{NE}}}}
\begin{lemma}\label{lem:first-sparse}
	Wh.p. for every $v \in \Vsparse \setminus (\Vsmall \cup \Vlarge)$ we have $\card{S_1(v)} \geq \deg_1(v) + \alpha \cdot \eps^6 \cdot \deg(v)$. 
\end{lemma}
\begin{proof}
	Let us define $\Enon(v)$ as the set of \emph{non-edge} in $N(v)$ between vertices $u$ and $w$ where neither $u$ nor $w$ belong to $\SSmall(v) \cup \LLarge(v)$, i.e., 
	\begin{align*}
		\Enon(v) := \set{u,w \in N(v) : (u,w) \notin E \wedge u \notin \SSmall(v) \cup \LLarge(v) \wedge w \notin \SSmall(v) \cup \LLarge(v)}. 
	\end{align*} 
	Define $\theta := (\eps/2)$. As $v$ is neither in $\Vsmall$ nor $\Vlarge$ but it is in $\Vsparse$ and hence is $\theta$-sparse by Lemma~\ref{lem:decomposition}, we have, 
	\begin{align}
		\card{\Enon(v)} &\geq \theta^2 \cdot {{\deg(v)}\choose{2}} - \card{\Vsmall} \cdot \deg(v) - \card{\Vlarge} \cdot \deg(v) \tag{each vertex $u \in \SSmall(v) \cup \LLarge(v)$ can only contribute $\deg(v)$ non-edges} \\
		&\geq \theta^2 \cdot {{\deg(v)}\choose{2}} - (4\theta^2/32) \cdot \deg(v)^2 - (4\theta^2/32) \cdot \deg(v)^2 \notag \\
		&> (\theta^2/3) \cdot \deg(v)^2 \label{eq:enon}. 
	\end{align}
\noindent
	Let $k := \card{\Enon(v)}$ and $f_1,\ldots,f_k$ denote these non-edges. We define the  random variable: 
	\begin{itemize}
		\item $Z$: number of colors in $S(v)$ that are sampled by \emph{at least} two activated vertices in $V(\Enon(v))$ and are additionally retained (i.e.,  assigned as a color to the vertex) by \emph{all} these activated neighbors. 
	\end{itemize}
	Since any color counted in $Z$ is used more than once to color a neighbor of $v$, we have, 
	\begin{align}
	\card{S_1(v)} \geq \deg_1(v) + Z. \label{eq:good-non}
	\end{align}
	We now lower bound the expectation of $Z$ and later on prove that it is concentrated. 
	
	\begin{claim}\label{clm:enon-expect}
		$\expect{Z} \geq \Theta(\eps^6) \cdot \deg(v)$. 
	\end{claim}
	\begin{proof}
	
	For every non-edge $f_i := (u_i,w_i)$, define the indicator random variable $Y_i$ where $Y_i = 1$ iff: 
	\begin{enumerate}[label=(\roman*)]
	\item $\event_1$: both $u_i$ and $w_i$ are activated and sample the same color $c_1(u_i) = c_1(w_i) \in S(v)$; 
	\item $\event_2$: color $c_1(u_i)=c_1(w_i)$ is not sampled by any active vertex in $N(u_i) \cup N(w_i)$;
	\item $\event_3$: color $c_1(u_i)=c_1(w_i)$ is not sampled by any active vertex in $V(\Enon(v)) \setminus \set{u_i,w_i}$.
	\end{enumerate}
	
	Define $Y = \sum_{i=1}^{k} Y_i$. Note that $Y \leq Z$ because in the definition of $Y$, we are counting number of colors sampled and retained by \emph{exactly two} neighbors of $v$ as opposed to \emph{at least two} neighbors in the definition of $Z$.  
	We can thus focus on lower bounding $Y$ instead. 

	Clearly, $\expect{Y_i} = \Pr\paren{\event_1 \wedge \event_2} \cdot \Pr\paren{\event_3 \mid \event_1,\event_2}.$ We compute each of these probabilities below. 
	By symmetry, let us assume that $\deg(u_i) \leq \deg(w_i)$. We define one more auxiliary event: 
	\begin{itemize}
		\item $\eventact(v,u_i,w_i)$: There are at most $4\pactive \deg(w_i)$ active vertices in $N(u_i) \cup N(w_i)$, and at most $2\pactive \deg(v)$ active vertices in $V(\Enon(v)) \setminus \set{u_i,w_i}$. 
	\end{itemize}
	As before, $\eventact(v,u_i,w_i)$ happens with high probability by Chernoff bound. We condition on this event and the activation coin flips of all vertices in $N(u_i) \cup N(w_i) \cup V(\Enon(v)) \setminus\set{u_i,w_i}$ (note that 
	we excluded $u_i,w_i$ from this conditioning).

	We now bound probability of $\event_1 \wedge \event_2$. Firstly, 
	\begin{align*}
		\dsmall(v) \geq \psi \cdot \deg(v) \geq \psi/2 \cdot \deg(w_i) \geq 8\pactive \cdot \deg(w_i). \tag{as $w_i \notin \LLarge(v)$, and by definition of $\pactive = \psi/16$}
	\end{align*}
	This implies that the number of colors in $S(u_i) \cap S(v) \subseteq S(w_i) \cap S(v)$ that have \emph{not} 
	been sampled by any active vertex in $N(u_i) \cup N(w_i)$ is at least (here we crucially use the fact that the underlying problem is $(\deg+1)$ coloring
	not $(\deg+1)$ \emph{list}-coloring): 
	\begin{align*}
		 \min\set{\deg(u_i),\deg(v)}-4\pactive \deg(w_i) \geq \min\set{\deg(u_i),\deg(v)}/2 \geq \deg(u_i)/4 \tag{by $\event(v,u_i,w_i)$ and because $\dsmall(v) \leq \deg(u_i) \leq 2\deg(v)$ and $\deg(v) > \dsmall(v)$}.
	\end{align*}
	 Clearly, if both $u_i$ and $w_i$ are activated and sample one of these colors (that belong to the lists of both of them), then 
	$\event_1 \wedge \event_2$ happens. 
	As such, 
	\begin{align*}
		\Pr\paren{\event_1 \wedge \event_2} \geq \pactive^2 \cdot \frac{\deg(u_i)/4}{\deg(u_i)+1} \cdot \frac{1}{\deg(w_i)+1} \geq \Theta(\eps^4) \cdot \frac{1}{\deg(v)}, \tag{$\deg(w_i) \leq 2\deg(v)$ and $\pactive = \Theta(\eps^2)$}
	\end{align*}
	
	To calculate $\event_3 \mid \event_1,\event_2$, we only need to bound the 
	probability of the event that each vertex $z \in V(\Enon(v)) \setminus (N(u_i) \cup N(w_i) \cup \set{u_i,w_i})$ samples the color $c$ (implied by events $\event_1,\event_2$). As the choice of vertices $z$
	are independent (and independent of the conditioned events), plus the fact that for every $z \in V(\Enon(v))$ we know that $\deg(z) \geq \dsmall(v)$, we have, 
	\begin{align*}
		 \Pr\paren{\event_3 \mid \event_1,\event_2} \geq \paren{1-\pactive \cdot\frac{1}{\dsmall(v)}}^{\deg(v)} \geq \exp\paren{-\pactive /2\psi}  = \Theta(1). \tag{$\pactive = \Theta(\psi)$}
	\end{align*}
	
	By the equations above, linearity of expectation, and Eq~\eqref{eq:enon} (and since $\theta=\eps/2$), 
	\begin{align*}
		\expect{Z} \geq \expect{Y} = \sum_{i} \expect{Y_i} \geq \card{\Enon(v)} \cdot \Theta(\eps^4) / {\deg(v)} \geq \Theta(\eps^6) \cdot \deg(v) 
	\end{align*}
	concluding the proof. \Qed{Claim~\ref{clm:enon-expect}}
	
	\end{proof}

	Let us now prove that $Z$ is concentrated which concludes the proof. The proof of this concentration is somewhat standard and appears in different forms (and with different techniques) in 
	several places, see, e.g.~\cite{MolloyR97,ElkinPS15,HarrisSS16,ChangLP18} (in particular~\cite[Lemma~2]{MolloyR97},~\cite[Lemma~3.1]{ElkinPS15},~\cite[Lemma~5.5]{HarrisSS16}, or~\cite[Lemma~3]{ChangLP18}). However, 
	as none of these results directly apply to our setting, we present this proof following the approach of~\cite[Chapter~10]{ColoringBook}. 
	
	For each vertex $u$ in the graph, let $\omega_u$ denote the random variable
	for the choice of activation coin and the random color sampled from $S(u)$ if $u$ is activated. 
	Notice that $Z$ is only a function of $\omega_u$ for $u \in N(v) \cup N(N(v))$ and by definition, these variables are independent of each other. To apply Talagrand's inequality, we need to show that $Z$ is $c$-Lipschitz and $r$-certifiable in these variables
	for some (ideally) small $c$ and $r$	(see Proposition~\ref{prop:talagrand} and its preceding paragraph for these definitions). 
	Unfortunately, this is in fact not the case for $Z$ (in particular, $Z$ may only be $\Omega(\deg(v))$-certifiable because for every color counted in $Z$, we need to reveal $w_u$ for $\Omega(\deg(v))$ vertices to ensure this color is retained; this is too large to 
	apply Talagrand's inequality directly.)
	
	We thus bound $Z$ indirectly as follows. Define the two additional variables: 
	\begin{itemize}
		\item $T$: number of colors in $S(v)$ that are sampled by \emph{at least} two neighbors of $v$ in $V(\Enon(v))$. 
		\item $D$: number of colors in $S(v)$ that are sampled by \emph{at least} two neighbors of $v$ but are \emph{not} retained by \emph{at least one} of them. 
	\end{itemize}
	
	Firstly, it is clear that $Z = T - D$. Also notice that both $T$ and $D$ are functions of $\omega_u$ for $u \in N(v) \cup N(N(v))$. Moreover, unlike $Z$, both $T$ and $D$ are $\Theta(1)$-certifiable (for $T$ point to two
	neighbors of $v$ that sampled the color; for $D$ additionally point to one of the neighbors of this pair that also
	sampled the color, hence not allowing one of them to retain it). They are also both $\Theta(1)$-Lipschitz: changing choice of one color for a vertex can only affect the two colors involved (the original one and the changed one). 
	As such, we can apply Talagrand's inequality (Proposition~\ref{prop:talagrand}) to obtain the desired bounds as follows.

	We first prove the bound for $T$. By bounding the total number of colors sampled in the neighborhood of $v$, it is easy to verify that, 
	\begin{align}
		\expect{T} \leq \deg(v) \cdot \pactive^2 = \Theta(\eps^4) \cdot \deg(v), \label{eq:T-expect}
	\end{align}
	as $\pactive = \Theta(\psi) = \Theta(\eps^2)$. 
	Moreover, $T$ is  both $\Theta(1)$-Lipschitz and $\Theta(1)$-certifiable as argued above. As such, by Talagrand's inequality (Proposition~\ref{prop:talagrand}): 
	\begin{align*}
		\Pr\paren{\card{T-\expect{T}} \geq \expect{Z}/100} &\leq \exp\paren{-\Theta(1) \cdot \frac{(\expect{Z}/100 - \Theta(1)\sqrt{\expect{T}})^2}{\expect{T}}} \\
		&\leq \exp\paren{-\Theta(1) \cdot \frac{\expect{Z}^2}{\expect{T}}} \tag{as $\expect{Z} > \sqrt{\expect{T}}/2$ by Claim~\ref{clm:enon-expect}, Eq~\eqref{eq:T-expect}, and Assumption~\ref{assumption}} \\
		&\leq \exp\paren{-\Theta(\eps^4) \cdot \expect{Z}} \tag{as $\expect{Z} \geq \Theta(\eps^4) \cdot \expect{T}$ by Claim~\ref{clm:enon-expect} and Eq~\eqref{eq:T-expect}}  \\
		&\leq \exp\paren{\Theta(\eps^{10}) \cdot \deg(v)} \tag{by Claim~\ref{clm:enon-expect}} \\
		&\ll n^{-4} \tag{by Assumption~\ref{assumption}}.
	\end{align*}
	
	We now focus on $D$. By applying Talagrand's inequality again (Proposition~\ref{prop:talagrand}):
	\begin{align*}
		\Pr\paren{\card{D-\expect{D}} \geq \expect{Z}/100} &\leq \exp\paren{-\Theta(1) \cdot \frac{(\expect{Z}/100 - \Theta(1)\sqrt{\expect{D}})^2}{\expect{D}}} \ll n^{-4},
	\end{align*}
	by exactly the same calculation as above since $\expect{D} \leq \expect{T}$ (as $D \leq T$). Combining the above two equations implies that with high probability, 
	\begin{align*}
		Z = T - D \geq (\expect{T} - \expect{Z}/100) - (\expect{D} + \expect{Z}/100) \geq (49/50) \cdot \expect{Z}.
	\end{align*}
	Plugging in this bound in Eq~\eqref{eq:good-non} concludes the proof. \Qed{Lemma~\ref{lem:first-sparse}}
	
\end{proof}

Lemma~\ref{lem:excess-colors} now follows directly from Lemmas~\ref{lem:first-uneven}, Lemma~\ref{lem:first-large} and~\ref{lem:first-sparse} and a union bound. 

\paragraph{Exploiting Excess Colors.} For the second step, consider the following procedure: 
\begin{tbox}
	$\SecondColor$: A procedure for finishing the proper coloring of $G[\Vsparse \cup \Vunbal]$. 
	\begin{enumerate}
	\item Iterate over uncolored vertices $v \in \Vsparse \cup \Vunbal$ in an \emph{arbitrary} order and for each vertex $v$, let $N^{<}(v)$ denote the neighbors of $v$ that appear before $v$ in this ordering \emph{plus} all neighbors
	of $v$ that have been colored in the first step. 
	\item\label{line:abort} For each vertex $v$, if there exists a color in $L(v)$ that is not used to color any vertex $u \in N^{<}(v)$, color $v$ with this color. Otherwise \abort. 
\end{enumerate}
\end{tbox}
It is immediate that if $\SecondColor$ does not \abort, we find a proper coloring  using the sampled colors in lists $L$. We now prove that \abort happens with only a small probability. 

\newcommand{\eventabort}{\event_{\textnormal{\textsf{abort}}}}

\begin{lemma}\label{lem:abort-second}
	W.h.p. $\SecondColor$ does \emph{not}  {\abort}. 
\end{lemma}
\begin{proof}
	Recall that $\alpha \in (0,1)$ is the constant in Lemma~\ref{lem:excess-colors}. Let $\ell := (\frac{10}{\alpha \cdot \eps^6} \cdot \log{n}+1)$ and suppose size of each list $L(v)$ is at least $\ell$ (which is $\Theta(\log{n})$ as both $\alpha,\eps \in \Theta(1)$). 
	Define the event: 
	\begin{itemize}
		\item $\eventabort(v)$: $L(v)$ is a subset of colors assigned to $N^{<}(v)$. 
	\end{itemize}
	We prove that $\Pr\paren{\eventabort(v)} \leq n^{-4}$; a union bound finalizes the proof as $\SecondColor$ would \abort only if at least one of the events $\eventabort(v)$ happens. 
	
	Suppose first $v$ belongs to $\Vsparse \cup \Vunbal \setminus \Vsmall$ (but not colored in the first step). Recall that at the beginning of this step, the list of available colors to $v$ is $S_1(v)$ and 
	$\deg_1(v)$ denotes the degree of $v$ to remaining uncolored vertices. By the time it is turn to color $v$ in $\SecondColor$, at most $\deg_1(v)$ other colors have been removed from available colors $S_1(v)$. As such, 
	\begin{align*}
		\Pr\paren{\text{$\eventabort(v)$}} &\leq \paren{\frac{\card{S(v)} - (\card{S_1(v)} - \deg_1(v))}{\card{S(v)}}}^{\ell-1} \leq \paren{1-\frac{\alpha \cdot \eps^6 \cdot \deg(v)}{\deg(v)+1}}^{\ell-1} 
		\tag{by Lemma~\ref{lem:excess-colors} and since $\card{S(v)} = \deg(v)+1$} \\
		&\leq \exp\paren{-\alpha \cdot \eps^6 \cdot \frac{10}{\alpha \cdot \eps^6} \cdot \log{n}} \ll n^{-4}. \tag{by the choice of $\ell$}
	\end{align*} 
	
	Now suppose  $v$ belongs to $\Vsmall$ instead. By definition, in this case $v$ has at least $2\dsmall(v)$ neighbors with degree $<\dsmall(v)$. For each such neighbor $u$, 
	$S(u) = \set{1,\ldots,\deg(u)+1}$ originally. As such, even if we have colored all neighbors of $v$ by the time we want to process $v$, there are at most $\deg(v) - 2\dsmall(v) + \dsmall(v) = (1-\eps^2/32)\deg(v)$ distinct colors that 
	have  appeared in the neighborhood of $v$. As such, 
	\begin{align*}
		\Pr\paren{\text{$\eventabort(v)$}} &\leq \paren{\frac{(1-\eps^2/32)\deg(v)}{\card{S(v)}}}^{\ell-1} \leq \exp\paren{- (\eps^2/32) \cdot \frac{10}{\alpha \cdot \eps^6} \cdot \log{n}} \ll n^{-4}. \tag{by the choice of $\ell$ and since $\card{S(v)} = \deg(v)+1$}
	\end{align*} 
	This concludes the proof. \Qed{Lemma~\ref{lem:abort-second}}
	
\end{proof}

Lemma~\ref{lem:coloring-sparse-unbalanced} now follows from Lemmas~\ref{lem:excess-colors} and~\ref{lem:abort-second} and a union bound.

\newcommand{\outdeg}{\ensuremath{\textnormal{out-deg}}}

\subsubsection*{Coloring Almost-Cliques} 

We are now left with the coloring of almost-cliques from the sampled lists after fixing the colors of remaining vertices. This is done by the following lemma. 
We note that this lemma is a simple generalization of a result of~\cite{AssadiCK19} for $(\Delta+1)$ coloring (see Lemma~\ref{lem:almost-clique-color} in Appendix~\ref{sec:background-ACK19}) and 
the proof is via a simple ``reduction'' to the proof of the analogous lemma for $(\Delta+1)$ coloring; hence, we claim no novelty for the proof of this lemma. 

Recall the definition of an $\eps$-almost-cliques $K$ in Definition~\ref{def:almost-clique}. For a vertex $v \in K$, we define $\outdeg(v)$ as the number of neighbors of $v$ that are outside $K$. 
Note that by definition of $\eps$-almost-cliques, $\outdeg(v) \leq 9\eps \cdot \Delta(K)$. 

\begin{lemma}\label{lem:degree-almost-clique}
	Let $K$ be an $\eps$-almost-clique in $G$ according to Definition~\ref{def:almost-clique} for some sufficiently small $\eps > 0$ and define $\Delta(K) :=\max_{v \in K} \deg(v)$. 
	Suppose for every vertex $v \in K$, we \emph{\underline{adversarially}} pick a set $\barS(v)$ of size at most $\outdeg(v) \leq 9\eps \cdot \Delta(K)$ from colors $\set{1,\ldots,\deg(v)+1}$. 
	If for every vertex $v \in V$, we sample a set $L(v)$ of $\Theta(\eps^{-1} \cdot \log{n})$ colors independently from the set of colors $\set{1,\ldots,\deg(v)+1}$, then, with high probability, the induced subgraph $G[K]$ can be properly colored
	from the lists $L(v) \setminus \barS(v)$ for $v \in C$. 
\end{lemma}
\begin{proof}
	Fix an $\eps$-almost-clique $K$ in $G$. For every vertex $v \in K$, we define $S'(v)$ to be $\barS(v)$ plus the colors $\set{\deg(v)+2,\ldots,\Delta(K)}$. Consider the graph $G'$ consisting of the $\eps$-almost-clique $K$ and 
	additionally for each $v \in K$, $\outdeg(v) + (\Delta(K) - \deg(v))$ dummy vertices that are only connected to $v$. For every vertex $v \in K$, define the set $S'(v) := \barS(v) \cup \set{\deg(v)+2,\ldots,\Delta(K)+1}$: we can think of this 
	as coloring $\outdeg(v)$ dummy vertices incident on $v$ by $S(v)$ and $(\Delta(K) - \deg(v))$ dummy vertices incident on $v$ by the ``new colors'' for $v$ (due to the increase in its degree), thus effectively canceling the contribution of these new colors for $v$. 
	
	Note that if we can find a coloring of $K$ in $G'$ in a scenario where every vertex samples a list of colors $L'(v)$ from $\set{1,\ldots,\Delta(K)+1}$ (as opposed to $\set{1,\ldots,\deg(v)+1}$ for $L(v)$), 
	and then coloring each vertex from $L'(v) \setminus S'(v)$ we will be done -- this is because the color used for coloring $v$ should still belong to $\set{1,\ldots,\deg(v)+1} \cap L(v) \setminus \barS(v)$ 
	as all the colors in $L'(v) \setminus L(v)$ belong to $S'(v)$. 
	
	The final observation here is that in the graph $G'$, $\Delta := \Delta(G') = \Delta(K)$ and so we have: 
	\begin{enumerate}[label=(\roman*)]
		\item we claim that $K$ in $G'$ is a $(\Delta,\eps')$-almost clique according to definition of Lemma~\ref{lem:extended-HSS-decomposition} of~\cite{AssadiCK19} for some $\eps'$ which is
		larger than $\eps$ by some constant factor ($\eps' = 20\eps$ certainly suffice): the only property of $(\Delta,\eps')$-almost clique that one needs to worry is the number of neighbors of each
		 vertex in $K$ to outside $K$ (as we increased it by adding some new dummy vertices). However, 
		this is not problematic because $\outdeg(v) \leq 9\eps \cdot \Delta(K)$ and $\Delta(K)-\deg(v) \leq 8\eps\Delta(K)$ and hence each vertex in $K$ has at most $17\eps\cdot\Delta(K)$ out degree in $G'$, which is smaller than $\eps'\Delta$.  
		For the remaining parameters $ (1-\eps') \cdot \Delta \leq \card{K} \leq (1+\eps') \cdot \Delta$ and number of non-neighbors inside is at most $8\eps\Delta(K) \leq \eps' \Delta$. Thus, $K$ is indeed a $(\Delta,\eps')$-almost-clique. 
		
		\item We still placed at most $\outdeg_{G'}(v)$ in the lists of colors $S'(v)$ that are ``blocked''; 
		\item $\eps'$ is still a sufficiently small constant (by taking $\eps$ to be small enough);
		\item We can ``simulate'' the sampling of colors $L'(v)$ from $\set{1,\ldots,\Delta(K)+1}$ by sampling $L(v)$ from $\set{1,\ldots,\deg(v)+1}$ (i.e., use the given colors in the lemma statement for $v$) and 
		sampling from $\set{\deg(v)+2,\ldots,\Delta(K)+1}$ separately (i.e., picking some ``artificial'' colors for $v$); as the latter
		colors cannot be assigned to $v$ anyway, this does not make a problem. 
	\end{enumerate}
	Hence, can apply Lemma~\ref{lem:almost-clique-color} (of~\cite{AssadiCK19}) to $K$ in $G'$ and obtain the coloring of $K$ in $G$.  \Qed{Lemma~\ref{lem:degree-almost-clique}}
	
\end{proof}

\subsubsection*{Concluding the Proof} 

\begin{proof}[Proof of Theorem~\ref{thm:ps-deg+1-coloring} -- Part~\ref{part:deg+1-p2}]
	We fix a decomposition of the graph $G$ according to Lemma~\ref{lem:decomposition} for some sufficiently small absolute {constant} $\eps > 0$ (taking $\eps = 10^{-4}$ would certainly suffice). 
	Lemma~\ref{lem:coloring-sparse-unbalanced} allows us to argue that with high probability, all vertices except for almost-cliques in the decomposition can be properly colored using the sampled lists. 
	We fix such a coloring of those vertices. We then iterate over almost-cliques one by one, and invoke Lemma~\ref{lem:degree-almost-clique} to each almost-clique $K_i$ by letting $\barS(v)$ for every $v \in K_i$ to be 
	the set of colors used so far in this process for coloring neighbors of $v$ outside this almost-clique. This allows us to color this almost-clique in a way that its coloring can be extended to the partial coloring computed 
	so far (with high probability). Iterating over all almost-cliques this way and using a union bound finalizes the proof.
\end{proof}



%% file: sublinear.tex

\newcommand{\Econf}{\ensuremath{E_\textnormal{\textsf{conflict}}}}
\newcommand{\Gconf}{\ensuremath{G_\textnormal{\textsf{conflict}}}}

\section{Sublinear Algorithms from Palette Sparsification }\label{sec:sublinear}

In this section, we describe some applications of our palette sparsification theorems to sublinear algorithms following the work of~\cite{AssadiCK19}. In the following, we give the definition of each of the two models of streaming algorithms and sublinear-time algorithms formally, followed by the resulting algorithms from palette sparsification for each one separately. 

\subsection{Streaming Algorithms} 
In the streaming model, edges of the graph are presented one by one to an algorithm that can make one or a few passes over the input and use a limited memory to process the stream and has to output
the answer at the end of the last pass. In this paper, we only consider \emph{single-pass} streaming algorithms. 
We can obtain the following algorithms from Results~\ref{res:od},~\ref{res:triangle-free}, and~\ref{res:deg+1}. 

\begin{corollary}\label{cor:ps-stream}
	There exists randomized single-pass streaming algorithms for finding each of the following colorings with high probability:
	\begin{itemize}
		\item a $(1+\eps)\Delta$ coloring of any general graph with $O_{\eps}(n\log{n})$ space;
		\item an $O(\frac{\Delta}{\gamma \cdot \ln{\Delta}})$ coloring of any triangle-free graph with $\Ot(n\cdot{\Delta}^{2\gamma})$ space;
		\item a $(1+\eps)\deg$-list coloring of any general graph with $O_{\eps}(n \cdot \log^2{n})$ space;  
		\item a $(\deg+1)$ coloring of any general graph with $O(n \cdot \log^2{n})$ space. 
	\end{itemize} 
\end{corollary}

The streaming algorithms in Corollary~\ref{cor:ps-stream} are basically as follows: we sample the colors in $L$ at the beginning of the stream and throughout the stream whenever an edge $(u,v)$ is presented, 
we check whether $L(u) \cap L(v) = \emptyset$ or not; if not we store this edge explicitly. At this point, obtaining the first two algorithms in Corollary~\ref{cor:ps-stream} from Results~\ref{res:od} and~\ref{res:triangle-free} is
straightforward (see also~\cite{AssadiCK19}). However, the results for the latter two parts does not immediately follow from the
argument for other two (or the one in~\cite{AssadiCK19}). This is due to the fact that both $(1+\eps)\deg$ and $(\deg+1)$ problems are ``local'' problems with dependence on $\deg$ instead of $\Delta$.  

To show that the above strategy still works even for these local coloring problems, we only need to show that the total number of edges stored by the algorithm is not ``too large''. This is equivalent 
to bounding the number of edges in the \emph{conflict-graph} $\Gconf(V,\Econf)$ where $\Econf := \set{(u,v) \in E : L(u) \cap L(v) \neq \emptyset}$. This is done in the following lemma.
We prove this result for $(\deg+1)$ coloring problem; the proof can be extended to $(1+\eps)\deg$ problem verbatim. We note that in the following we assume we know $\deg(v)$ of each vertex
beforehand (so that we can sample the needed colors from $S(v)$). This assumption is \emph{not} needed and we show how to remove it in Lemma~\ref{lem:sub-time} and Remark~\ref{rem:know-Delta}. 

\begin{lemma}\label{lem:stream-conf}
	W.h.p. the total number of edges in $\Econf$ in palette sparsification for $(\deg+1)$ coloring problem is at most $O(n\cdot\log^2{n})$. 
\end{lemma}
\begin{proof}
	In $(\Delta+1)$ coloring, we can simply show that maximum degree of $\Gconf(V,\Econf)$ is at most $O(\log^2{n})$. This is no longer true for $(\deg+1)$ -- consider the center of an induced star with $\Theta(n)$ petals. 
	We fix this issue as follows. Let us orient the edges $E$ of $G$ from lower degree endpoint to the higher degree one (breaking the ties arbitrarily). Let $\deg^+_G(v)$ denote the out-degree of $v$ in $G$ under this orientation. 
	We show that even though $\deg_{\Gconf}(v)$ can be too large, $\deg^+_{\Gconf}(v)$ is still $O(\log^2{n})$ for every $v$ with high probability. 
	
	Consider any vertex $u$ which is counted toward $\deg^+_G(v)$, i.e., in the orientation, $v$ has an outgoing edge to $u$. Since $\deg_G(u) \geq \deg_G(v)$ the probability that $u$ samples one of the $O(\log{n})$ colors 
	in $L(v)$ is at most $O(\log^2{n})/\deg(u) \geq O(\log^{2}(n))/\deg(v)$. As such, $\expect{\deg^+_{\Gconf}(v)} = O(\log^{2}{n})$. By Chernoff bound, we have that $\deg^+_{\Gconf(v)}$ is also $O(\log^{2}{n})$. As every edge of $\Gconf$ 
	is counted exactly once in $\deg^+_{\Gconf}(\cdot)$ across all vertices, we obtain that $\card{\Econf} = O(n \cdot \log^{2}{n})$.
\end{proof}
It is now easy to see that the last two parts of Corollary~\ref{cor:ps-stream} also follow from Result~\ref{res:deg+1}. 

We conclude this part by noting that our results can be extended to dynamic streams where edges can be both inserted to and deleted from the stream by increasing the space of the algorithm with $\polylog(n)$ factors as was done in~\cite{AssadiCK19}.  

\subsection{Sublinear-Time Algorithms} 
When designing sublinear-time algorithms, it is crucial to specify the data model as the algorithm cannot even read the entire input once. 
We assume the standard query model for sublinear-time algorithms on general graphs (see, e.g.,~\cite[Chapter 10]{Goldreich17}).
In this model, we have the following three types of queries $(i)$ what is the degree of a vertex $v$; $(ii)$ what is the $i$-th neighbor of a given vertex $v$; and $(iii)$ whether a given pair of vertices $(u,v)$ are neighbor to each other or not. 
We say an algorithm is \emph{non-adaptive} if it asks all its queries in parallel in one go. 

We can obtain the following algorithms from Results~\ref{res:od},~\ref{res:triangle-free}, and~\ref{res:deg+1}. 

\begin{corollary}\label{cor:sub-time}
	There exists randomized non-adaptive sublinear-time algorithms for finding each of the following colorings with high probability:
	\begin{itemize}
		\item a $(1+\eps)\Delta$ coloring of any general graph in $\Ot_{\eps}(n^{3/2})$ time;
		\item an $O(\frac{\Delta}{\gamma \cdot \ln{\Delta}})$ coloring of any triangle-free graph in $\Ot(n^{3/2+2\gamma})$ time;
		\item a $(1+\eps)\deg$-list coloring of any general graph in $\Ot(n^{3/2})$ time;  
		\item a $(\deg+1)$ coloring of any general graph in $\Ot(n^{3/2})$ time.
	\end{itemize} 
\end{corollary}

The sublinear-time algorithms in Corollary~\ref{cor:sub-time} are again based on finding the edges of the conflict-graph $\Econf$ using $\Ot(\min\set{n\Delta,n^{2}/\Delta})$ queries for the case of $(1+\eps)\Delta$ coloring and
 $\Ot(\min\set{n\Delta,n^{2}/\Delta^{1-2\gamma}})$ queries for triangle-free graphs. This can be done using the simple approach of~\cite{AssadiCK19} but as before that does not work for the last two parts. Here, we give another simple
 way for finding edges of the conflict-graph using a small number of queries. We again only prove it for $(\deg+1)$ coloring problem; the same argument extends to other problems as well. 
 
\begin{lemma}\label{lem:sub-time}
	W.h.p. all edges in $\Econf$ can be found using $\Ot(n^{3/2})$ queries non-adaptively. 
\end{lemma}
\begin{proof}
	 Define $t:= O(\log{n})$  ``potential'' palettes $P_1,\ldots,P_t$ where for every $i \in [t]$, $P_i := \set{1,\ldots,2^{i}}$. 
	 Let $\ell = \Theta(\log{n})$ denote the number of sampled colors in the palette sparsification theorem for $(\deg+1)$ coloring problem. 
	 For every vertex $v \in V$, we sample $t$ ``potential'' lists $\hL_1(v),\ldots,\hL_t(v)$ where each $\hL_i(v)$ is obtained by sampling each color in $P_i$ with probability $10\ell/\card{P_i}$. Note that all this has been done
	 without querying the graph yet. 
	 
	 We now make the following queries non-adaptively for every vertex $v \in V$: 
	 \begin{enumerate}[label=(\roman*)]
	 	\item We make a single degree-query on $v$;
		\item We make $10\sqrt{n}$ neighbor-queries on $v$ to return  $\min\set{\deg(v),10\sqrt{n}}$ neighbors of $v$;
		\item For every $i,j$ where $\card{P_i} \geq \sqrt{n}$ and $\card{P_j} \geq \sqrt{n}$, we make a pair query between $(v,u)$ whenever  $\hL_i(v) \cap \hL_j(u) \neq \emptyset$. 
		A simple application of Chernoff bound ensures that in this case also we make at most $\Ot(\sqrt{n})$ queries as size of both $P_i,P_j$ is at least $\sqrt{n}$.  
	 \end{enumerate}
	Overall with high probability we made at most $\Ot(n\sqrt{n})$ queries. 
	
	After getting the answer to those queries, we know $\deg(v)$ for every $v \in V$. We then pick the smallest integer $i$ and $P_i$ with $\card{P_i} \geq \deg(v)$, and consider $L(v) := \hL_i(v) \cap P_i \setminus \set{\deg(v)+2,\ldots,\card{P_i}}$. Again, by Chernoff bound, 
	size of each $L(v)$ is at least $\ell$ as $\deg(v)$ and $\card{P_i}$ differ from each other by at most a factor of $2$ and by the construction of $\hL_i(v)$. This way, we obtain $\ell$ colors $L(v)$ chosen uniformly at random from 
	$\set{1,\ldots,\deg(v)+1}$. These lists define $\Econf$ uniquely. 
	
	Finally, any edge $(u,v) \in \Econf$, if either $\deg(u) < 10\sqrt{n}$ or $\deg(v) < 10\sqrt{n}$ we have found this edge using the neighbor queries for the lower degree vertex in item $(ii)$. On the other hand, if both vertices have degree larger than
	$4\sqrt{n}$ then we will find this edge using the pair queries in item $(iii)$. This concludes the proof. 
\end{proof}

So far, we only analyzed the query complexity of the algorithms and ignored the runtime needed to compute the list-coloring of the conflict-graph. It is easy to see that all our proofs also imply an efficient algorithm for finding the coloring 
in time linear in the size of the conflict-graph (when needed, we can run algorithmic variants of Lov\'asz Local Lemma using the Moser-Tardos framework~\cite{MoserT10}). The only exception is for $(\deg+1)$ coloring problem
when we invoke the result of~\cite{AssadiCK19}; for that particular instance the runtime of the algorithms is $\Ot(n\sqrt{n})$ (as shown in~\cite{AssadiCK19}) even though the conflict graph is sparser. 

It is now easy to see that all items in Corollary~\ref{cor:sub-time} follow from Lemma~\ref{lem:sub-time} and Results~\ref{res:od},~\ref{res:triangle-free}, and~\ref{res:deg+1} (we remark that for $(1+\eps)\deg)$-list coloring our sublinear time algorithm works 
even without having direct access to the list $S(v)$ as long as it can  be sampled). 
 
 \subsection{Further Remarks}\label{sec:sub-remarks}
 
 We conclude this section by the following remarks. These remarks also apply the same exact way to our algorithms in Section~\ref{sec:vertex-sampling}. 
 
 \begin{remark}[\textbf{Knowledge of $\bm{\Delta}$}]\label{rem:know-Delta}
	{\textbf{We do not require a prior knowledge of $\bm{\Delta}$}. As was shown already in Lemma~\ref{lem:sub-time}, there is a simple ``guessing'' mechanism for easily working with unknown values of $\Delta$ (which is more crucial for the local versions),
	and whenever needed we can run that approach at a cost of increasing the complexity of the algorithms by a $\polylog{(n)}$ factor. We note that this is not new to our paper and also holds for previous work in~\cite{AssadiCK19,BeraCG19}.} 
\end{remark}

\begin{remark}[\textbf{Deterministic Guarantee on Resource Requirements}]
{The resource requirement of our algorithms, as stated, is bounded with high probability but not deterministically. However, this is easy to fix by a standard argument: whenever the resources used by the algorithm exceed the bound implied by the high-probability-result, simply terminate the whole algorithm -- this can only increase the error probability by a negligible factor. As such, there is a \textbf{deterministic guarantee on the resource requirement of algorithms} in this paper.}
\end{remark}

%% file: vertex-sampling.tex

\newcommand{\colf}{\ensuremath{{\zeta}}}
\newcommand{\colftri}{\ensuremath{{\zeta_{\textnormal{\textsf{tri-free}}}}}}

\section{Sublinear Algorithms from Graph Partitioning}\label{sec:vertex-sampling}

In this section, we deviate from our theme of palette sparsification and consider another technique for designing sublinear algorithms for graph coloring.
A simple technique that lies at the core of various algorithms for graph coloring in different models is \emph{random graph partitioning} (see, e.g.~\cite{Parter18,ParterS18,HarveyLL18,ChangFGUZ18,BeraCG19}).  
While the exact implementation of this technique varies significantly from one application to another, the basic idea is as follows: Partition the vertices of the graph $G$ randomly into multiple parts $V_1,\ldots,V_k$, then color the induced subgraphs
$G[V_1],\ldots,G[V_k]$ separately using disjoint palettes of colors for each subgraph. The hope is that each subgraph $G[V_i]$ has become ``simpler enough'' so that it can be colored ``easily'' with a ``small'' palette of colors so that using disjoint palette for each subgraph would not be too wasteful. 

We apply the same basic idea in this section. To state our result, we need some definitions first. 
We say that a family $\FG$ of graphs is \emph{hereditary} iff for every $G \in \FG$, every induced subgraph of $G$ also belongs to $\FG$, namely, $\FG$ is closed under vertex deletions. 

\begin{definition}\label{def:colorable-family}
	Let $\FG$ be a hereditary family of graphs and $\colf: \IN^+ \rightarrow \IN^+$ be a non-decreasing function. We say that $\FG$ is \textbf{$\bm{\colf}$-colorable} iff every graph $G$ in $\FG$ is $\colf(\Delta)$-colorable, where $\Delta:= \Delta(G)$ denotes
	the maximum degree of $G$. 
\end{definition}

For instance, the family of all graphs is an $\colf$-colorable family for the function $\colf(\Delta) = \Delta+1$, and triangle-free graphs are $\colf$-colorable for $\colf(\Delta) = O(\frac{\Delta}{\ln{\Delta}})$. 


\begin{theorem}\label{thm:vertex-sampling}
	Let $\FG$ be a $\colf$-colorable family of graphs (see Definition~\ref{def:colorable-family}) and $G(V,E)$ be an $n$-vertex graph with maximum degree $\Delta$ in $\FG$. 
	For the parameters 
	\[
		\eps>0, ~\qquad ~ 1 \leq k \leq \frac{\eps^2\cdot\Delta}{9\ln{n}}, ~\qquad~ C := C(\eps,k) = k \cdot \colf\Paren{(1+\eps)\cdot\frac{\Delta}{k}},
	\]
	suppose we partition $V$ into $k$ sets $V_1,\ldots,V_k$ uniformly at random; then with high probability $G$ can be $C$-colored by coloring each $G[V_i]$ with a distinct palette of size $C/k$. 
\end{theorem}

The proof of this theorem is by simply showing that the maximum degree of each graph $G[V_i]$ is sufficiently small, itself a simple application of Chernoff bound. 

\begin{lemma}\label{lem:sample-degree}
	The maximum degree of any $G[V_i]$ is at most $(1+\eps)\cdot\frac{\Delta}{k}$ with high probability. 
\end{lemma}
\begin{proof}
	For any vertex $v \in V_i$, let $\deg_i(v)$ denote the number of neighbors of $v$ in $G[V_i]$. Clearly, 
	$\expect{\deg_i(v)} = \frac{1}{k} \cdot \deg(v) \leq \frac{\Delta}{k}$. As the choice of neighbors of $v$ in $V_i$ are independent, by Chernoff bound (Proposition~\ref{prop:chernoff} with $\mu = \frac{1}{k} \cdot \Delta$ and $\delta = \eps$), 
	\begin{align*}
		\Pr\paren{\deg_i(v) \geq (1+\eps)\frac{\Delta}{k}} \leq \exp\paren{-\eps^2 \cdot \frac{\Delta}{3k}} = 1/n^3. \tag{as $k \leq \frac{\eps^2\cdot\Delta}{9\ln{n}}$}
	\end{align*}
	A union bound on all $n$ vertices finalizes the proof. \Qed{Lemma~\ref{lem:sample-degree}}
	
\end{proof}

\begin{proof}[Proof of Theorem~\ref{thm:vertex-sampling}]
Since $\FG$ is a hereditary family, $G[V_i]$ also belongs to $\FG$, and since $\FG$ is $\colf$-colorable and maximum degree of $G[V_i]$ is at most $(1+\eps)\cdot\frac{\Delta}{k}$ by Lemma~\ref{lem:sample-degree}, with high probability, 
the total number of colors needed for coloring $G$ this way is at most
\begin{align*}
	\sum_{i=1}^{k} \colf\Paren{(1+\eps)\cdot\frac{\Delta}{k})} = k \cdot \colf\Paren{(1+\eps)\cdot\frac{\Delta}{k})} = C,
\end{align*}
finalizing the proof. \Qed{Theorem~\ref{thm:vertex-sampling}}

\end{proof}

Even though Theorem~\ref{thm:vertex-sampling} is quite simple, it has various interesting implications combined with known results on chromatic number of different families of ``locally sparse" graphs. 
In the following, we first show how this theorem implies a ``recipe'' for designing sublinear algorithms and then state several of its implications.

\subsection{Sublinear Algorithms from Theorem~\ref{thm:vertex-sampling}}
As before, we only focus on streaming and query algorithms in this section. Table~\ref{tab:vertex} contains a summary of our results in this part. 
 Before getting to our results though, we first prove a simple auxiliary lemma. 

\begin{lemma}\label{lem:sample-edge}
	In the setting of Theorem~\ref{thm:vertex-sampling}, the maximum number of vertices in any graph $G[V_i]$ is at most $O(n/k)$  with high probability. 
\end{lemma}

The proof of this lemma is identical to that of
Lemma~\ref{lem:sample-degree} and is hence omitted. In the following two algorithms, the parameters $C$ and
$k$ are the same as in Theorem~\ref{thm:vertex-sampling}. 

\paragraph{Streaming Algorithms from Theorem~\ref{thm:vertex-sampling}.} The algorithm is simply as follows: 
\begin{tbox}
	\begin{enumerate}[label=(\roman*)]
	\item At the beginning, sample a random $k$-partitioning of the vertices into $V_1,\ldots,V_k$.
	\item Throughout the stream, store any edge that belongs to one of the graphs $G[V_i]$. 
	\item At the end, use the stored subgraphs to find a $C$-coloring of $G$ by coloring each $G[V_i]$ with a distinct palette of size $C/k$. 
	\end{enumerate}
\end{tbox}

	The correctness of the algorithm (with high probability) follows from Theorem~\ref{thm:vertex-sampling}. The space complexity of this algorithm is also: 
	$O(n)$ (to store the random partitioning) $+ k \cdot O(n{\Delta}/{k^2})$ (by Lemmas~\ref{lem:sample-degree} and~\ref{lem:sample-edge}) $= O(n \cdot \frac{\Delta}{k})$.
	This implies the following corollary. 

\input{tab-vertex}

	\begin{corollary}\label{cor:stream-vertex-sampling}
		Let $\FG$ be a $\colf$-colorable family of graphs (Definition~\ref{def:colorable-family}). There exists a randomized streaming algorithm that makes a single pass over any graph $G$ from $\FG$ with maximum degree $\Delta$, 
		and for any setting of parameters: 
		\[
		\eps>0, ~\qquad ~ 1 \leq k \leq \frac{\eps^2\cdot\Delta}{9\ln{n}}, ~\qquad~ C := C(\eps,k) = k \cdot \colf\Paren{(1+\eps)\cdot\frac{\Delta}{k}},
		\]
		with high probability computes a proper $C$-coloring of $G$ using $O(n \cdot \frac{\Delta}{k})$ space. 
	\end{corollary}
	
\paragraph{Query Algorithms from Theorem~\ref{thm:vertex-sampling}.} The algorithms is as follows: 
\begin{tbox}
	\begin{enumerate}[label=(\roman*)]
	\item Sample a random $k$-partitioning of the vertices into $V_1,\ldots,V_k$.
	\item Obtain the subgraphs $G[V_1],\ldots,G[V_k]$ using the following procedure: 
	\begin{itemize}
		\item If $\Delta > n/k$, then non-adaptively query all pairs of vertices $u,v$ where both $u,v$ belong to the same  $V_i$ (using pair queries); 
		\item Otherwise, non-adaptively query all neighbors of all vertices $u$ (using neighbor queries). 
	\end{itemize}
	\item Find a $C$-coloring of $G$ by coloring each $G[V_i]$ with a distinct palette of size $C/k$ (with no further access to $G$). 
	\end{enumerate}
\end{tbox}

	The correctness of the algorithm (with high probability) again follows from Theorem~\ref{thm:vertex-sampling}. The query complexity of this algorithm is also (by Lemma~\ref{lem:sample-edge}): 
	$\min\set{O(n\Delta) + O(n^2/k)}$ queries (note that the first term on its own is trivial as it requires looking at the entire graph). It now follows: 
	
		\begin{corollary}\label{cor:query-vertex-sampling}
		Let $\FG$ be a $\colf$-colorable family of graphs (Definition~\ref{def:colorable-family}). There exists a randomized non-adaptive algorithm that given query access to any graph $G$ from $\FG$ with maximum degree $\Delta$, 
		for any setting of parameters: 
		\[
		\eps>0, ~\qquad ~ 1 \leq k \leq \frac{\eps^2\cdot\Delta}{9\ln{n}}, ~\qquad~ C := C(\eps,k) = k \cdot \colf\Paren{(1+\eps)\cdot\frac{\Delta}{k}},
		\]
		with high probability computes a proper $C$-coloring of $G$ using $\min\set{O(n\Delta) + O(n^2/k)}$ queries. 
	\end{corollary}

We conclude this section with some important remarks about Corollaries~\ref{cor:stream-vertex-sampling} and~\ref{cor:query-vertex-sampling}. 

\begin{remark}[\textbf{Runtime of our algorithms}]
{We did not state the runtime of our algorithms in this section and focused primarily on space and query complexity of algorithms, respectively. 
	This is because in both cases, the runtime of the algorithm crucially depends on the runtime of the coloring algorithm for finding a $\colf$-coloring of each subgraph $G[V_i]$ which is specific to the family $\FG$ (and $\colf$) and thus 
	not known a-priori.}
	
	{Nevertheless, for {\textbf{almost all} our applications to specific families of graphs (with one exception)}, \textbf{the runtime of the algorithms is also sublinear in the input size}.} 
\end{remark}

\subsection{Particular Implications of Theorem~\ref{thm:vertex-sampling}}

We now list the applications of Theorem~\ref{thm:vertex-sampling} and Corollaries~\ref{cor:stream-vertex-sampling} and~\ref{cor:query-vertex-sampling} to different families of ``locally sparse'' graphs that
are colorable with much fewer than $(\Delta+1)$ colors. 

\subsubsection*{Triangle-Free Graphs} 

As stated earlier, triangle-free graphs admit an $O(\frac{\Delta}{\ln{\Delta}})$ coloring. This was first proved by Johansson~\cite{Johansson96a} by showing an upper bound of $9\frac{\Delta}{\ln{\Delta}}$ on the chromatic number of these graphs\footnote{This 
result of Johansson was never published -- see~\cite[Chapter~13]{ColoringBook} for a lucid presentation of the original proof.}. The leading constant was then improved to $4$ by Pettie and Su~\cite{PettieS15} and very recently to $1+o(1)$
by Molloy~\cite{Molloy19} matching the result of Kim for graphs of girth $5$~\cite{Kim95}. Moreover, Molloy's result implies an $\Ot(n\Delta^2)$ time algorithm for finding such a coloring.

Note that triangle-free graphs form a hereditary family of graphs and aforementioned results imply that they are $\colftri$-colorable for $\colftri(\Delta) = O(\frac{\Delta}{\ln{\Delta}})$. As such, Corollaries~\ref{cor:stream-vertex-sampling} and~\ref{cor:query-vertex-sampling} imply the following algorithms for any $\gamma \in (0,1/2)$ as small as $\Theta(\frac{\ln\ln{\Delta}}{\ln{\Delta}})$: 

\begin{itemize}[leftmargin=15pt]
	\item \textbf{Streaming Model:} A randomized single-pass $\Ot(n^{1+\gamma})$ space  algorithm for $O(\frac{\Delta}{\gamma \ln{\Delta}})$ coloring of triangle-free graphs.
	The post-processing time of this algorithm is $\Ot(n\cdot\Delta^{\gamma})$. 
	\item \textbf{Query Model:} A randomized non-adaptive $\Ot(n^{3/2+\gamma})$-query algorithm for $O(\frac{\Delta}{\gamma \ln{\Delta}})$ coloring of triangle-free graphs. 
	The runtime of this algorithm is also $\Ot(n^{3/2+2\gamma})$. 
\end{itemize}
Both results above are proved by picking $\eps = \Theta(1)$ and $k = \Theta(\Delta^{1-\gamma})$, thus obtaining a $C$-coloring: 
\begin{align*}
	C = C(\eps,k) = k \cdot \colftri\Paren{\Theta(\Delta/k)} = O(k) \cdot \frac{\Delta/k}{\ln{(\Delta/k)}} = O(\frac{\Delta}{\ln{\Delta^{\gamma}}}) = O(\frac{\Delta}{\gamma \ln{\Delta}}). 
\end{align*}

\begin{remark}
	{The above approach can also be used to obtain a linear time classical algorithm for $O(\frac{\Delta}{\ln{\Delta}})$ coloring of triangle-free graphs faster than the state-of-the-art algorithm of  Molloy~\cite{Molloy19} (albeit with a larger number of colors by a constant factor). For any $\gamma \in (0,1/2)$, we obtain an algorithm for $O(\frac{\Delta}{\gamma \cdot \ln{\Delta}})$ coloring of triangle-free graphs in $O(n\Delta) + \Ot(n\Delta^{2\gamma}) = O(n\Delta)$ time. }
\end{remark}

\subsubsection*{$\bm{K_r}$-Free Graphs}

For any fixed integer $r \geq 1$, we refer to any graph that does not contain a copy of the $K_r$, namely, the clique on $r$ vertices, as a $K_r$-free graph. Johansson proved that any $K_r$-free graph 
admits an $O(\frac{\Delta \ln\ln{\Delta}}{\ln{\Delta}})$ coloring~\cite{Johansson96b} and gave an $O(n \cdot \poly(\Delta))$ time algorithm for finding it\footnote{This result of Johansson was also
never published -- see~\cite{BansalGG15} for a streamlined version of this proof.}. This result was very recently simplified (and extended to $r$ beyond a fixed constant) by Molloy~\cite{Molloy19} (however the latter result does not imply an efficient algorithm).  

Similar to the case of triangle-free graphs, combining these results with Corollaries~\ref{cor:stream-vertex-sampling} and~\ref{cor:query-vertex-sampling} imply the following algorithms
for any $\gamma \in (0,1/2)$ as small as $\Theta(\frac{\ln\ln{\Delta}}{\ln{\Delta}})$: 

\begin{itemize}[leftmargin=15pt]
	\item \textbf{Streaming Model:} A randomized single-pass $\Ot(n^{1+\gamma})$ space  algorithm for $O(\frac{\Delta\ln\ln{\Delta}}{\gamma \ln{\Delta}})$ coloring of $K_r$-free graphs.
	The post-processing time of this algorithm is $O(n^{1+\Theta(\gamma)})$. 
	\item \textbf{Query Model:} A randomized non-adaptive $\Ot(n^{3/2+\gamma})$-query algorithm for $O(\frac{\Delta\ln\ln{\Delta}}{\gamma \ln{\Delta}})$ coloring of $K_r$-free graphs. 
	The runtime of this algorithm is also $O(n^{3/2+\Theta(\gamma)})$. 
\end{itemize}

\subsubsection*{Graphs with $\bm{r}$-Colorable Neighborhoods}

For any fixed integer $r \geq 1$, we say that a graph $G$ is locally $r$-colorable iff neighborhood of every vertex in $G$ is $r$-colorable. Johansson also proved that $r$-colorable graphs admits an $O(\frac{\Delta}{\ln{\Delta}} \cdot \ln{r})$ coloring~\cite{Johansson96b}; 
see~\cite{BansalGG15} for a proof and also an algorithm that finds such a coloring in $\poly(n \cdot 2^{\Delta})$ time (which uses, as a subroutine, a result of~\cite{BjorklundHK09}). 

It is easy to see that locally $r$-colorable graphs also form a hereditary family. Consequently, as before, Corollaries~\ref{cor:stream-vertex-sampling} and~\ref{cor:query-vertex-sampling} imply the following 
for any $\gamma \in (0,1/2)$ as small as $\Theta(\frac{\ln\ln{\Delta}}{\ln{\Delta}})$: 

\begin{itemize}[leftmargin=15pt]
	\item \textbf{Streaming Model:} A randomized single-pass $\Ot(n^{1+\gamma})$ space algorithm for $O(\frac{\Delta}{\gamma \ln{\Delta}} \cdot \ln{r})$ coloring of locally $r$-colorable graphs.
	The post-processing time of the algorithm is $\poly(n \cdot 2^{\Delta^{\gamma}})$. 
	\item \textbf{Query Model:} A randomized non-adaptive $\Ot(n^{3/2+\gamma})$-query algorithm for $O(\frac{\Delta}{\gamma \ln{\Delta}} \cdot \ln{r})$ coloring of locally $r$-colorable graphs. 
	The runtime of this algorithm is also $\poly(n \cdot 2^{\Delta^{\gamma}})$. 
\end{itemize}

\begin{remark}\label{rem:r-colorable}
	By picking $k = \Theta(\Delta/\log{n})$ in the query algorithm above (instead of $k=\Delta^{1-\gamma}$ in the above part), we obtain a (classical) algorithm for $O(\frac{\Delta}{\ln\ln{n}} \cdot \ln{r})$ coloring
	in $\poly(n \cdot 2^{\Theta(\log{n})}) = \poly(n)$ time. Although the number of colors of this algorithm is sub-optimal for $\Delta > (\log{n})^{\omega(1)}$, this gives a polynomial time algorithm for 
	coloring these graphs. 
\end{remark}

\subsubsection*{Graphs with $\bm{\delta}$-Sparse Neighborhoods} 
For any $\delta \in (0,1)$, we say a graph $G(V,E)$ has a $\delta$-sparse neighborhood iff the total number of edges in the neighborhood of any vertex $v$ (i.e., edges between neighbors of $v$) is at most $\delta \cdot \Delta^2$ (not to be confused with Definition~\ref{def:eps-sparse} for $\eps$-sparse vertices, albeit the two definitions are equivalent for $\Delta$-\emph{regular} graphs by setting $\delta=(1-\eps^2)$). 
Alon, Krivelevich and Sudakov~\cite{AlonKS99} proved that any graph $G$ with maximum degree $\Delta$ and $\delta$-sparse neighborhood admits an $O(\frac{\Delta}{\log{(1/\delta)}})$ coloring and that this is tight for all admissible values of $\delta$
and $\Delta$. 

We note that unlike all other families of graphs considered in this section, the family of sparse-neighborhood graphs is \emph{not} a hereditary family. As such, we cannot readily apply Theorem~\ref{thm:vertex-sampling} (and hence 
Corollaries~\ref{cor:stream-vertex-sampling} and~\ref{cor:query-vertex-sampling}). However, we can modify the proof of Theorem~\ref{thm:vertex-sampling} slightly to apply to this case as well. In particular, we prove the following result. 

\begin{lemma}\label{lem:vertex-sampling-sparse}
	For any $\delta \in (0,1)$, let $G(V,E)$ be an $n$-vertex graph with maximum degree $\Delta$ and $\delta$-sparse neighborhoods. For the parameters 
	\[
		1 \leq k \leq \frac{\delta \cdot \Delta}{9\cdot\ln{n}}, ~\qquad~ C := \Theta(\frac{\Delta}{\ln{(1/\delta)}}),
	\]
	suppose we partition $V$ into $k$ sets $V_1,\ldots,V_k$ uniformly at random; then with high probability $G$ can be $C$-colored by coloring each $G[V_i]$ with a distinct palette of size $C/k$. 
\end{lemma}

The proof of this result is by simply showing that not only the maximum degree of each graph $G[V_i]$ is sufficiently small (Lemma~\ref{lem:sample-degree}), but also it is a $(2\delta)$-sparse neighborhood graph.  

\begin{lemma}\label{lem:sample-sparse}
	 With high probability $G[V_i]$ has a $(2\delta)$-sparse neighborhood. 
\end{lemma}
\begin{proof}
	Fix a vertex $v \in V_i$. For any vertex $u \in N(v)$, let $\deg_{N(v)}(u)$ denote the degree of $u$ to other vertices in $N(v)$. 
	Moreover, define $\deg^i_{N(v)}(u)$ as the degree of $u$ to vertices in $N(v)$ that are also present in $V_i$, hence, 
	\begin{align*}
		\expect{\deg^i_{N(v)}(u)}  =  \frac{1}{k}\cdot \deg_{N(v)}(u) \leq \Delta/k. 
	\end{align*}
	Moreover, $\deg^i_{N(v)}(u)$ is a sum of $\Delta$ independent random variables and hence by Chernoff bound (Proposition~\ref{prop:chernoff} with $\mu = \Delta/k$ and $\delta = 1$): 
	\begin{align*}
		\Pr\paren{\deg^i_{N(v)}(u) \geq 2\Delta/k} \leq \exp\paren{-\frac{\Delta}{3k}} \leq \frac{1}{n^3},
	\end{align*}
	by the condition on value of $k$. By a union bound, with high probability, for all vertices $v \in V_i$ and $u \in N(v)$ the above inequality holds. In the following, we condition on this event. Note that as this is a
	``high probability'' event, this conditioning does not change the distribution of random variables by more than a negligible factor. 
	
	Again fix a vertex $v \in V_i$. Define (at most) $\Delta$ random variables $X_u$ for $u \in N(v)$ where $X_u = \deg^i_{N(v)}(u)$ iff $u$ is also sampled in $V_i$ and otherwise $X_u=0$. 
	Define $X:= \sum_{u \in N(v)} X_u$ to be the number of edges between vertices in $N(v) \cap V_i$. As each edge appears in $G[V_i]$ w.p. $1/k^2$, by linearity of expectation, 
	\begin{align*}
		\expect{X} \leq \delta \cdot \Delta^2 \cdot \frac{1}{k^2}. 
	\end{align*}
	Moreover, as $X$ is a sum of independent random variables which are in $[0,2\Delta/k]$ (by the high probability event we conditioned on), an application of Chernoff bound (Proposition~\ref{prop:chernoff}) implies that: 
	\begin{align*}
		\Pr\paren{X \geq 2 \cdot \delta \cdot \frac{\Delta^2}{k^2}} \leq \exp\paren{-\frac{(\delta^2 \cdot \Delta^4/k^4)}{3 \cdot \Delta^2/k^2}} = \exp\paren{-\frac{\delta^2 \cdot \Delta^2}{3k^2}} \leq \frac{1}{n^3}, 
	\end{align*}
	by the choice of $k$. We take another union bound over all vertices $v \in V$. 
	
	Finally, as by Lemma~\ref{lem:sample-degree}, we have that maximum degree of $G[V_i]$ is at most $2\Delta/k$ and since by the above argument, neighborhood of each vertex contains at most $2\delta \cdot \Delta^2/k^2$ edges, 
	we obtain that $G[V_i]$ has $(2\delta)$-sparse neighborhoods, concluding the proof. 
\end{proof}

Lemma~\ref{lem:vertex-sampling-sparse} now follows from Lemma~\ref{lem:sample-sparse} (the same exact way as in the proof of Theorem~\ref{thm:vertex-sampling}). 
Similar to Corollaries~\ref{cor:stream-vertex-sampling} and~\ref{cor:query-vertex-sampling}, this in turn implies the following algorithms: 

\begin{itemize}[leftmargin=15pt]
	\item \textbf{Streaming Model:} A randomized single-pass $\Ot(n/\delta)$ space algorithm for $O(\frac{\Delta}{\ln{(1/\delta)}})$ coloring of graphs with $\delta$-sparse neighborhoods.
	The post-processing time  is $\Ot(n \cdot \poly{(1/\delta)})$. 
	\item \textbf{Query Model:} A randomized non-adaptive $\Ot(n^{3/2}/\delta)$-query algorithm for $O(\frac{\Delta}{\ln{(1/\delta)}})$ coloring of graphs with $\delta$-sparse neighborhoods.
	The runtime of the algorithm is $\Ot(n^{3/2} \cdot \poly{(1/\delta)})$ 
\end{itemize}

%% file: tab-vertex.tex
 \def\arraystretch{2}

\setlength{\tabcolsep}{12pt}

\newsavebox{\tabtwo}

\sbox{\tabtwo}{
             {\small
             
        \centering
      
        \begin{tabular}{|c||c|c|c|}
             \hline
     
        \textbf{ \# of Colors} & \textbf{Graph Family} &  \textbf{Streaming} & \textbf{Sublinear-Time} \\
             \hline 
             \hline
	     {$O(\frac{\Delta}{\gamma \cdot \ln{\Delta}})$ }  & 	Triangle-Free 	&  $O(n\Delta^{2\gamma})$ space  & $\Ot(n^{3/2+2\gamma}) $ time  \\	
	     {$O(\frac{\Delta \ln\ln{\Delta}}{\gamma \cdot \ln{\Delta}})$ }  & 	$K_r$-Free 	&  $O(n\Delta^{2\gamma})$ space  & $\Ot(n^{3/2+\Theta(\gamma)}) $ time  \\	
	     {$O(\frac{\Delta}{\gamma \ln{\Delta}} \cdot \ln{r})$ }  & 	Locally $r$-Colorable 	&  $O(n\Delta^{2\gamma})$ space  & $\Ot(n^{3/2+2\gamma}) $ \underline{queries}  \\	
	     {$O(\frac{\Delta}{\gamma \ln{\ln{n}}} \cdot \ln{r})$ }  & 	Locally $r$-Colorable 	&  $O(n\Delta^{2\gamma})$ space  & $\poly(n) $ \underline{time}  \\	
	     {$O(\frac{\Delta}{\ln{(1/\delta)}})$ }  & 	$\delta$-Sparse-Neighborhood 	&  $O(n/\delta)$ space  & $\Ot(n^{3/2} \cdot \poly(1/\delta)) $ time  \\	

	   \hline
        \end{tabular}
      }
  }
  
 \begin{table}[t!]
\begin{tikzpicture}
   \node[fill=white](boz){};
  \node[drop shadow={black, shadow xshift=5pt,shadow yshift=-5pt, opacity=0.5}, fill=white, inner xsep=-7pt, inner ysep=0pt](table)[right=5pt of boz]{\usebox{\tabtwo}};
\end{tikzpicture}
\vspace{0.25cm}
          \caption{A sample of our sublinear algorithms obtained as corollaries of Theorem~\ref{thm:vertex-sampling}.  
          All the streaming algorithms here are \emph{single-pass} and all sublinear-time algorithms are \emph{non-adaptive}. Note the two different rows for locally $r$-colorable graphs; see also Remark~\ref{rem:r-colorable}. 
        \label{tab:vertex}}

    \end{table}

%% file: proof-triangle.tex

\section{Proof of Proposition~\ref{prop:lc-triangle}}\label{app:lc-triangle}

We present the proof of Proposition~\ref{prop:lc-triangle}, restated below, in this section. 

\begin{proposition*}[Restatement of Proposition~\ref{prop:lc-triangle}]
   There exists an absolute constant $d_0$ such that for all $d \geq d_0$ the following holds. 
   Suppose $G(V,E)$ is a triangle-free graph with lists $S(v)$ for every $v \in V$ such that:
   \begin{enumerate}[label=(\roman*)]
    	\item for every vertex $v$, $\card{S(v)} \geq 8 \cdot \frac{d}{\ln{d}}$, and
	\item for every vertex $v$ and color $c \in S(v)$, $\deg_S(v,c) \leq d$. 
    \end{enumerate} 
    Then, there exists a proper coloring of $G$ from these lists. 
\end{proposition*}

We prove Proposition~\ref{prop:lc-triangle} using the probabilistic method and in particular a version of the so-called ``R\"{o}dl Nibble'', the ``semi-random method'', or the ``wasteful coloring procedure''; see, e.g.~\cite{Rodl85,ColoringBook}: The idea 
is to iteratively find a partial coloring of $G$ from the given
lists by coloring a small fraction of the vertices randomly, update the lists of their neighbors, and continue until we can color $G$ entirely. We shall remark 
that our approach in proving Proposition~\ref{prop:lc-triangle} closely follows the distributed algorithm of Pettie and Su~\cite{PettieS15} and we borrow several ideas from their work although there are many differences as well. 

\newcommand{\alphastar}{\ensuremath{\alpha^{*}}}
\newcommand{\betastar}{\ensuremath{\beta^{*}}}
\newcommand{\hA}{\ensuremath{\widehat{A}}}
\newcommand{\ha}{\ensuremath{\widehat{a}}}
\newcommand{\hB}{\ensuremath{\widehat{B}}}
\newcommand{\hb}{\ensuremath{\widehat{b}}}

\newcommand{\bb}{\ensuremath{{b}}}

\renewcommand{\ss}{\ensuremath{\eta}}
\newcommand{\Ss}{\ensuremath{\eta^{*}}}

\newcommand{\hsss}{\ensuremath{\widehat{\ss}}}
\newcommand{\hlambda}{\ensuremath{\widehat{\lambda}}}

\newcommand{\alphaid}{\ensuremath{\alpha^{\textnormal{\textsf{ideal}}}}}
\newcommand{\betaid}{\ensuremath{\beta^{\textnormal{\textsf{ideal}}}}}

\newcommand{\amin}{\ensuremath{a^{\textnormal{\textsf{min}}}}}
\newcommand{\bmax}{\ensuremath{b^{\textnormal{\textsf{max}}}}}
\newcommand{\ssmax}{\ensuremath{\ss^{\textnormal{\textsf{max}}}}}

\subsubsection*{Preliminaries and Parameters}

Our procedure is iterative. Each iteration $i$ of the procedure uses the following parameters: 
\begin{itemize}
	\item $\alphaid_i$: used as an ``ideal'' \emph{lower bound} for size of each list;
	\item $\betaid_i$: used as an ``ideal'' \emph{upper bound} on the $c$-degree of each vertex $v \in G_i$ for every $c \in A_i(v)$.
\end{itemize}
These parameters are defined recursively as follows (these expressions would become clear shortly):
\begin{align}
		&\pkeep_i := \paren{1-\frac{1}{2\ln{d} \cdot \alphaid_i}}^{2\betaid_i}  && \pcolor_i := \paren{1-\frac{1}{2\ln{d} \cdot \alphaid_i}}^{\pkeep_i \cdot \alphaid_i/2} \notag \\
		&\alphaid_1 = 8 \cdot \frac{d}{\ln{d}}  && \alphaid_{i+1} :=  \pkeep_i \cdot \alphaid_i \notag \\
		&\betaid_1 = d   && \betaid_{i+1} :=  \pcolor_i \cdot \pkeep_i \cdot \betaid_i. \label{eq:starting}
\end{align}
The following lemma lists some of the main relations between parameters $\alphaid_i$ and $\betaid_i$ that we use throughout the proof. The proof is by some rather straightforward (albeit daunting) calculations. 
\begin{lemma}\label{lem:ideal-list}
	The parameters $\alphaid_i$ and $\betaid_i$ satisfy the following properties: 
	\begin{enumerate}[label=(\roman*)]
		\item For every $i$, ${\betaid_i}/{\alphaid_i} \leq \betaid_1/\alphaid_1 \leq \ln{d}/8$. 
		\item There exists some sufficiently small $\delta = \Theta(1)$ such that for every $i$, $\alphaid_i \geq d^{\delta}$. 
		\item There exists an $\istar = O(\log^2{d})$ such that $\betaid_{\istar} < \alphaid_{\istar} /100$. 
	\end{enumerate}
\end{lemma}
\begin{proof}
	The first part is immediate as the ratio $\betaid_i/\alphaid_i$ drops by a factor $\pcolor_i \in (0,1)$ in each iteration. We now prove the second part. Firstly, 
	\begin{align*}
		\pkeep_i = \paren{1-\frac{1}{2\ln{d} \cdot \alphaid_i}}^{2\betaid_i} \geq \exp\paren{-\frac{2\betaid_i}{\ln{d} \cdot \alphaid_i}} \geq \exp\paren{-\frac{2\betaid_1}{\ln{d} \cdot \alphaid_1}} \geq 3/4.
	\end{align*}
	By definition of $\pcolor_i$: 
	\begin{align*}
		\pcolor_i = \paren{1-\frac{1}{2\ln{d} \cdot \alphaid_i}}^{\pkeep_i \cdot \alphaid_i/2} \leq \exp\paren{-\frac{\pkeep_i}{4\ln{d}}}  \leq 1-\frac{1}{6\ln{d}}.  
	\end{align*}
	 Define $r_i := \betaid_i/\alphaid_i$. By the above equation:
	 \begin{align*}
	 	r_{i+1} = \pcolor_i \cdot r_i \leq \Paren{1-\frac{1}{6\ln{d}}} \cdot r_i \leq \Paren{1-\frac{1}{6\ln{d}}}^{i} \cdot r_1. 
	 \end{align*}
	 This in turn allows us to bound $\pkeep_i$: 
	 \begin{align*}
	 	\pkeep_i &= \paren{1-\frac{1}{2\ln{d} \cdot \alphaid_i}}^{2\betaid_i} \geq \exp\paren{-\frac{1+o(1)}{\ln{d}}\cdot r_i} \\
		&\geq \exp\paren{-\frac{1+o(1)}{\ln{d}} \cdot \Paren{1-\frac{1}{6\ln{d}}}^{i-1} \cdot r_1} \\
		&= \exp\paren{-\frac{(1+o(1))}{8} \cdot \Paren{1-\frac{1}{6\ln{d}}}^{i-1}} \tag{as $r_1 = \ln{d}/8$}. 
	 \end{align*}
	 By using this bound in the definition of $\alphaid_{i}$, we get that: 
	 \begin{align*}
	 	\alphaid_{i} = \alphaid_1 \cdot \prod_{j=1}^{i-1} \pkeep_j  &\geq \alpha_1 \cdot \exp\paren{-\frac{(1+o(1))}{8} \cdot \sum_{j=1}^{i-1}\Paren{1-\frac{1}{6\ln{d}}}^{j-1}} \\
		&\geq \alpha_1 \cdot \exp\paren{-\frac{(1+o(1))}{8} \cdot 6\ln{d}} \\
		&\geq d^{\delta} \tag{for some  small $\delta=\Theta(1)$}
	 \end{align*}
	 This proves the second part. For the third part, note that as long as $\betaid_{i} \geq \alphaid_i/100$, we have, 
	 \begin{align*}
	 	\pkeep_i &=   \paren{1-\frac{1}{2\ln{d} \cdot \alphaid_i}}^{2\betaid_i} \leq \exp\paren{-\frac{1}{\ln{d}} \cdot\frac{\betaid_i}{\alphaid_i}} \leq \exp\paren{-\frac{1}{100\ln{d}}} \leq 1-\frac{\Theta(1)}{\ln{d}}.
	 \end{align*}
	 This, together with the upper bound on $\pcolor_i$ implies that:	 
	 \begin{align*}
	 	\betaid_{i} \leq \betaid_1 \cdot \prod_{j=1}^{i-1} \pcolor_j \cdot \pkeep_j \leq d \cdot \paren{1-\frac{\Theta(1)}{\ln{d}}}^{i-1}.
	 \end{align*}
	 As such, as long as $\betaid_i \geq \alphaid_i/100$, $\betaid_i$ will decrease by at least some fixed rate while by the second part we know that $\alphaid_i$ will never go below $d^\delta$ for some constant $\delta$. Hence,
	after $\istar = O(\log^2{d})$ steps we will have $\betaid_i < \alphaid_i/100$. 
\end{proof}

\paragraph{Notation.} We further define the following notation to describe our procedure. 
The definition of some of these parameters would become more clear later but we still list them all here for ease of reference 
(in the following $G_1 := G$ and $A_1(v) = S(v)$).  
\begin{itemize}
	\item $G_i$: The remaining graph to color at the beginning of iteration $i$;
	\item $A_i(v)$: List of available colors to $v \in G_i$ at the beginning of iteration $i$ -- let $a_i(v) := \card{A_i(v)}$; \\ we further define $\amin_i := \min_v a_i(v)$;
	\item $B_i(v,c)$: Set of vertices $u \in N(v)$ such that $c \in A_i(u)$ -- let $b_i(v,c) := \card{B_i(v,c)}$; \\ we further define $\bmax_i := \max_{v,c} b_i(v,c)$;
	\item $\hA_i(v)$: The {intermediate} list of colors of vertex $v$ during iteration $i$ -- let $\ha_i(v) := \card{\hA_i(v)}$; 
	\item $\hB_i(v,c)$: Set of vertices  $u \in N(v)$ such that $c \in \hA_i(u)$ -- let $\hb_i(v,c) := \card{\hB_i(v,c)}$.
\end{itemize}

\subsection{The Coloring Procedure}

Each iteration $i$ of our procedure is as follows (with a minor modification described below): 
\begin{tbox}
	\wasteful: The algorithm for each iteration $i$ of the coloring procedure. 
\begin{enumerate}
		\item For every vertex $v \in G_i$ and every color $c \in A_i(v)$, we \textbf{assign} $c$ to $v$ with probability $p_i(v) := \frac{1}{2\ln{d} \cdot \alphaid_i}$ and include the assigned colors in a set $C_i(v)$.
		\item\label{line:tri-modify} For every $v \in G_i$ we obtain the intermediate list $\hA_i(v)$ from $A_i(v)$ by removing each color $c$ assigned to some $u \in N_{G_i}(v)$, i.e., if $c \in C_i(u)$. 
		\item If there exists a color $c \in \hA_i(v) \cap C_i(v)$, \textbf{color} $v$ with $c$ (breaking the ties arbitrarily). 
		\item\label{line:tri-def} Update the following parameters for the next iteration:
		\[ G_{i+1} :=  G_i \setminus \set{\text{colored vertices in iteration $i$}},  A_{i+1}(v) := \set{c \in \hA_i(v) \mid \hb_{i}(v,c) \leq 2  \betaid_{i+1}}.\] 
\end{enumerate}
\end{tbox}
\noindent
Several remarks are in order. Firstly, it is easy to see that the partial coloring found by this procedure is always feasible: we (conservatively) throw out any color $c$ from the list $A_i(v)$ of a vertex $v$ if it is assigned to (and not even
necessarily used to color) a neighbor of $v$. Secondly, at the end of each iteration, we additionally throw out any color $c$ from $A_i(v)$ that has a ``large'' $c$-degree $b_i(v,c) > 2 \betaid_i$, hence, the $c$-degrees
of vertices is at most twice the ideal value $\betaid_i$. Finally, we will run this procedure 
up until a certain point where we can guarantee that the size of $A_i(v)$ for every vertex $v \in G_i$ is some constant factor larger than the $c$-degree of $v$ for $c \in A_i(v)$: at this point, 
we can simply apply Proposition~\ref{prop:lc-const} to color the remainder of the graph.

\paragraph{Equalizing probabilities:} Let $\pkeep_i(v,c)$ denote the probability that color $c \in A_i(v)$ is being kept in $\hA_i(v)$. 
It would make our proof much easier if all valid choices of $v,c$ have the same probability $\pkeep_i(v,c) = \pkeep_i$ (where $\pkeep_i$ is defined in Eq~\eqref{eq:starting}). 
While this is  not guaranteed by the \wasteful procedure, as we show below
a simple additional step in every iteration can ensure this property. Note that for every choices of $v$ and $c \in A_i(v)$:
\begin{align}
	\pkeep_{i}(v,c) &= \Pr\paren{\text{$c$ is not assigned to any vertex in $B_i(v,c)$}} \notag \\
	&= \prod_{u \in B_i(v,c)} (1-\frac{1}{\ln{d}} \cdot \frac{1}{2\alphaid_i}) \geq \Paren{1-\frac{1}{2\ln{d} \cdot \alphaid_i}}^{2\betaid_i} = \pkeep_i \label{eq:keepdef}.
\end{align}
 We \textbf{modify the procedure} by {removing} each color $c \in \hA_i(v)$ with probability $1-\frac{\pkeep_i}{\pkeep_i(v,c)}$ in Line~\eqref{line:tri-modify} of $\wasteful$. 
As a consequence of this, in the modified procedure, for every valid choices of $v,c$ in iteration $i$: 
\begin{align}
	&\Pr\paren{\text{$c \in A_i(v)$ belongs to $\hA_{i}(v)$ in iteration $i$}} = \pkeep_i. \label{eq:keep} 
\end{align}
\noindent
From now on, we work with this modified procedure and hence we can use Eq~\eqref{eq:keep}.

\subsubsection*{The Setup}

Recall that $\amin_i$ denotes the minimum list size and $\bmax_i$ denotes the maximum $c$-degree in each iteration $i$. 
Our goal is to maintain the invariant that in each iteration $i$, $\amin_i \geq \alphaid_i/2$ and $\bmax_i \leq 2\betaid_i$ (as stated, this invariant is ``too tight'' and thus 
in the proof we actually allow for some small approximation to take care of the errors due to the concentration bounds). As we know by Lemma~\ref{lem:ideal-list} that eventually 
$\betaid_i < \alphaid_i / 100$, such an invariant allows us to reach an iteration $i$ where $\amin_i > \bmax_i/10$. At this point, we can apply Proposition~\ref{prop:lc-const} and color the rest of the graph. 

It turns out for the purpose of bounding $\amin_i$ and $\bmax_i$, working with the parameters $a_i(v)$ and $b_i(v,c)$ directly is a hard task due to the lack of appropriate concentration (in particular, $b_i(v,c)$'s are not concentrated). 
To address this, let us further define the following parameters: 
\begin{itemize}
	\item $\lambda_i(v) := \min\set{1,a_i(v)/\alphaid_i}$: the ratio of size of list $A_i(v)$ to the ideal size $\alphaid_i$; 
	\item $\bb_i(v) := \sum_{c \in A_i(v)} b_i(v,c)/a_i(v)$: the average $c$-degree of $v$ in $A_i(v)$;
	\item $\ss_i(v) := \lambda_i(v) \cdot \bb_i(v) + (1-\lambda_i(v)) \cdot 2\betaid_i$; we further define $\ssmax_i := \max_v~\ss_i(v)$. 
\end{itemize}
We note that $\ss_i(v)$ can be seen as the average $c$-degree of $v$ if we add $\alphaid_i - a_i(v)$ new artificial colors with $c$-degree $2\betaid_i$ to $v$. 
Let us first see how does these parameters can help with our goal of bounding $\amin_i$ and $\bmax_i$. 
\begin{claim}\label{clm:tri-ss-helps}
	For any iteration $i$: 
		\begin{align*}
		\amin_i \geq \alphaid_i \cdot \paren{1-\frac{\ssmax_i}{2\betaid_i}} \qquad \textnormal{and} \qquad \bmax_i \leq 2 \cdot \betaid_i. 
	\end{align*}
\end{claim}
\begin{proof}
	The proof of the second part follows from the condition in Line~\eqref{line:tri-def} of $\wasteful$ as $b_i(v,c) \leq \hb_{i-1}(v,c) \leq 2\cdot\betaid_i$ for every $v$ and any $c \in A_{i}(v)$. 
	For the first part, consider any $v$ where $a_i(v) < \alphaid_i$ (if no such $v$ exists we are already done):
	\begin{align*}
		a_i(v) &= \lambda_i(v) \cdot \alphaid_i \geq \alphaid_i \cdot \paren{1-\frac{\ss_i(v)}{2\betaid_i}} \geq \alphaid_i \cdot \paren{1-\frac{\ssmax_i}{2\betaid_i}}, 
	\end{align*}
	where the first inequality follows from the definition of $\ss_i(v)$. 
\end{proof}

As such, instead of directly computing $\amin_i$ and $\bmax_i$, we instead maintain the invariant that $\ssmax_i \leq \betaid_i$ (again modulo some small approximation terms), and 
then plugin in this value in Claim~\ref{clm:tri-ss-helps} to obtain the desired bounds on $\amin_i$ and $\bmax_i$. We shall note that this invariant on $\ssmax_i$ is analogous to the induction hypothesis of~\cite{PettieS15}
and is heart of the proof. The rest of the proof from there is straightforward as we already discussed. 

\subsection{Bounding $\ssmax_i$ in Each Iteration} 

We now state and prove the aforementioned bound on $\ssmax_i$ for each iteration $i$. 
The following lemma allows us to bound $\ssmax_i$ inductively using the fact that  $\ssmax_1 = \bmax_1 = d$ as a base case. 

\begin{lemma}\label{lem:tri-main}
	Consider any iteration $i < \istar$ and let $\eps \in (0,1)$ be a parameter such that $\eps > d^{-\delta/10}$ (for $\istar$ and $\delta$ defined in Lemma~\ref{lem:ideal-list}). Suppose
	\begin{align*}
		\ssmax_i \leq (1+\eps) \cdot \betaid_i. 
	\end{align*}
	Then, with positive probability, 
	\begin{align*}
		\ssmax_{i+1} \leq (1+19\eps) \cdot \betaid_{i+1}. 
	\end{align*}
\end{lemma}

We prove Lemma~\ref{lem:tri-main} in this part. In the following, we condition on the events that happened in iterations $<i$ 
so far including the assumption that $\ssmax_i \leq (1+\eps) \cdot \betaid_i$ and only consider the probability of events with respect to random choices in iteration $i$. Claim~\ref{clm:tri-ss-helps} then implies that: 
\begin{align}
	\ssmax_i \leq (1+\eps) \cdot \betaid_i \qquad \textnormal{,} \qquad \amin_i &\geq \frac{1-\eps}{2} \cdot \alphaid_i \qquad \textnormal{and} \qquad \bmax_i \leq 2 \betaid_i. \label{eq:min-max}
\end{align}

Recall that $A_{i+1}(v)$ is obtained by first moving from $A_i(v)$ to $\hA_i(v)$ through the process of assigning colors and then from $\hA_i(v)$ to $A_{i+1}(v)$ by filtering out the high $c$-degree colors. 
Our main goal is to understand the change between $A_i(v)$ to $\hA_i(v)$. To this end, let us further define: 

\begin{itemize}
	\item $\hb_i(v) := \sum_{c \in \hA_i(v)} \hb_i(v,c)/\ha_i(v)$: the average $c$-degree of $v$ in $\hA_i(v)$.
\end{itemize}

In the following two lemmas, we prove that both $\ha_i(v)$ and $\hb_i(v)$ are concentrated. These are the main parts of the proof and in the only part when we use $G$ is triangle-free. 

\begin{lemma}\label{lem:tri-conc-a}
	For any vertex $v \in G_i$: 
	\begin{align*}
		\Pr\paren{\ha_{i}(v) < (1-\eps) \cdot \pkeep_i \cdot a_i(v)} < \exp\paren{-\Theta(d^{4\delta/5})}.
	\end{align*}

\end{lemma}
\begin{proof}
	Recall that $\hA_i(v)$ is obtained  by picking each color $c \in A_i(v)$ that is not assigned to a neighbor of $v$. By Eq~\eqref{eq:keep}, the probability of this event for each color is precisely $\pkeep_i$. 
	Moreover, the colors are chosen independently of each other to be included in $\hA_i(v)$. Hence, $\ha_i(v)$ is a sum of $a_i(v)$ independent $\set{0,1}$-random variables with $\expect{\ha_i(v)} = \pkeep_i \cdot a_i(v)$. 
	Hence, by Chernoff bound (Proposition~\ref{prop:chernoff} and since $\pkeep_i = \Omega(1)$):
	\begin{align*}
		\Pr\paren{\ha_{i}(v) < (1-\eps) \cdot \pkeep_i \cdot a_i(v)} \leq \exp\paren{-\Theta(1) \cdot \eps^2 \cdot a_i(v)} \leq \exp\paren{-\Theta(1) \cdot d^{4\delta/5}}, 
	\end{align*}
	where the last inequality is because 	by Eq~\eqref{eq:min-max}, $a_i(v) \geq \Theta(1) \cdot \alphaid_i$, by Lemma~\ref{lem:ideal-list}, $\alphaid_i \geq d^{\delta}$, and since $\eps > d^{-\delta/10}$. 
\end{proof}

\begin{lemma}\label{lem:tri-conc-b}
	For any unfinished iteration $i$ and vertex $v \in G_i$: 
	\begin{align*}
		\Pr\paren{\hb_i(v) >  \pcolor_i \cdot \pkeep_i \cdot \bb_i(v) + 8\eps \cdot \pcolor_i \cdot \pkeep_i \cdot \bmax_i} < \exp\paren{-\Theta(d^{4\delta/5})}.
	\end{align*}
\end{lemma}
\begin{proof}
	Let us additionally define the following parameters similar to $\hb_{i+1}(v,c)$ and $\bb_i(v),\hb_i(v)$: 
	\begin{itemize}
		\item $\bp_{i}(v,c)$: number of neighbors $u \in B_i(v,c)$  that keep the color $c \in \hA_{i}(v)$ \emph{regardless} of whether they are colored in this iteration or not (in other words, $u$ will be
		 counted in $\bp_{i}(v,c)$ even if $u$ is colored in this iteration as long as $c \in \hA_{i}(u)$).  As such, $\bp_{i}(v,c) \geq \hb_{i}(v,c)$.
		\item $\bp_{i}(v):= \sum_{c \in \hA_{i}(v)} \bp_{i}(v,c)$ (we emphasize that unlike $\bb_i(v)$ and $\hb_i(v)$ which are the \emph{average} of $b_i(v,c)$'s and $\hb_i(v,c)$'s, here we take $\bp_{i}(v)$ to be the \emph{sum} of 
		$\bp_{i}(v,c)$'s for simplicity).
	\end{itemize}
	In the following claims, we first upper bound $\bp_{i}(v)$ and then relate it $\hb_i(v)$ and $\bb_i(v)$. 
	\begin{claim}\label{clm:tri-conc-b1}
		$\Pr\paren{\bp_{i}(v) > \pkeep^2_i \cdot a_i(v) \cdot \bb_i(v) + \eps \cdot \pkeep^2_i \cdot a_i(v) \cdot \bmax_i} \leq \exp\paren{-\Theta(d^{4\delta/5})}.$
	\end{claim}
	\begin{proof}
	We argue that: 
	\begin{align}
		\expect{\bp_{i}(v)} = \expect{\sum_{c \in A_i(v)} \II\bracket{c \in \hA_{i}(v)} \cdot \bp_{i}(v,c)} &=\sum_{c \in A_i(v)} \Pr\paren{\text{$c \in \hA_{i}(v)$}} \cdot \expect{\bp_{i}(v,c)}. \label{eq:expect-kept}
	\end{align}
	To do this, we prove that the event $c \in \hA_{i}(v)$ is \emph{independent} of the random variable $\bp_{i}(v,c)$. Indeed, the event $c \in \hA_{i}(v)$ is only a function of random choices of vertices $u \in B_{i}(v,c)$.
	On the other hand, for any vertex $u \in B_{i}(v,c)$ the choice of whether $u$ is counted in $\bp_{i}(v,c)$ is only a function of vertices $w \in B_{i}(u,c)$ (note that in definition of $\bp_{i}(v,c)$ we crucially excluded the possibility 
	of $u$ changing $\bp_{i}(v,c)$ by coloring itself). Now note that since $G$ is triangle-free, for any vertex $u \in B_{i}(v,c)$, $B_{i}(u,c) \cap B_{i}(v,c)$ is disjoint (otherwise we find a triangle with $u,v$ and the intersecting vertex). 
	This shows the correctness of Eq~\eqref{eq:expect-kept}. By expanding the RHS of~\eqref{eq:expect-kept},
	\begin{align*}
		\expect{\bp_{i}(v)} &= \sum_{c \in A_i(v)} \pkeep_i \cdot \sum_{u \in B_{i}(v,c)}\Pr\paren{\text{$c \in \hA_{i}(u)$}} \tag{by Eq~\eqref{eq:keep} for the first term and by definition for second one} \\
		&= \sum_{c \in A_i(v)} \pkeep_i \cdot b_{i}(v,c) \cdot \pkeep_i \tag{again by Eq~\eqref{eq:keep}} \\
		&= \pkeep_i^2 \cdot a_i(v) \cdot \bb_i(v). \tag{by definition of $\bb_i(v)$}
	\end{align*}
	We now prove a concentration bound for $\bp_{i}(v)$. For any $c \in A_i(v)$ define the random variable $X_c = \bp_{i}(v,c)$ if $c$ is kept in $\hA_{i}(v)$ as well and $X_c = 0$ otherwise. 
	Additionally, define $X := \sum_{c \in A_i(v)} X_c$. By Eq~\eqref{eq:expect-kept}, $X = \bp_{i}(v)$ and by the discussion after this equation plus the fact that the choices of $\bp_{i}(v,c)$ and $\bp_{i}(v,c')$ for colors $c \neq c'$
	are independent, we have that $X_c$'s are independent. Moreover, each $X_c \leq b_i(v,c) \leq \bmax_i$ by definition. As such, by Chernoff bound (Proposition~\ref{prop:chernoff} and since $\pkeep_i=\Omega(1)$), 
	\begin{align*}
		\Pr\paren{X - \expect{X} > \eps \cdot \pkeep^2_i \cdot a_i(v) \cdot \bmax_i} &\leq \exp\Paren{-\Theta(1) \cdot \frac{\eps^2 \cdot a_i(v)^2 \cdot \paren{\bmax_i}^2}{a_i(v) \cdot \paren{\bmax_i}^{2}}} \\
		&\leq \exp\Paren{-\Theta(1) \cdot \eps^2 \cdot a_i(v)}  \\
		&\leq \exp\paren{-\Theta(1) \cdot d^{4\delta/5}} \tag{as already calculated in the proof of Lemma~\ref{lem:tri-conc-a}}
	\end{align*}
	Since $X = \bp_{i}(v)$ and by the value of $\expect{X}$ calculated earlier, this finalizes the proof. 
	\Qed{Claim~\ref{clm:tri-conc-b1}}
	
	\end{proof}

	\begin{claim}\label{clm:tri-conc-b2}
		$\Pr\paren{\hb_{i}(v,c) > \pcolor_i \cdot \bp_{i}(v,c) + 2\eps \cdot  \pcolor_i \cdot \bmax_i} < \exp\paren{-\Theta(d^{4\delta/5})}.$
	\end{claim}
	\begin{proof}
		Consider the complement of the event in Lemma~\ref{lem:tri-conc-a} for all vertices $u \in B_{i}(v,c)$. Note that the choice of colors in $\hA_{i}(u)$ is entirely independent of
		the randomness of vertex $u$ itself. Similarly, let $\Bp_{i}(v,c)$ denote the set of vertices $u \in B_{i}(v,c)$ that are {counted in} $\bp_{i}(v,c)$ (defined at the beginning of the proof of the lemma).
		Note that again for each vertex $u \in B_{i}(v,c)$, the choice
		whether $u$ joins $\Bp_{i}(v,c)$ or not is independent of randomness of $u$ itself (this is the key difference between $\Bp_i$ and $\hB_i$). 
		In the following, we condition on the choice of $\hA_i(u)$ for vertices $u \in B_i(v,c)$ as well as the choice of $\Bp_{i}(v,c)$; by union bound over at most $\poly(d)$ vertices in the constant-hop neighborhood of $v$, we have that the complement of 
		the event in both Lemma~\ref{lem:tri-conc-a} and Claim~\ref{clm:tri-conc-b1} happens with sufficiently probability for the assertion of the claim.


		Now consider each vertex $u \in \Bp_{i}(v,c)$. For $u$ to join $\hB_{i}(v,c)$ as well (and hence counted in $\hb_{i}(v,c)$), $u$ should not be colored in this iteration. 
		This is equivalent to the event that no color in $\hA_i(u)$ is assigned to $u$. This choice is only a function of randomness of $u$. As such, 		
		\begin{align*}
			\Pr\paren{\text{$u \in \Bp_i(v,c)$ joins $\hB_i(v,c)$}} &= \Pr\paren{\text{$u$ is not colored in iteration $i$}} \\
			&= \prod_{c \in \hA_{i}(u)}(1-\Pr\paren{\text{$c$ is assigned to $u$}}) \\ 
			&= (1-\frac{1}{2\ln{d} \cdot \alphaid_i})^{\ha_i(u)} \tag{by the choice of $p_i(u)$ in $\wasteful$}\\ 
			&\leq (1-\frac{1}{2\ln{d} \cdot \alphaid_i})^{(1-\eps) \cdot \pkeep_i \cdot a_i(v)} \tag{by Lemma~\ref{lem:tri-conc-a}}\\
			&\leq (1-\frac{1}{2\ln{d} \cdot \alphaid_i})^{(1-\eps)^2 \cdot \pkeep_i \cdot \alphaid_i/2} \tag{by Eq~\eqref{eq:min-max}}\\
			&\leq \pcolor_i \cdot (1-\frac{1}{2\ln{d} \cdot \alphaid_i})^{-2\eps \cdot \pkeep_i \cdot \alphaid_i/2} \tag{by definition of $\pcolor_i$ in Eq~\eqref{eq:starting} and since $(1-\eps)^2 \leq 1-2\eps$}\\
			&\leq \pcolor_i \cdot \exp\paren{\frac{\eps \cdot \pkeep_i \cdot \alphaid_i}{2\ln{d} \cdot \alphaid_i}} \\
			&\leq \pcolor_i \cdot (1+\eps) \tag{as $\pkeep_i = \Theta(1) \ll \ln{d}$}.
		\end{align*}
		This implies that $\expect{\hb_i(v,c) \mid \Bp_i(v,c)} \leq \pcolor_i \cdot (1+\eps) \cdot \bp_i(v,c)$. 
		Moreover, as stated earlier, at this point all choices of whether $u \in \Bp_{i}(v,c)$ also belongs to $\hB_{i}(v,c)$ depend on the randomness of $u$ itself and are thus independent across different $u \in \Bp_{i}(v,c)$.
		As such, $\hb_i(v,c)$ is a sum of $\bp_i(v,c)$ $\set{0,1}$-independent random variables and hence by Chernoff bound (Proposition~\ref{prop:chernoff} and since $\pcolor_i = \Omega(1)$): 
		\begin{align*}
			\Pr\paren{\hb_{i}(v,c) > \pcolor_i \cdot (1+\eps) \cdot \bp_i(v,c) + \eps \cdot \bmax_i} &\leq \exp\paren{-\Theta(1) \cdot \eps^2 \cdot \bmax_i} \\
			&\leq \exp\paren{-\Theta(1) \cdot \eps^2 \cdot \amin_i} \tag{as iteration $i < \istar$ has $\bmax_i \geq \amin_i/100$} \\
			&\leq \exp\paren{-\Theta(1) \cdot d^{4\delta/5}} \tag{ by the choice of $\eps$ as already calculated in Lemma~\ref{lem:tri-conc-a}}
		\end{align*}
		This concludes the proof. \Qed{Claim~\ref{clm:tri-conc-b1}}
		
	\end{proof}

	We are now ready to finalize the proof of Lemma~\ref{lem:tri-conc-b}. We condition on the complements of the events in Claims~\ref{clm:tri-conc-b1} and~\ref{clm:tri-conc-b2} and 
	by union bound (over $\poly(d)$ vertices in the constant-hop neighborhood of $v$), this happens with sufficiently high probability for the proof. We now have,
	\begin{align*}
		\sum_{c \in \hA_i(v)} \hb_{i}(v,c) &\leq \sum_{c \in \hA_{i}(v)} \Paren{\pcolor_i \cdot \bp_{i}(v,c) + 2\eps \cdot \pcolor_i \cdot \bmax_i} \tag{by Claim~\ref{clm:tri-conc-b2}} \\
		&\leq \paren{\pcolor_i \cdot \sum_{c \in \hA_{i}(v)} \bp_{i}(v,c)} +  2\eps \cdot \pcolor_i \cdot a_i(v) \cdot \bmax_i  \tag{as $\ha_i(v) \leq a_i(v)$}\\
		&= \pcolor_i \cdot \bp_{i}(v)  +  2\eps \cdot \pcolor_i \cdot a_i(v) \cdot \bmax_i \tag{by definition of $\bp_i(v)$} \\
		&\leq \pcolor_i \cdot \Paren{\pkeep^2_i \cdot a_i(v) \cdot \bb_i(v) + \eps \cdot \pkeep^2_i \cdot a_i(v) \cdot \bmax_i} + 2\eps \cdot \pcolor_i \cdot a_i(v) \cdot \bmax_i  \tag{by Claim~\ref{clm:tri-conc-b1}} \\
		&\leq \pcolor_i \cdot \pkeep^2_i \cdot a_i(v) \cdot \bb_i(v) +  3\eps \cdot \pcolor_i \cdot a_i(v) \cdot \bmax_i \tag{as $\pkeep_i < 1$}.
	\end{align*}
	Let us now further condition on the event of Lemma~\ref{lem:tri-conc-a}. We will thus have, 
	\begin{align*}
		\hb_i(v) &= \frac{1}{\ha_i(v)} \cdot \sum_{c \in \hA_i(v)} \hb_{i}(v,c) \\
		&\leq \frac{1}{(1-\eps) \cdot \pkeep_i \cdot a_i(v)} \cdot \Paren{\pcolor_i \cdot \pkeep_i^2 \cdot a_i(v) \cdot \bb_i(v) + 3\eps \cdot \pcolor_i \cdot a_i(v) \cdot \bmax_i} \\
		&= \frac{\pcolor_i \cdot \pkeep_i^2 \cdot a_i(v) \cdot \bb_i(v)}{(1-\eps) \cdot \pkeep_i \cdot a_i(v)} + \frac{3\eps \cdot \pcolor_i \cdot a_i(v) \cdot \bmax_i}{(1-\eps) \cdot \pkeep_i \cdot a_i(v)} \\
		&\leq {\pcolor_i \cdot \pkeep_i  \cdot \bb_i(v)} \cdot (1+2\eps) + \frac{4\eps \cdot \pcolor_i \cdot \bmax_i}{\pkeep_i} \\
		&\leq {\pcolor_i \cdot \pkeep_i  \cdot \bb_i(v)} \cdot (1+2\eps) + 5\eps \cdot \pcolor_i \cdot \bmax_i \tag{as calculated in Lemma~\ref{lem:ideal-list}, for $i < \istar$ $\pkeep_i \geq (1-\Theta(1)/\ln{d})$} \\
		&\leq {\pcolor_i \cdot \pkeep_i  \cdot \bb_i(v)} + 8\eps \cdot \pcolor_i \cdot \pkeep_i \cdot \bmax_i \tag{again by the lower bound on $\pkeep_i$}, \\
	\end{align*}
	concluding the proof. \Qed{Lemma~\ref{lem:tri-conc-b}}

\end{proof}

We now combine the above lemmas to prove the following bound on $\ssmax_i$. 
\begin{lemma}\label{lem:tri-conc-ss}
	For every $v \in G_{i+1}$, assuming the events in Lemmas~\ref{lem:tri-conc-a} and~\ref{lem:tri-conc-b}:
	\begin{align*}
		\ss_{i+1}(v)  \leq \pcolor_i \cdot \pkeep_i \cdot \ss_i(v) + 18\eps \cdot  \betaid_{i+1}. 
	\end{align*}
\end{lemma}
\begin{proof}
	Let us define two new parameters for the purpose of this proof (similar to $\lambda_i$ and $\ss_i$): 
	\begin{itemize}
	\item $\hlambda_i(v) := \min\set{1,\ha_i(v)/\alphaid_{i+1}}$: the ratio of size of $\hA_i(v)$ to the ideal size $\alphaid_{i+1}$; 
	\item $\hsss_i(v) := \hlambda_i(v) \cdot \hb_i(v) + (1-\hlambda_i(v)) \cdot 2\betaid_{i+1}$. 
	\end{itemize}
	Firstly, as $\ss_{i+1}(v)$ is obtained from $\hsss_i(v)$ by changing the contribution of any color in $\hA_i(v) \setminus A_{i+1}(v)$ from something larger than $2\betaid_{i+1}$ down to $2\betaid_{i+1}$, 
	we have $\ss_{i+1}(v) \leq \hsss_i(v)$. We use this in the following claim. 
	\begin{claim}\label{clm:tri-ss1}
		$\ss_{i+1}(v) \leq \lambda_i(v) \cdot \hb_i(v) + (1-\lambda_i(v)) \cdot 2\betaid_{i+1} + 2\eps \cdot  \betaid_{i+1}$. 
	\end{claim}
	\begin{proof}
	We have, 
	\begin{align*}
	\hlambda_i(v) &= \frac{\ha_i(v)}{\alphaid_{i+1}} \geq \frac{(1-\eps) \cdot \pkeep_i \cdot a_i(v)}{\pkeep_i \cdot \alphaid_i}  \geq (1-\eps) \cdot \lambda_i(v).
	\tag{by Lemma~\ref{lem:tri-conc-a} in the nominator and definition of $\alphaid_i$ in Eq~\eqref{eq:starting} for the denominator} 
	\end{align*}
	Moreover, 
	\begin{align*}
	\hb_i(v) &\leq \pcolor_i  \cdot \pkeep_i \cdot \bb_i(v) + 8\eps \cdot \pcolor_i \cdot \pkeep_i \cdot \bmax_i \tag{by Lemma~\ref{lem:tri-conc-b}} \\
	 &\leq (1+\eps)\betaid_i + 16\eps \cdot \betaid_i  \tag{by definition, $\bb_i(v) \leq \ss_i(v)$ and by Eq~\eqref{eq:min-max}, $\ss_i(v) \leq (1+\eps)\betaid_i$ and $\bmax_i \leq 2\betaid_i$} \\
	&<2\betaid_i \tag{for $\eps$ sufficiently small -- $\eps < 1/100$ certainly suffices}.
	\end{align*} 
	Consequently, 
	\begin{align*}
		\ss_{i+1}(v) &\leq \hsss_i(v) = \hlambda_i(v) \cdot \hb_i(v) + (1-\hlambda_i(v)) \cdot 2\betaid_{i+1} \\
		&\leq (\lambda_i(v)-\eps\lambda_i(v)) \cdot \hb_i(v) + (1-\lambda_i(v)+\eps\lambda_i(v)) \cdot 2\betaid_{i+1} \tag{by the two equations above}\\
		&\leq \lambda_i(v)\cdot \hb_i(v) + (1-\lambda_i(v)) \cdot {2\betaid_{i+1}} + 2\eps\cdot \betaid_{i+1}.  \tag{as $\lambda_i(v) \leq 1$} \Qed{Claim~\ref{clm:tri-ss1}}
	\end{align*}

	\end{proof}
	Finally, by Claim~\ref{clm:tri-ss1},
	\begin{align*}
		\ss_{i+1}(v)&\leq \lambda_i(v) \cdot \hb_i(v) + (1-\lambda_i(v)) \cdot 2\betaid_{i+1} + 2\eps \cdot  \betaid_{i+1} \\
		&\leq \lambda_i(v) \cdot \Paren{\pcolor_i \cdot \pkeep_i \cdot \bb_i(v) +  8\eps \cdot \pcolor_i \cdot \pkeep_i \cdot \bmax_i} + (1-\lambda_i(v)) \cdot 2\betaid_{i+1} + 2\eps \cdot  \betaid_{i+1}   \tag{by Lemma~\ref{lem:tri-conc-b}} \\
		&\leq \lambda_i(v) \cdot \Paren{\pcolor_i \cdot \pkeep_i \cdot \bb_i(v) +  16\eps \betaid_{i+1}} + (1-\lambda_i(v)) \cdot 2\betaid_{i+1} + 2\eps \cdot  \betaid_{i+1}  
		\tag{by Eq~\eqref{eq:min-max}, $\bmax_i \leq 2\betaid_i$ and by definition of $\betaid_{i+1}$} \\
		&= \pcolor_i \cdot \pkeep_i \Paren{\lambda_i(v) \cdot \bb_i(v) + (1-\lambda_i(v)) \cdot 2\betaid_{i}} + 18\eps  \cdot  \betaid_{i+1}   \tag{by definition of $\betaid_{i+1} = \pcolor_i \cdot \pkeep_i \cdot \betaid_i$ in Eq~\eqref{eq:starting}} \\
		&= \pcolor_i \cdot \pkeep_i \cdot \ss_i(v) +  18\eps  \cdot  \betaid_{i+1}  \tag{by definition of $\ss_i(v)$}. 
	\end{align*}
	This finishes the proof of the lemma. \Qed{Lemma~\ref{lem:tri-conc-ss}}
	
\end{proof}
Lemma~\ref{lem:tri-main} now follows easily from this as follows. 
\begin{proof}[Proof of Lemma~\ref{lem:tri-main}]
	For any vertex $v$ and color $c \in A_i(v)$, the events of Lemmas~\ref{lem:tri-conc-a} and~\ref{lem:tri-conc-b} are only a function of random choices in the constant-hop neighborhood
	of $v$. Hence, each such event depends on at most $\poly(d)$ other events. As such, by the bounds on the probability of success in these two lemmas and Lov\'asz Local Lemma (Proposition~\ref{prop:lll}), 
	we obtain that with positive probability none of these events happen. We can thus apply Lemma~\ref{lem:tri-conc-ss} to any vertex $v \in G_{i+1}$ and hence obtain that: 
	\begin{align*}
		\ssmax_{i+1} &\leq \pcolor_i \cdot \pkeep_i \cdot \ssmax_i + 18\eps \cdot  \betaid_{i+1} \\
		&\leq \pcolor_i \cdot \pkeep_i \cdot (1+\eps) \cdot \betaid_i + 18\eps \cdot  \betaid_{i+1} \tag{by Eq~\eqref{eq:min-max}} \\
		&=  (1+\eps) \cdot \betaid_{i+1} + 18\eps \cdot  \betaid_{i+1} \tag{by definition of $\betaid_{i+1}$ in Eq~\eqref{eq:starting}} \\
		&= \betaid_{i+1} + 19 \eps \cdot \betaid_{i+1},
	\end{align*}
	finishing the proof. 
	\Qed{Lemma~\ref{lem:tri-main}}
	
\end{proof}

\subsection{Concluding the Proof of Proposition~\ref{prop:lc-triangle}} 

We now show that by repeatedly applying Lemma~\ref{lem:tri-main}, we can reach the desired state whereby size of the lists for remaining vertices is sufficiently larger than 
their $c$-degrees and thus apply Proposition~\ref{prop:lc-const} to obtain the coloring of all remaining vertices in one shot.  

\begin{proof}[Proof of Proposition~\ref{prop:lc-triangle}]
	We run the $\wasteful$  procedure over iterations $i \leq \istar$ (recall the definition of $\istar$ from Lemma~\ref{lem:ideal-list}). Let us define the following parameter $\eps_i$ recursively: 
	\begin{align*}
		\eps_1 = d^{-\delta/20} \qquad \textnormal{and} \qquad \eps_{i+1} = (1+19\eps_{i}). 
	\end{align*}
	It is easy to see that for $i \leq \istar$, all $\eps_i > d^{-\delta/10}$ and since $\istar = O(\log^{2}{d})$, we also have $\eps_i  = o(1)$. As such, we can repeatedly apply Lemma~\ref{lem:tri-main} with parameters $\eps_i$ 
	and $\ssmax_i \leq (1+\eps_i) \cdot \betaid_i$ to with positive probability obtain $\ssmax_{i+1} \leq (1+19\eps_i) \cdot \betaid_{i+1} = (1+\eps_{i+1}) \cdot \betaid_{i+1}$. At iteration $\istar$, by Lemma~\ref{lem:ideal-list}, 
	we have that $\betaid_{\istar} < \alphaid_{\istar}/100$. At this point, by Eq~\eqref{eq:min-max}, we have, 
	\begin{align*}
		\amin_{\istar} \geq \alphaid_{\istar}/2 \cdot (1-\eps_{\istar}) > \alpha_{\istar}/3 > 30 \betaid_{\istar} \geq 15 \cdot \bmax_{\istar}.  
	\end{align*}
	We can now simply apply Proposition~\ref{prop:lc-const} and obtain a proper coloring of $G$.  \Qed{Proposition~\ref{prop:lc-triangle}}

\end{proof}

%% file: background-ACK19.tex

\section{Background on the Palette Sparsification Theorem of~\cite{AssadiCK19}}\label{sec:background-ACK19}

Our main results are closely related to the palette sparsification theorem of Assadi, Chen, and Khanna~\cite{AssadiCK19} and our Result~\ref{res:deg+1}  
involves using components of this result in a non-black-box way. As such, we give a brief high level overview of this result here, and state 
the main properties that we use in our proofs. The palette sparsification theorem of~\cite{AssadiCK19} is as follows. 

\begin{proposition}[Palette sparsification theorem of~\cite{AssadiCK19}]\label{prop:ACK19}
In any graph $G(V,E)$ with $n$ vertices and maximum degree $\Delta$, if we sample $\Theta(\log{n})$ colors $L(v)$ for each vertex $v \in V$ independently and uniformly at random from colors $\set{1,\ldots,\Delta+1}$, 
then $G$ can be properly colored from the sampled lists $L(v)$ for $v \in V$ with high probability. 
\end{proposition}

The proof of this result is carried out in three main steps in~\cite{AssadiCK19}: $(i)$ introducing a proper decomposition of every graph $G$ into \emph{sparse} and \emph{dense} vertices, 
$(ii)$ proving that sampled colors are sufficient for coloring sparse vertices, and $(iii)$ proving that after fixing the colors for sparse vertices (even adversarially), the sampled colors are sufficient for coloring the dense vertices. 
We shall note the idea of decomposing the graph into sparse and dense parts and analyzing each part separately in the context of $(\Delta+1)$ coloring 
has a long history in the graph theory literature starting with the pioneering work of Reed~\cite{Reed98}; see, e.g.~\cite{MolloyR98,Reed99b,MolloyR10,MolloyR14}. 

\smallskip
We now briefly describe each of the three components of the proof of Proposition~\ref{prop:ACK19} in~\cite{AssadiCK19}.

\paragraph{Graph Decomposition.} For a parameter $\eps \in (0,1)$, we say a vertex $v \in V$ in a graph $G(V,E)$ is \emph{$(\Delta,\eps)$-sparse} iff there are at least $\eps^2\cdot{{\Delta}\choose{2}}$ \emph{non}-edges in the neighborhood of $v$ (when $\deg(v) < \Delta$, we first append the neighborhood of $v$ with $\Delta-\deg(v)$ dummy vertices connected only to $v$). 
We use $\Vsparse_{\eps}$ to denote the set of $(\Delta,\eps)$-sparse vertices. The following decomposition proven in~\cite{AssadiCK19} is an extension of the HSS decomposition of~\cite{HarrisSS16} (itself based on anearlier decomposition
of~\cite{Reed98}). 

\begin{lemma}[Extended HSS Decomposition~\cite{AssadiCK19}]\label{lem:extended-HSS-decomposition}
	For any parameter $\eps \in [0,1)$, any graph $G(V,E)$ can be partitioned into a collection of vertices $V:= \Vsparse_\star \sqcup C_1 \sqcup \ldots \sqcup C_k$ such that: 
	\begin{enumerate}[leftmargin=15pt]
		\item\label{HSS-p1} $\Vsparse_\star \subseteq \Vsparse_{\eps}$, i.e., any vertex in $\Vsparse_\star$ is $(\Delta,\eps)$-sparse. 
		\item\label{HSS-p2} For any $i \in [k]$, $C_i$ has the following properties (we refer to $C_i$ as an \emph{\underline{$(\Delta,\eps)$-almost-clique}}): 
		\begin{enumerate}
			\item\label{HSS-p2a}\label{p1} $(1-\eps)\Delta \leq \card{C_i} \leq (1+6\eps)\Delta$. 
			\item\label{HSS-p2b} Any $v \in C_i$ has at most $7\eps\Delta$ neighbors outside of $C_i$.
			\item\label{HSS-p2c} Any $v \in C_i$ has at most $6\eps\Delta$ non-neighbors inside of $C_i$. 
		\end{enumerate}
	\end{enumerate}
\end{lemma}

The approach in~\cite{AssadiCK19} is then as follows. The authors first pick some small enough \emph{constant} $\eps > 0$ (say $\eps = 10^{-4}$ for concreteness). Let 
 $\Vsparse_\star \sqcup C_1 \sqcup \ldots \sqcup C_k$ be a decomposition of the given graph $G(V,E)$ in Lemma~\ref{lem:extended-HSS-decomposition} for this parameter $\eps$. 
 The rest is to color $\Vsparse_\star$ and $C_1 \cup \ldots \cup C_k$ from the sampled colors in lists $L$ using different arguments. 
 
\paragraph{Coloring Sparse Vertices.} The first (and the easy) part of the argument is to color sparse vertices, ignoring entirely all the dense vertices. 
This is done using the following lemma. 

\begin{lemma}[\!\!\cite{AssadiCK19}]\label{lem:sparse-color}
	Suppose for every vertex $v \in \Vsparse_{\eps}$, we sample a set $L(v)$ of $\Theta(\eps^{-2} \cdot \log{n})$ colors 
	independently and uniformly at random from $\set{1,\ldots,\Delta+1}$. Then, with high probability, the induced subgraph $G[\Vsparse_\eps]$ can be properly colored from 
	the sampled lists $L(v)$ for $v \in \Vsparse_{\eps}$. 
\end{lemma}

This lemma is proven in~\cite{AssadiCK19} by ``simulating'' a simple greedy coloring procedure for coloring $G$ using by-now standard ideas from~\cite{ElkinPS15,HarrisSS16,ChangLP18} (which are all rooted in~\cite{MolloyR97} that proved that chromatic
number of any graph where all vertices are $\eps$-sparse is at most $(1-\Theta(\eps)) \cdot \Delta$).  
Equipped with this lemma, one can then color all vertices in $\Vsparse_\star \subseteq \Vsparse_\eps$ in the decomposition using the sampled lists in the palette sparsification theorem (recall that $\eps$ is a sufficiently small constant).

\paragraph{Coloring Almost-Cliques.} The second (and the main) part of the argument in~\cite{AssadiCK19} is to color almost-cliques, which is done using the following lemma. 

For a vertex $v$ in a $(\Delta,\eps)$-almost-clique $C$, we define the out-degree of $v$ in $C$, denoted by $\outdeg_C(v)$ as the number of neighbors of $v$ in $G$ that are outside $C$. Recall that by definition of 
a $(\Delta,\eps)$-almost-clique, $\outdeg(v) \leq 7\eps\Delta$. 
\begin{lemma}[\!\cite{AssadiCK19}]\label{lem:almost-clique-color}
	Let $C$ be a $(\Delta,\eps)$-almost-clique in $G$. Suppose for every  $v \in C$, we \emph{\underline{adversarially}} pick a set $\barS(v)$ of size $\leq \outdeg_C(v)$ colors from $\set{1,\ldots,\Delta+1}$. 
	Now, if for every vertex $v \in V$, we sample a set $L(v)$ of $\Theta(\eps^{-1} \cdot \log{n})$ colors independently from $\set{1,\ldots,\Delta+1}$, then, with high probability, the induced subgraph $G[C]$ can be properly colored
	from the lists $L(v) \setminus \barS(v)$ for $v \in C$. 
\end{lemma}

Lemma~\ref{lem:almost-clique-color} is the heart of the argument in~\cite{AssadiCK19}. It states that the no matter how we color the remainder of the graph, there is ``enough'' randomness in the lists of almost-cliques so that we can (with high probability) find a 
coloring of each almost-clique to extend to the previous coloring. As such, we can simply go over the almost-cliques one by one and color each almost-clique $C$ using Lemma~\ref{lem:almost-clique-color} as follows: As every vertex $v \in C$ has
at most $\outdeg_C(v) \leq 7\eps\Delta$ neighbors outside $C$ (by definition of $(\Delta,\eps)$-almost-cliques in Lemma~\ref{lem:extended-HSS-decomposition}), we pick the colors used for these neighbors in the set $\barS(v)$ and then invoke Lemma~\ref{lem:almost-clique-color} 
to color $C$ with high probability. We iterate like this until we find a proper coloring of $G$. This concludes the high level approach of the proof in~\cite{AssadiCK19}. 


%% file: omitted-proofs.tex

\section{Proofs of Basic Random Graph Theory Results}\label{app:random-graph-theory}

\begin{lemma*}[Restatement of Lemma~\ref{lem:random-triangle}]
	For $G \sim \FG_{n,p}$, $\expect{t(G)} \leq (np)^3$, and w.h.p. 
	\[t(G) \leq (1+o(1))\expect{t(G)}.\]
\end{lemma*}
\begin{proof}
	$\expect{t(G)} \leq \sum_{u,v,w} \Pr\paren{\text{$(u,v),(v,w),(w,u)$ belongs to $G$}} = {{n}\choose{3}} \cdot p^3 \leq (np)^3$. 
The high probability result can be proven in several ways and is
well known, see, for example, \cite{HMS2019}.
\end{proof}

\begin{lemma*}[Restatement of Lemma~\ref{lem:g5-random}]
	For $G \sim \FG_{n,p}$, $\expect{\alpha(G)} \leq  
\frac{3 \cdot \ln{(np)}}{p}$, and w.h.p. 
\[
\alpha(G)  \leq  \frac{3 \cdot \ln{(np)}}{p}. \]
\end{lemma*}
\begin{proof}
	Fix any set $S$ of $k := \frac{3\cdot \ln{(np)}}{p}$ vertices in $G$. We have,
	\begin{align*}
		\Pr\paren{\text{$S$ is an independent set}} = (1-p)^{{k}\choose{2}} \leq \exp\paren{-p \cdot {{k}\choose{2}}} \leq \exp\paren{-\frac{4}{p} \cdot \ln^2{(np)}}.
	\end{align*}
	On the other hand, the total number of choices for $S$ is:
	\begin{align*}
		{\# \text{of $k$-subsets of $V$}} = {{n}\choose{k}} \leq \paren{\frac{e \cdot n}{k}}^{k} \leq \exp\paren{k\cdot\ln{(\frac{n}{k})} + k} \leq \exp\paren{\frac{3}{p} \cdot\ln^2{(np)}}. 
	\end{align*}
	Taking a union bound over all $k$-subsets $S$, we obtain that w.h.p, none of the subsets can be an independent set. This implies  $\alpha(G) < k$ with high probability and $\alpha(G) \leq k$ in expectation. 
\end{proof}
\noindent
We note that the constant $3$ above can be easily reduced to $2+o(1)$ but this is not needed here. 

\begin{lemma*}[Restatement of Lemma~\ref{lem:g5-max-degree}]
	For $G \sim \FG_{n,p}$, w.h.p. $\Delta(G) \leq 2np$. 
\end{lemma*}
\begin{proof}
	A direct application of Chernoff bound and union bound. 
\end{proof}